%% file: wellS_full.tex
\documentclass{LMCS}

\def\doi{9(3:2)2013}
\lmcsheading%
{\doi}
{1--51}
{}
{}
{Jul.~31, 2012}
{Aug.~\phantom09, 2013}
{}

\usepackage{enumerate}
\usepackage{hyperref}
\usepackage{url}
\usepackage{amsmath}
\usepackage{amssymb}
\usepackage{mathrsfs,euscript}
\usepackage{bbm}

\usepackage{lipsum}

\usepackage[all]{xy}
\SelectTips{cm}{}
\newdir{x}{{}*!/-10pt/\dir^{(}}
\newdir{y}{{}*!/5pt/\dir{>}}

\def\refeq#1{(\ref{#1})}

\input ldiagxy

\renewcommand{\tilde}[1]{\bigcirc #1}
\newcommand{\KVec}{\mathsf{Vec}_{\scriptstyle K}}
\newcommand{\takeout}[1]{\relax}

\newcommand{\Set}{\ensuremath{\St}}
\newcommand{\A}{\ensuremath{\mathscr{A}}}

\newcommand{\PP}{\ensuremath{\mathscr{P}}}
\newcommand{\pow}{\PP}
\newcommand{\set}[1]{\{ #1 \}}

\newcommand{\M}{\ensuremath{\mathscr{M}}}
\newcommand{\nn}{\ensuremath{\mathbb{N}}}
\newcommand{\Sig}{\ensuremath{\Sigma}}
\newcommand{\tec}{{\cdot}}
\newcommand{\nsi}[1]{\xrightarrow{\rule{1mm}{0mm}#1\rule{1mm}{0mm}}}

\newcommand{\Bool}{\ensuremath{\mathbf{ Bool}}}

\DeclareMathOperator{\wellp}{\mathsf{wp}}
\DeclareMathOperator{\St}{\textnormal{\textbf{Set}}}
\DeclareMathOperator{\id}{id}

\DeclareMathOperator{\Coalg}{\textnormal{\textbf{Coalg}}}

\DeclareMathOperator{\Ord}{\textnormal{\textbf{Ord}}}

\DeclareMathOperator{\Sub}{\textnormal{\textbf{Sub}}}

\DeclareMathOperator{\Gra}{\textnormal{\textbf{Gra}}}

\DeclareMathOperator*{\colim}{colim}

\DeclareMathOperator{\card}{card}
\DeclareMathOperator{\curry}{curry}
\DeclareMathOperator{\out}{out}
\DeclareMathOperator{\tail}{tail}
\DeclareMathOperator{\head}{head}
\DeclareMathOperator{\arit}{ar}

\theoremstyle{plain}
\newtheorem{theorem}{Theorem}[section]
\newtheorem{proposition}[theorem]{Proposition}

\newtheorem{corollary}[theorem]{Corollary}

\newtheorem{lemma}[theorem]{Lemma}
\theoremstyle{definition}
\newtheorem{definition}[theorem]{Definition}
\newtheorem{assumption}[theorem]{Assumption}
\newtheorem{example}[theorem]{Example}
\newtheorem{examples}[theorem]{Examples}

\newtheorem{notation}[theorem]{Notation}
\newtheorem{terminology}[theorem]{Terminology}
\newtheorem{observation}[theorem]{Observation}

\newtheorem{remark}[theorem]{Remark}

\newtheorem{construction}[theorem]{Construction}

\newtheorem*{remb}{Remark}
\numberwithin{equation}{section}

\begin{document}
%
%
%
\title{Well-pointed Coalgebras}

\author[J.~Ad\'amek]{Ji\v r\'\i\ Ad\'amek\rsuper a}
\address{{\lsuper a}Institut f\"ur Theoretische Informatik, Technische Universit\"at Braunschweig, Germany}
\email{adamek@iti.cs.tu-bs.de}

\author[S.~Milius]{Stefan Milius\rsuper b}
\address{{\lsuper b}Lehrstuhl f\"ur Theoretische Informatik, Friedrich-Alexander
  Universit\"at Erlangen-N\"urnberg, Germany}
\email{mail@stefan-milius.eu}

\author[L.~S.~Moss]{Lawrence S.~Moss\rsuper c}
\address{{\lsuper c}Department of Mathematics, Indiana University, Bloomington, IN, USA}
\email{lsm@cs.indiana.edu}

\author[L.~Sousa]{Lurdes Sousa\rsuper d}
\address{{\lsuper d}Polytechnic Institute of Viseu, Portugal \& Centre for Mathematics of the University of Coimbra, Portugal}
\email{sousa@mat.estv.ipv.pt}
\thanks{{\lsuper d}Financial support by CMUC/FCT (Portugal) and the FCT Grant PTDC/MAT/120222/2010 is acknowledged by the last author.}

\ACMCCS{[{\bf Theory of computation}]:Models of computation; Formal languages and automata theory}

\begin{abstract}
For endofunctors of varieties preserving intersections, a new description of the
final coalgebra and the initial algebra is presented: the former
consists of all well-pointed coalgebras. These are the pointed
coalgebras having no proper subobject and no proper quotient. 
The
initial algebra consists of all well-pointed coalgebras that are
well-founded in the sense of Osius~\cite{O} and Taylor~\cite{Ta2}.  And initial algebras are precisely the final well-founded coalgebras.
Finally, the initial
iterative algebra consists of all finite well-pointed
coalgebras.  Numerous examples are discussed e.g. automata, graphs,
and labeled transition systems.
\end{abstract}

\maketitle


\section{Introduction}
\label{intr}


Initial algebras are known to be of primary interest in denotational
semantics, where abstract data types are often presented as initial
algebras for an endofunctor $H$ expressing the type of the
constructor operations of the data type. For example, finite binary trees are
the initial algebra for the functor $HX = X \times X +
1$ on sets. Analogously, final coalgebras for an endofunctor $H$ play an important role in the theory of
systems developed by Rutten~\cite{R1}: $H$ expresses the system type,
i.\,e., which kind of one-step reactions states can exhibit (input, output, state
transitions etc.), and the coalgebras for $H$ are precisely systems
with a set of states having reactions of type $H$. 
The elements of a final coalgebra
represent the behavior of all states in all systems of type $H$, and the unique
homomorphism from a system into the final one assigns to every state
its behavior. For example, deterministic automata with input alphabet
$I$ are coalgebras for $HX = X^I \times \{0,1\}$, and the final coalgebra is
the set of all languages on $I$.  

In this paper a unified description is presented for~(a) initial algebras, (b) final
coalgebras and (c)~initial iterative algebras (in the automata example
this is the set of all regular languages on $I$). We also demonstrate
that this new description provides a unifying view of a
number of important examples.   We first work with set functors~$H$
preserving intersections. This is an extremely mild requirement that
most ``everyday'' set functors satisfy, see Example \ref{E-coc}. We prove that the final
coalgebra for $H$ can then be described as the set of all
\textit{well-pointed coalgebras}, i.e., pointed coalgebras not having
any proper subobject and also not having any proper quotient. 
Moreover, the initial algebra
can be described as the set of all well-pointed coalgebras which are
well-founded in the sense of Osius~\cite{O} and
Taylor~\cite{Ta1,Ta2}.  We then extend these results to all endofunctors of varieties preserving intersections.

Before we mention the definition of well-founded coalgebra, recall that
the notion of well-foundedness of relations $R \subseteq X
\times X$
has several alternative forms:
\begin{enumerate}[(1)]
\item  No proper subset $Y$ of $X$ has  the property
that if all $R$-successors of a given point $x\in X$  lie in $Y$, then $x\in Y$
as well.   
\item There is no infinite path
$x_0 R x_1 R  x_2 R \cdots$.
\item There is a map  $\mathsf{rk}$ from $X$ to ordinals such that
$\mathsf{rk}(x) > \mathsf{rk}(y)$ whenever $x R y$.
\end{enumerate}
For sets and relations as usual, these  conditions are equivalent.
The first of these is an \emph{induction principle}, and this
is closest to what we are calling well-foundedness in this paper,
following Taylor.   
The equivalence of the first and the second requires
Dependent Choice, 
a weak form of the Axiom of Choice; in any case, our 
work in this area does not use this at all.   The last condition
is close to a result which we will see, but note as well that
even this requires something special about sets,
namely the Replacement Axiom.

The notion of well-foundedness of a coalgebra $(A,\alpha)$
generalizes condition (1) above.   It says that
no proper subcoalgebra~ $m:(A',\alpha')\hookrightarrow (A,\alpha)$
forms a pullback
\[\bfig
\square/->`^{ (}->`^{ (}->`->/<700,400>[A'`HA'`A`HA;%
\alpha'`m`Hm`\alpha]
\POS(120,300)
\ar@{-}+(-50,0)
\ar@{-}+(0,50)
\efig\]
This concept was first studied by Osius~\cite{O} for graphs considered as coalgebras for
the power-set functor $\pow$: a graph is well-founded 
in the coalgebraic sense iff it is well-founded in any of the equivalent 
senses above.  Taylor~\cite{Ta1,Ta2} introduced
well-founded coalgebras for general endofunctors, and he proved that
for set functors preserving inverse images the concepts of initial
algebra and final well-founded coalgebra coincide.

\takeout{%
We must mention that our motivation differs from Taylor's.
He is concerned with foundational matters connected to
recursion and induction, while we are interested in studying 
initial algebras and final coalgebras in as wide a setting as possible.
}%

Returning to our topic, we are going to prove that for \emph{every} set
functor $H$ the concepts of initial algebra and final well-founded
coalgebra coincide; the step towards making no assumptions on $H$ is
non-trivial.  We also prove the same result for endofunctors of locally finitely presentable categories preserving finite intersections. And if $H$ preserves    (wide) intersections, we describe its
final coalgebra and initial algebra using well-pointed coalgebras.

The last section takes a number of known important special cases:
deterministic (Mealy and Moore) automata, trees, labeled transition systems,
non-well-founded sets, etc., and demonstrates how well-pointed
coalgebras work in each case. Here we describe, in every example,
besides the initial algebra and the final coalgebra, the initial
iterative algebra~\cite{AMV} (equivalently, final locally finite coalgebra, see~\cite{M2,bms11})
as the set of all finite well-pointed coalgebras.

\section{Well-founded coalgebras}
\label{druha}
In this section we recall the concept of well-founded  coalgebra of
Osius~\cite{O} and Taylor~\cite{Ta1}. Our main result is that
\[\text{initial algebra}=\text{final well-founded coalgebra}\]
holds for all endofunctors of~\Set. (In the case where the endofunctor
preserves inverse images, this result can be found in~\cite{Ta1}.)
For more general categories the  result above holds whenever the
endofunctor preserves finite intersections.

\subsection{Well-founded coalgebras in locally finitely presentable categories}\label{Asss}
\hfill

\par\smallskip\noindent
We make several assumptions on the base category $\A$ in our study.

\begin{definition}\label{D-2.1}\hfill
\begin{enumerate}[(1)]
\item
A category $\A$ is \emph{locally finitely presentable} (LFP) if
\begin{enumerate}[(a)]
\item \A~is complete;
\item there is a set of finitely presentable objects whose closure under filtered colimits is all of~\A.
\end{enumerate}
 (See \cite{GU} or~\cite{AR} for more on LFP categories.)
\item An object $A$ of  (any category) $\A$ is called \emph{simple} if every strong epimorphism (see Remark \ref{R-fact}) with domain $A$ is invertible.  (In categories with (strong epi, mono)-factorizations, see Remark \ref{R-fact}, this is equivalent to saying that every morphism with domain $A$ is a monomorphism.)
\end{enumerate}
\end{definition}

\begin{remark}\label{R-new-simple} The concept of simple object stems from general Algebra, where
strong epimorphisms are precisely the surjective homomorphisms, thus,
an algebra is simple iff it has no nontrivial congruence.
\end{remark}

\begin{assumption}\label{ASS}   Throughout this section our base category $\A$ is  locally finitely presentable 
and has a simple initial object $0$. 
\end{assumption}

\begin{examples}
The categories of sets, graphs, posets, and semigroups 
are locally finitely presentable.   The initial objects of these categories are 
empty, hence simple. The LFP category of rings has the initial object
$\mathbb{Z}$ that is not simple.%
\takeout{For another such example see \ref{E-bip} below.
The category~$\Set_{0,1}$ of bipointed sets fulfils 1, but
not~2.}
\end{examples}

\begin{notation}
\label{N-fact2}
For every endofunctor~$H$
denote by
\[\Coalg H\]
the category of coalgebras $\alpha\colon A\to HA$ and coalgebra
homomorphisms. 
\end{notation}

\takeout{%
\noindent Since subcoalgebras play a basic role in the whole paper, and
quotients are important from Section~\ref{treti} onwards, we need to
make clear what we mean by those. Quotients are no problem: it is
clear that the forgetful functor of the category of coalgebras
preserves and reflects all colimits.  Consequently, epimorphisms in
$\Coalg H$ are precisely the homomorphisms carried by epimorphisms in
the base category.  And they represent the quotients of the domain
coalgebra (up to isomorphism, as usual).  What about subcoalgebras? If
the base category is $\Set$, it turns out that the homomorphisms
carried by monomorphisms are precisely the strong monomorphisms of
$\Coalg H$.  (Recall that a monomorphism is called \emph{strong} if it
has the diagonal fill-in property w.r.t. all epimorphisms.  In
``everyday'' categories this is equivalent to being a regular
monomorphism.)  As shown in Lemma \ref{str}, for general base categories we have an analogous fact
whenever the endofunctor $H$ preserves strong monomorphisms:
strong monomorphisms in~$\Coalg H$ are precisely the homomorphisms
$h\colon(A,\alpha)\to(B,\beta)$ for which $h$~is strongly monic
in~\A. 
For that reason we use the term \emph{subcoalgebra} of a coalgebra $(A,
\alpha)$ to mean a subobject
represented by a strong monomorphisms $m\colon (A', \alpha') \to (A,
\alpha)$ in $\Coalg H$.  But as we point out in Section~\ref{M}, one can obtain
analogous results for more general factorization systems.
}%

\begin{remark}
\label{R-fact}
There are some consequences of the LFP assumption that play an important role in our development:
\begin{enumerate}[1.]
\item $\A$ has (strong epi, mono)-factorizations, see
\cite[Proposition 1.16]{AR}. (Recall that an epimorphism $e$ is called \textit{strong} if it fulfils the diagonal fill-in property w.~r.~t.~all monomorphisms, i.e., $fe=mg$ with $m$ a monomorphism implies the existence of a unique factorization of $g$ through $e$.)

\item  $\A$ is  wellpowered,   see \cite[Remark 1.56]{AR}.  This implies that for every object $A$ the
poset $\Sub(A)$ of all  subobjects of $A$ is a complete lattice.

\item
Monomorphisms are closed under filtered colimits (see \cite[Proposition 1.62]{AR}).  We also use the fact (true in every category) that monomorphisms are closed under wide intersections and inverse images. \end{enumerate}
\end{remark}

\noindent Since subcoalgebras play a basic role in the whole paper, and
quotients are important from Section~\ref{treti} onwards, we need to
make clear what we mean by those. This is the aim of Remark \ref{R-str} and Terminology \ref{T-sub}.

\begin{remark} \label{R-str}  Assuming that $H$ preserves
    monomorphisms, homomorphisms of coalgebras factorize into those
    carried by strong epimorphisms followed by those carried by
    monomorphisms. Moreover, the two classes of homomorphisms form a
    factorization system in $\Coalg H$.  Indeed, let $h$ be a
    coalgebra homomorphism from the coalgebra $(A, \alpha)$ to the
    coalgebra $(B, \beta)$ and let $h=m\tec e$ be a
(strong epi,  mono)-factorization in~\A, then the diagonal fill-in property
yields a coalgebra for which $m$ and~$e$ are homomorphisms:
\[\bfig
\square(0,400)/->`->`->`/<700,400>[A`C`HA`B;e`\alpha`m`]
\square(0,0)/`->`->`->/<700,400>[HA`B`HC`HB;`He`\beta`Hm]
\morphism(700,800)/-->/<-700,-800>[C`HC;\gamma]
\efig\]
The diagonal fill-in property in $\Coalg H$ follows easily, too.

We also point out that the monomorphisms of $\Coalg H$ need not be
carried by mono\-mor\-phisms in $\A$.
\end{remark}

\begin{terminology}\label{T-sub}  When we speak about \textit{subcoalgebras} of a coalgebra $(A,\alpha)$ we mean those represented (up to isomorphism) by homomorphisms $m:(A',\alpha')\to (A, \alpha)$ with $m$ a monomorphism in $\A$. As usual, if $m$ is not invertible, the subcoalgebra is said to be \textit{proper}.  \textit{Quotients} of $(A,\alpha)$ are represented by homomorphisms with domain $(A,\alpha)$ carried by a strong epimorphism in $\A$; again, properness means they are not invertible.
\end{terminology}

\takeout{%
   \begin{example}
   \label{set-0-1}\label{BiPH}
    For a non-example which is still interesting for this paper, 
    we consider the category 
$\Set_{0,1}$   of \emph{bipointed sets};  
these are sets with two chosen points which morphisms must fix.
 $\Set_{0,1}$ is LFP.
The initial object $0$ is a set with two different elements, both chosen.
The final object $1$ is a single point.  The map $0\to 1$ is  an epimorphism,
so $0$ is not simple. 

In the category~$\Set_{0,1}$ of
bipointed sets, put
\[H(X,x_0,x_1)=\begin{cases}
1\quad\text{(final object)}&\text{if $x_0=x_1$}\\
(X+1,x_0,x_1)&\text{else}
\end{cases}\]
This $H$ preserves monomorphisms. 
However, we saw above that $0$ is not simple.
So $H$ will re-appear in examples which show
that  the simplicity of $0$ is necessary in most of our results below.
\takeout{as will the functor from Example~\ref{E-gr}.}
\end{example}
 
\begin{example}
\label{E-gr}
On the category $\Gra=\Set^{\mbox{\tiny$\stackrel{\scriptscriptstyle\rightarrow}{%
\scriptscriptstyle\rightarrow}$}}$
of graphs define an
endofunctor~$H$ by
\[HX=\begin{cases}
X+\{t\}\ (\text{no edges})&\text{if $X$ has no edges}\\
1, \text{terminal graph,}&\text{else.}
\end{cases}\]
Of course, this functor
$H$ does not preserve monomorphisms,
and so most of the foregoing results do not apply to it. (See also Example \ref{E-gr3}.)
\end{example}
}%
\takeout{%
\begin{lemma}
\label{efine
  $a^*_i$ to be the colimit morphism. }
Given a filtered colimit with a cocone $c_i\colon C_i\to C$ $(i\in
I)$,   every morphism $f\colon C\to D$ for which $f\tec c_i$~are
strong monomorphisms $(i\in I)$ is a strong monomorphism.
\end{lemma}
\begin{proof}
It is our task, for every commutative square
\[\bfig
\square<700,400>[X`Y`C`D;e`u`v`f]
\efig\]
where $e$~is an epimorphism to find a diagonal. We can assume,
without loss of generality, that $X$~is finitely presentable: indeed,
every epimorphism in a locally finitely presentable category is a
filtered colimit of epimorphisms with finitely presentable domains.

Since $C=\colim C_i$ is a filtered colimit, there exists~$i$ such
that $u$~factorizes through~$c_i$.
\[\bfig
\dtriangle<600,500>[X`C_i`C;u'``c_i]
\square(600,0)<800,500>[X`Y`C`D;e`u`v`f]
\morphism(1400,500)/@{-->}_<>(.35){d}/<-1400,-500>[Y`C_i;]
\efig\]
This yields, since $f\tec c_i$~is a strong monomorphism, a diagonal
$d\colon Y\to C_i$ for the outward square. Then $c_i\tec d$~is the
desired diagonal for the original square. 

To show that $f$ is a monomorphism, assume that $f \cdot m = f \cdot
n$. Take the coequalizer $e$ of $m$ and $n$, and let $w$ be the unique
mediating morphism with $w \cdot e = f$. Then the unique diagonal of
the commutative square $f \cdot \id = w \cdot e$ satifies $d \cdot e =
\id$, whence $e$ is an isomorphism. Thus, $m = n$ as desired. 

\end{proof}
}%

\begin{definition}
\label{D-wf}  
A \textbf{cartesian subcoalgebra} of a coalgebra $(A,\alpha)$ is a
subcoalgebra  $m:(A',\alpha')\hookrightarrow (A,\alpha)$  forming a pullback
\[\bfig
\square/->`^{ (}->`^{ (}->`->/<700,500>[A'`HA'`A`HA;%
\alpha'`m`Hm`\alpha]
\POS(120,400)
\ar@{-}+(-50,0)
\ar@{-}+(0,50)
\efig\]
A coalgebra is called \textbf{well-founded} if it has no proper
cartesian subcoalgebra.
\end{definition}

\begin{example}
\label{E-wf}\hfill
\begin{enumerate}[(1)]
\item The concept of well-founded coalgebra was introduced originally
  by Osius~\cite{O} for the power set functor~\PP.  Recall that coalgebras for~\PP~are simply graphs: given $\alpha:A\rightarrow \PP A$, then $\alpha(x)$ is the set of neighbors of
  $A$ in the graph.  However, coalgebra homomorphisms $h:A\rightarrow B$ are stronger than graph homomorphisms: $h$ not only preserves edges of $A$, but also for every edge $h(a)\rightarrow b$ in $B$ there exists an edge $a\rightarrow a'$ in $A$ with $b=h(a')$. Then a subcoalgebra of $A$ is an (induced)
  subgraph $A'$ with the property that every neighbor of a vertex of
  $A'$ lies in $A'$.  The subgraph $A'$ is cartesian iff it contains
  every vertex all of whose neighbors lie in~$A'$.  
  
  The graph $A$ is a
  well-founded coalgebra iff it has no infinite path. Indeed, the set $A'$ of all vertices lying on no infinite path forms clearly a cartesian subcoalgebra. And $A$ is well-founded iff $A=A'$.

\item Let $A$~be a deterministic automaton considered as a coalgebra
  for $HX=X^I\times\{0,1\}$. A subcoalgebra~$A'$ is cartesian iff it
  contains every state all whose successors (under the inputs from
  $I$) lie in~$A'$. This holds, in particular, for $A'=\emptyset$.
  Thus, no nonempty automaton is well-founded.

\item Coalgebras for $HX=X+1$ are dynamical systems with deadlocks. A
  subcoalgebra~$A'$ of a coalgebra~$A$ is cartesian iff $A'$
  contains all deadlocks and every state whose next state lies
  in~$A'$. So a dynamical system is well-founded iff it has no infinite
  computation.
\end{enumerate}
\end{example}

\begin{proposition}\label{P-init} Initial algebras are, as coalgebras, well-founded. 
\end{proposition}

\begin{remb}No assumptions on the base category are needed in the proof.
\end{remb}

\begin{proof}Let $\varphi:HI\rightarrow I$ be an initial algebra. Given a pullback
\[\bfig
\square/->`->`->`->/<700,400>[B`HB`I`HI;%
\beta`m`Hm`{\varphi^{-1}}]
\POS(90,310)
\ar@{-}+(-50,0)
\ar@{-}+(0,50)
\efig\]
with $m$ monic, we prove that $m$ is invertible. It is clear that $\beta$ is invertible (since $\varphi^{-1}$ is), and for the algebra $\beta^{-1}:HB\rightarrow B$ there exists an algebra homomorphism $f:(I,\varphi)\rightarrow (B, \beta^{-1})$. Since $m$ is also an algebra homomorphism, we conclude that $mf$ is an endomorphism of the initial algebra. Thus, $mf=\id$, proving that $m$ is invertible.
\end{proof}

\begin{remark} In contrast, final coalgebras are never well-founded, unless they coincide with initial algebras.
\end{remark}

To prove this, we are going to use the \textit{initial chain} defined in~\cite{A}.
This is the chain
\begin{equation}
H^i 0\quad(i\in\Ord)\qquad\text{and}\qquad w_{ij}\colon H^i 0\to
H^j 0\quad(i\leq j)
\label{R-chain}
\end{equation}
defined uniquely up to natural isomorphism by
\begin{align*}
H^0 0&=0\qquad\quad\text{(initial object of \A)}\\
H^{i+1} 0&=HH^i 0\quad\text{and}\quad w_{i+1,j+1}=Hw_{i,j}\\
\intertext{and for limit ordinals~$i$}
H^i 0&=\colim_{j<i} H^j 0\qquad\text{with colimit cocone
$w_{ij}\quad(i<j)$.}
\end{align*}
The chain is said to \textit{converge} at~$i$ if the connecting map
$w_{i,i+1}\colon H^i 0\to HH^i 0$ is invertible. The inverse then makes $H^i 0$
an initial algebra.

\begin{proposition}[\cite{TAKR}]
\label{P-chain}  Let $H$ preserve monomorphisms.
\begin{enumerate}[\em(1)]
\item Whenever there exists a fixed point of $H$, i.e. an object $X\cong HX$, 
then $H$ has an
  initial algebra. 
\item If $H$  has an initial algebra, then the initial chain converges. 
\end{enumerate}
\end{proposition}

\begin{remark}\label{R-cone}
This result was shown in Theorem II.4 of~\cite{TAKR}.
The proof uses 2.~and 3.~of Remark \ref{R-fact}.  It is based  on the fact that the isomorphism $u:HX\rightarrow X$ yields a cone $m_i:H^i0\rightarrow X$ ($i\in \Ord$) of the initial chain with all $m_i$ monic: $m_0:0\rightarrow X$ is unique and $m_{i+1}=u\cdot Hm_i$. Thus, the initial chain converges because it is a chain of subobjects of $X$.
\end{remark}
 
\takeout{%
\begin{theorem}
\label{T-initA}  Let $H$ preserve monomorphisms.
Initial algebras are, as coalgebras, well-founded.
\end{theorem}

This was proved in Taylor~\cite{Ta2} under the additional assumption that
$H$~preserves inverse images.

\begin{proof}
If $H$~has an initial algebra, then, by Proposition~\ref{P-chain}, it has the
form $(w_{j,j+1})^{-1}\colon HW_j\to W_j$ for some ordinal~$j$. We
prove that
for $i\leq j$, 
 the morphisms $a^*_i$ of Notation~\ref{N-starr}
are the same as the morphisms $w_{ij}$.
Consequently, $a^*=\id_{W_j}$, as requested. For
$i=0$
the equality $a^*_0=w_{0j}$ is clear. For the isolated step we need to
prove that the square
\[\bfig
\square<800,400>[H^{i+1} 0`HH^i 0`W_j`HW_j;\id`w_{i+1,j}`Hw_{ij}`w_{j,j+1}]
\efig\]
is a pullback. Indeed, the square commutes since
$Hw_{ij}=w_{i+1,j+1}$, and it is a pullback since both horizontal
arrows are invertible. Limit steps follow automatically.
 The coalgebra
$W_j$ is thus well-founded by Proposition~\ref{P-starr}.

\end{proof}
}%

\begin{proposition} \label{only}
If $H$ preserves monomorphisms, the only well-founded fixed points of $H$ are the initial algebras.
\end{proposition}

\begin{proof} Let $u:HX\stackrel{\sim\;}{\rightarrow}X$ be a fixed point such that $u^{-1}:X\rightarrow HX$ is a well-founded coalgebra. Then we prove that $(X,u)$ is an initial algebra. Let $m_i:H^i0\rightarrow X$ be the cone of Remark \ref{R-cone}. We know that there exists an ordinal $j$ such that $m_j$ and $m_{j+1}$ represent the same subobject, thus, $w_{j, j+1}: H^i0 \rightarrow H(H^i0)$ is invertible. Consequently, $w^{-1}_{j,j+1}:H(H^i0)\rightarrow H^i 0$ is an initial algebra.

The following square
$$\xymatrix{H^j 0\ar[rr]^{w_{j,j+1}}\ar[d]_{m_j}&&H(H^j 0)\ar[d]^{Hm_j}\\
X\ar[rr]_{u^{-1}}&&HX}$$
commutes: by definition we have $m_{j+1}=u\cdot Hm_j$ and since $m_j=m_{j+1}\cdot w_{j,j+1}$ (due to the compatibility of the $m_i$'s) we conclude 
$$m_j=u\cdot Hm_j\cdot w_{j, j+1}.$$
Since both horizontal arrows are invertible, the  square above is a pullback. From the well-foundednes of $(X,u^{-1})$ we conclude that $m_j$ is invertible. Thus, the algebra $(X,u)$ is isomorphic to the initial algebra $(H^j 0, w^{-1}_{j,j+1})$ via $m_j$. This proves that $(X,u)$ is initial.
\end{proof}

\begin{corollary}
\label{only2}
 If $H$ preserves monomorphisms and has a well-founded final coalgebra, then the initial algebra and final coalgebra coincide.
\end{corollary}

\takeout{
\begin{example}
\label{E-bip}
Here we demonstrate that our assumption that  $0$ be simple is essential
in Proposition~\ref{only} and Corollary~\ref{only2}.
\begin{enumerate}[(a)]
\item
On the category of rings with unit consider the identity functor. Its terminal coalgebra is the trivial ring ($0=1$) and it is well-founded because it has no proper subrings. However, the initial algebra is the initial ring $\mathbb{Z}$.

\item  For a simpler category, consider 
 the category 
$\Set_{0,1}$   of \emph{bipointed sets};  
these are sets with two distinguished points which morphisms must fix.
 $\Set_{0,1}$ is LFP.
The initial object $0$ is a set with two different elements, both distinguished.
The final object $1$ is a single point. Again the identity functor has a well-founded final coalgebra different from its initial algebra.
\end{enumerate}
\end{example}
}

\begin{example}
\label{E-graphh}
This demonstrates that the assumption that $H$ preserves monomorphisms is essential. Consider the category $\Gra$ of graphs and graph morphisms (i.e., functions preserving edges). All assumptions in \ref{ASS} are fulfilled. 
The endofunctor
\[HX=\begin{cases}
X+\{t\}\ (\text{no edges})&\text{if $X$ has no edges}\\
1, \text{terminal graph,}&\text{else.}
\end{cases}\]
does not preserve monomorphisms. Its final coalgebra $1=H1$ is well-founded because neither
 of the two proper  subcoalgebras is cartesian. However, the initial algebra is carried by an infinite graph without edges.
\end{example}

\begin{definition}
\label{D-tilde}
Assume that $H$ preserves monomorphisms. Then for every coalgebra $\alpha\colon A\to HA$ we denote by $\tilde$ the endofunction on~$\Sub(A)$ (see Remark~\ref{R-fact}.2)  assigning to every subobject $m\colon A'\to
A$ the inverse image of~$Hm$ under~$\alpha$, i.\,e., we
have a pullback square:
\begin{equation}
\label{rdj}
\bfig
\square<700,500>[\tilde{A'}`HA'`A`HA;%
\alpha{[m]}`\tilde{m}`Hm`\alpha]
\POS(100,400)
\ar@{-}+(-50,0)
\ar@{-}+(0,50)
\efig
\end{equation}
\end{definition}

This function $m\to/|->/\tilde{m}$ is obviously order-preserving.  
By the Knaster-Tarski fixed point theorem, it
has a least fixed point.

\begin{corollary}
\label{C-least}
A coalgebra $(A,a)$ is well-founded iff the least fixed point of $\tilde$ is all of $A$.
\end{corollary}

Incidentally, the notation $\tilde{m}$ comes from modal logic, especially 
the areas of temporal logic where one reads $\tilde{\phi}$ as ``$\phi$ is true
in the next moment,'' or ``next time $\phi$" for short.

\begin{example}   
\label{E-ptilde}
Recall our discussion of graphs from Example~\ref{E-wf} (1).
The pullback  $\tilde{A'}$ of a subgraph $A'$
is  the set of points in the graph $A$
all of whose neighbors belong to $A'$.
\end{example}

\begin{remark}
  \label{R-wf} As we mentioned in the introduction, the concept of
  well-founded coalgebra was introduced by Taylor~\cite{Ta1,Ta2}.
Our formulation is a bit simpler. In~\cite[Definition~6.3.2]{Ta2} he calls
a coalgebra $(A,\alpha)$ well-founded if for every pair of
monomorphisms $m\colon U \to A$ and $h\colon H \to U$ such that $h\tec
m$ is the inverse image of $Hm$ under $\alpha$ it follows that $m$ is
an isomorphism. Thus, in lieu of fixed 
points of $m \longmapsto \bigcirc{m}$ he uses pre-fixed points.

In addition, 
our overall work has a \emph{methodological} difference from Taylor's that is worth mentioning at this point.
Taylor is giving a general account of recursion and induction,
and so he is concerned with general principles that
underlie these phenomena.
Indeed, he is interested in settings like non-boolean toposes
where classical reasoning is not necessarily valid.
On the other hand, in this paper we are studying initial algebras,
final coalgebras, and similar concepts, using
standard classical mathematical reasoning. In particular, we
make free use of transfinite induction. 
\end{remark}

\takeout{%
\begin{example}
\label{825}
Here is an example showing that preservation of strong monomorphisms does not in general imply preservation of monomorphisms.
On the category $\Gra$ of graphs, this time let
$$\begin{array}{lcl}
HA &  = & \mbox{all finite independent $a \subseteq A$, together with a new point $t$}  \\
& & \mbox{with $a \leftrightarrow t$ for all $a$, and also $t \rightarrow t$} \\
\end{array}
$$
For a graph morphism $f: A\to B$, we take $Hf: HA \to HB$ to be 
$$
Hf(a)  = \begin{cases}
f[a] & \mbox{if $f[a]$ is independent in $B$}  \\
t & \mbox{otherwise} \\
\end{cases}
$$
This functor $H$ preserves strong monomorphisms
(they are the induced subgraphs),
and indeed it preserves intersections of them as well.
However, $H$ does not preserve   monomorphisms.
So we expect some of the results which depend on preservation 
of $\M$ will fail with $\M = $ all monomorphisms.
\end{example}
}%

\takeout{ 
\begin{example}
\label{E-gr2}
We return to Example~\ref{E-gr} concerning an 
endofunctor of graphs not preserving monomorphisms.
The final coalgebra $1=H1$ is well-founded. Indeed, the only proper
subgraphs of~1 are $m_0\colon0\to1$ and $m_1\colon\{t\}\to1$ with
0~the empty graph, and $\{t\}$~having no edges. Neither of the squares
\[
\bfig
\square/->`->`->`=/<700,500>[0`\{t\}`1`H1;`m_0`Hm_0`]
\square(1500,0)/->`->`->`=/<700,500>[\{t\}`\{t\}+\{t\}`1`H1;%
`m_1`Hm_1`]
\efig
\]
is a pullback. However, $m_1$~is a prefixed point of
$m\to/|->/\tilde{m}$.
\end{example}
}

\begin{notation}
\label{N-starr}\hfill
\begin{enumerate}[(a)]
\item Assume that $H$ preserves  monomorphisms.  For every
  coalgebra $\alpha\colon A\to HA$ denote by
  \begin{equation}
    \label{rdd}
    a^*\colon A^*\to A
  \end{equation}
  the least fixed point of the function $m\to/|->/\tilde{m}$
  of Definition~\ref{D-tilde}. (Thus, $(A,\alpha)$~is well-foun\-ded iff $a^*$~is
  invertible.) Since $a^*$~is a fixed point we have a coalgebra
  structure $\alpha^*\colon A^*\to HA^*$ making~$a^*$ a coalgebra homomorphism.
  
\item For every coalgebra $\alpha\colon A\to HA$ we define a chain of
   subobjects
  \[
  a^*_i\colon A^*_i\to A\qquad (i\in\Ord)
  \]
  of~$A$ in~\A\ by transfinite recursion:
  $a^*_0\colon 0\to A$ is unique; given~$a^*_i$, define~$a^*_{i+1}$ by the pullback
  \[
  \bfig
  \square<700,500>[A^*_{i+1}`HA^*_i`A`HA;%
  `a^*_{i+1}`Ha^*_i`\alpha]
  \POS(110,400)
  \ar@{-}+(-50,0)
  \ar@{-}+(0,50)
  \efig
  \]
  and for limit ordinals $i$ 
  define $a^*_i\colon A^*_i\to A$ to be the union of the chain of monomorphisms $a^*_j:A^*_j \to A$,
  \[a^*_i=\bigcup_{j<i} a^*_j.\]  
 It is easy to prove by transfinite induction that all $a_i^*$
   are monic (for $i=0$ recall that $0$ is simple). Moreover, for
   every limit ordinal $i$ the  union above coincides with the colimit
 of the chain, that is, the monomorphism $a^*_i:A^*_i\to A$ is just the induced morphism from the colimit of the chain to $A$, 
 see  Remark \ref{R-fact}, point 3.
   \end{enumerate}
\end{notation}
\begin{remark}
\label{R-connect}
We observe that for all ordinals $i\leq j$ the connecting maps
\[\bfig
\Vtriangle[A^*_i`A^*_j`A;a^*_{ij}`a^*_i`a^*_j]
\efig\]
of the chain of Notation~\ref{N-starr} form the following commutative 
diagram which can be used as a definition of the maps $a_{ij}^*$ (via the
universal property of pullbacks):
\begin{equation}
\label{nrdjt}
\bfig
\square|brlb|<800,400>[A^*_{j+1}`HA^*_j`A`HA;%
\alpha{[a^*_j]}`a^*_{j+1}`Ha^*_j`\alpha]
\square(0,400)|arlb|<800,400>[A^*_{i+1}`HA^*_i`A^*_{j+1}`HA^*_j;%
{\alpha[a^*_i]}`a^*_{i+1,j+1}`Ha^*_{ij}`]
\morphism(0,800)|l|/{@{>}@/_2em/}/<0,-800>[A^*_{i+1}`A;%
a^*_{i+1}]
\morphism(800,800)|r|/{@{>}@/^2em/}/<0,-800>[HA^*_i`HA;%
Ha^*_i]
 \POS(110,285)
  \ar@{-}+(-50,0)
  \ar@{-}+(0,50)
\efig
\end{equation}
\end{remark}

\begin{remark}
\label{R-least} This way, what
we have is nothing else than the construction of the least fixed point
of $m\to/|->/\tilde{m}$, see Remark~\ref{R-wf}, in the proof of the
Knaster-Tarski Theorem in \cite{Tar}. Thus,
$a^*=\bigcup_{i\in\Ord}a^*_i$.
However, since $A$~has only a set of subobjects,
\begin{equation}
  a^*=a^*_{i_0}\qquad\text{for some ordinal $i_0$.}
  \label{starsfso}
\end{equation}
And for this ordinal $i_0$, an easy verification shows that the coalgebra structure of $A^*$ above is
\begin{equation}
  \alpha^* =\alpha[a^*_{i_0}] = \alpha[a^*].
  \label{starsfso2}
\end{equation}
Henceforth, we  call $A^*$ the \emph{smallest} cartesian subcoalgebra of $A$.
\end{remark}

From now on, whenever we use the notations $\tilde{m}$
and $a^*$, we only do so when $H$ preserves monomorphisms.

\begin{example}
\label{E-gr3} On the category 
$\Gra$  consider the functor $H$ of Example \ref{E-graphh}. It has $1=H1$ as its final coalgebra, and this coalgebra is well-founded. However, for 
$\alpha$ as $\id: 1 \to H1$ (the final coalgebra), there is no ordinal $i$ 
such that $a^*_i = \id_1$.
This shows that Notation \ref{N-starr} is meaningful only if we assume
that $H$ preserves monomorphisms.
\end{example}

\begin{example}
\label{E-starr}
For every graph~$A$  considered as a coalgebra for~\PP, $A^*$~is the subgraph on
all vertices of~$A$ from which no infinite path starts. 
Since $m\mapsto \tilde{m}$ is not necessarily continuous, 
the
ordinal~$i_0$  
of (\ref{starsfso})
above can be arbitrarily large. Here is an example with
$i_0=\omega+1$:
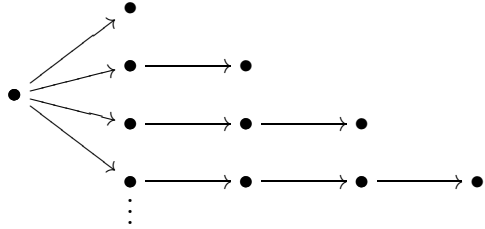
\[\bfig
\morphism(400,350)<400,-300>[\bullet`\bullet;]
\morphism(400,350)<400,-100>[\bullet`\bullet;]
\morphism(400,350)<400,100>[\bullet`\bullet;]
\morphism(400,350)<400,300>[\bullet`\bullet;]
\morphism(800,450)<400,0>[\bullet`\bullet;]
\morphism(800,250)<400,0>[\bullet`\bullet;]
\morphism(800,50)<400,0>[\bullet`\bullet;]
\morphism(1200,50)<400,0>[\bullet`\bullet;]
\morphism(1200,250)<400,0>[\bullet`\bullet;]
\morphism(1600,50)<400,0>[\bullet`\bullet;]
\place(800,-70)[\vdots]
\efig\]
\end{example}

\begin{proposition}  
\label{P-starr}
 If $H$ preserves monomorphisms then well-founded coalgebras form a full coreflective subcategory of $\Coalg~H$: For every coalgebra~$(A,\alpha)$, the  smallest cartesian
subcoalgebra $(A^*,\alpha^*)$ is its coreflection.
\end{proposition}

\begin{remb}
We thus prove that $(A^*,\alpha^*)$~is well-founded, and for every
homomorphism $f\colon(B,\beta)\to(A,\alpha)$ with
$(B,\beta)$~well-founded there exists a unique homomorphism
\[\bar{f}\colon(B,\beta)\to(A^*,\alpha^*)\qquad\text{with}\qquad
f=a^*\tec\bar{f}.\]
\end{remb}

\begin{proof}
(i) $(A^*,\alpha^*)$~is clearly well-founded: From Definition \ref{D-tilde} and Notation \ref{N-starr}, we know that $(A^*,a^*)$ is the least fixed point of $\tilde:\Sub(A)\to\Sub(A)$, that is, $(A^*,\alpha^*)$ is the smallest cartesian subcoalgebra of $(A,\alpha)$. Then $(A^*,\alpha^*)$ cannot have proper cartesian subcoalgebras since its cartesian subcoalgebras are   cartesian subcoalgebras of  $(A,\alpha)$.

(ii)
 Since $a^*$~is a
monomorphism there is
at most one coalgebra homomorphism  $\bar{f}\colon B \to A^*$ with
$a^* \tec \bar{f} =f$.
Thus, we are finished if we show that $\bar{f}$~exists.
To this end, for all ordinals $i\leq j$, let $a^*_{ij}:A_i^*\to A_j^*$ be the connecting maps
of the chain of  Remark \ref{R-connect}. Analogously, use
$b^*_{ij}\colon B^*_i \to B^*_j$ for
the chain of the subobjects $b_i^*:B_i^*\to B$, whose union is $B^* = B$. We define 
  the components of a
  natural transformation $\bar{f}_i\colon B_i^* \to A_i^*$, $i\in \Ord$,
  by transfinite recursion on 
ordinals $i$,   satisfying
\begin{equation}
\label{diag:sq}
\bfig
\square<700,400>[B^*_i`A^*_i`B`A;\bar{f}_i`b^*_i`a^*_i`f]
\efig
\end{equation}
\takeout{%
The naturality follows due to the commutativity from the  below
for $i \leq j$:
\begin{equation}
\label{diag:sq2}
\bfig
\square|arlb|<800,400>[B^*_j`A^*_j`B`A;%
\bar{f}_j`b^*_j`a^*_j`f]
\square(0,400)|arlb|<800,400>[B^*_i`A^*_i`B^*_j`A^*_j;%
\bar{f}_i`b^*_{i,j}`a^*_{i,j}`]
\morphism(0,800)|l|/{@{>}@/_2em/}/<0,-800>[B^*_i`B;%
b^*_i]
\morphism(800,800)|r|/{@{>}@/^2em/}/<0,-800>[A^*_i`A;%
a^*_i]
\efig
\end{equation}
}%
Let $\bar{f}_0 = \id\colon 0 \to 0$.
 For isolated
steps
consider the diagram below:
\begin{equation}
\label{diag:sq3}
\bfig
\square(600,400)<800,400>[A^*_{i+1}`HA^*_i`A`HA;%
{\alpha[a^*_i]}`a^*_{i+1}`Ha^*_i`\alpha]
\morphism(0,1200)<2000,0>[B^*_{i+1}`HB^*_i;{\beta[b^*_i]}]
\morphism(0,1200)/-->/<600,-400>[B^*_{i+1}`A^*_{i+1};\bar{f}_{i+1}]
\morphism(0,1200)|l|<0,-1200>[B^*_{i+1}`B;b^*_{i+1}]
\morphism(2000,1200)|l|<-600,-400>[HB^*_i`HA^*_i;H\bar{f}_i]
\morphism(2000,1200)|r|<0,-1200>[HB^*_i`HB;Hb^*_i]
\morphism(0,0)<600,400>[B`A;f]
\morphism(0,0)|b|<2000,0>[B`HB;\beta]
\morphism(2000,0)|a|<-600,400>[HB`HA;Hf]
\efig
\end{equation}
The inner and outside squares commute by the definition of $A^*_{i+1}$
and
$B^*_{i+1}$, respectively. For the lower square we use that $f$ is a
coalgebra homomorphism, and the right-hand one commutes by the
induction
hypothesis. The inner pullback induces the desired morphism
$\bar{f}_{i+1}$
and the commutativity of the left-hand square is that
of~\eqref{diag:sq} for~${i+1}$.
Finally, for a limit ordinal $j$ let $\bar{f}_j = \colim_{i<j}
\bar{f}_i$, in other words, $\bar{f}_j$ is the unique morphism such
that the squares
\begin{equation}\label{diag:nat}
\bfig
\square/->`<-`<-`->/<700,400>[B^*_j`A^*_j`B^*_i`A^*_i;%
\bar{f}_j`b^*_{i,j}`a^*_{i,j}`\bar{f}_i]
\efig
\end{equation}
commute for all $i < j$. It is easy to prove  by transfinite induction that $\bar{f}_j:B^*_j\rightarrow A^*_j$ is natural in $j$.

We need to verify that~\eqref{diag:sq}
commutes for
$\bar{f}_j$. This is clear for $j=0$ and for $j$ isolated this follows from the  definition of $\bar{f}_{i+1}$. Let $j$ be a limit ordinal. Then \eqref{diag:sq} commutes due to the following diagram for every $i<j$:
\begin{equation}
\label{diag:sq2}
\bfig
\square|arlb|<800,400>[B^*_j`A^*_j`B`A;%
\bar{f}_j`b^*_j`a^*_j`f]
\square(0,400)|arlb|<800,400>[B^*_i`A^*_i`B^*_j`A^*_j;%
\bar{f}_i`b^*_{i,j}`a^*_{i,j}`]
\morphism(0,800)|l|/{@{>}@/_2em/}/<0,-800>[B^*_i`B;%
b^*_i]
\morphism(800,800)|r|/{@{>}@/^2em/}/<0,-800>[A^*_i`A;%
a^*_i]
\efig
\end{equation}

To complete the proof consider any ordinal~$i$ such that $B^*_i =
B^* = B$
and $A^*_i = A^*$ hold. Then $\bar{f} = \bar{f}_i\colon B \to A^*$ is
a coalgebra
homomorphism with $a^*_i \tec \bar{f} = f$ by the commutativity of the
upper and left-hand parts of Diagram~\eqref{diag:sq3}. %
\takeout{%
(i)
We first observe that for all ordinals $i\leq j$ the connecting maps
\[\bfig
\Vtriangle[A^*_i`A^*_j`A;a^*_{ij}`a^*_i`a^*_j]
\efig\]
of the chain of Notation~\ref{N-starr} form the following commutative
diagram which can be used as a definition of~$a_{ij}^*$'s (via the
universal property of pullbacks):
\begin{equation}
\label{nrdjt}
\bfig
\square|brlb|<800,400>[A^*_{j+1}`HA^*_j`A`HA;%
\alpha{[a^*_j]}`a^*_{j+1}`Ha^*_j`\alpha]
\square(0,400)|arlb|<800,400>[A^*_{i+1}`HA^*_i`A^*_{j+1}`HA^*_j;%
{\alpha[a^*_i]}`a^*_{i+1,j+1}`Ha^*_{ij}`]
\morphism(0,800)|l|/{@{>}@/_2em/}/<0,-800>[A^*_{i+1}`A;%
a^*_{i+1}]
\morphism(800,800)|r|/{@{>}@/^2em/}/<0,-800>[HA^*_i`HA;%
Ha^*_i]
\efig
\end{equation}

(ii)~$A^*$ is a well-founded coalgebra: for every ordinal
number the outside square of the diagram
\[
\bfig
\square|arlb|<800,400>[A^*`HA^*`A`HA;%
\alpha{[a^*]}`a^*`Ha^*`\alpha]
\square(0,400)|arlb|<800,400>[A^*_{i+1}`HA^*_i`A^*`HA^*;%
{\alpha[a^*_i]}```]
\morphism(0,800)|l|/{@{>}@/_2em/}/<0,-800>[A^*_{i+1}`A;%
a^*_{i+1}]
\morphism(800,800)|r|/{@{>}@/^2em/}/<0,-800>[HA^*_i`HA;%
Ha^*_i]
\efig
\]
is a pullback, thus so is the upper square. This shows that
$(A^*)_i^* = A^*_i$.     Therefore $(A^*)^*=A^*$ 
by  (\ref{starsfso}).

(iii)
Suppose we are given a well-founded coalgebra $\beta: B \to
HB$ and
a coalgebra homomorphism $f\colon B \to A$. Since $a^*$~is a
monomorphism there is
at most one coalgebra homomorphism  $\bar{f}\colon B \to A^*$ with
$a^* \tec \bar{f} =f$.
Thus, we are finished if we show that $\bar{f}$~exists.
To this end denote by $b^*_{i,j}\colon B^*_i \to B^*_j$
the chain whose colimit is $B^* = B$ with the
colimit
  injections $b^*_i\colon B^*_i \to B^*$. We define 
  the components of a
  natural transformation $\bar{f}_i\colon B_i^* \to A_i^*$
  by transfinite recursion on 
ordinals $i$,    satisfying
\begin{equation}
\label{diag:sq}
\bfig
\square<700,400>[B^*_i`A^*_i`B`A;\bar{f}_i`b^*_i`a^*_i`f]
\efig
\end{equation}
The naturality follows
due to the commutativity from the diagram below
for $i \leq j$:
\begin{equation}
\label{diag:sq2}
\bfig
\square|arlb|<800,400>[B^*_j`A^*_j`B`A;%
\bar{f}_j`b^*_j`a^*_j`f]
\square(0,400)|arlb|<800,400>[B^*_i`A^*_i`B^*_j`A^*_j;%
\bar{f}_i`b^*_{i,j}`a^*_{i,j}`]
\morphism(0,800)|l|/{@{>}@/_2em/}/<0,-800>[B^*_i`B;%
b^*_i]
\morphism(800,800)|r|/{@{>}@/^2em/}/<0,-800>[A^*_i`A;%
a^*_i]
\efig
\end{equation}

The desired upper square commutes because all other parts and the
outside   square do and since $a^*_j$ is a monomorphism.

We define~$\bar{f}_i$ by transfinite recursion.
Let $\bar{f}_0 = \id\colon \emptyset \to
\emptyset$. Then~\eqref{diag:sq} clearly commutes for $i = 0$. For isolated
steps
consider the diagram below:
\begin{equation}
\label{diag:sq3}
\bfig
\square(600,400)<800,400>[A^*_{i+1}`HA^*_i`A`HA;%
{\alpha[a^*_i]}`a^*_{i+1}`Ha^*_i`\alpha]
\morphism(0,1200)<2000,0>[B^*_{i+1}`HB^*_i;{\beta[b^*_i]}]
\morphism(0,1200)/-->/<600,-400>[B^*_{i+1}`A^*_{i+1};\bar{f}_{i+1}]
\morphism(0,1200)|l|<0,-1200>[B^*_{i+1}`B;b^*_{i+1}]
\morphism(2000,1200)|l|<-600,-400>[HB^*_i`HA^*_i;H\bar{f}_i]
\morphism(2000,1200)|r|<0,-1200>[HB^*_i`HB;Hb^*_i]
\morphism(0,0)<600,400>[B`A;f]
\morphism(0,0)|b|<2000,0>[B`HB;\beta]
\morphism(2000,0)|a|<-600,400>[HB`HA;Hf]
\efig
\end{equation}
The inner and outside squares commute by the definition of $A^*_{i+1}$
and
$B^*_{i+1}$, respectively. For the lower square we use that $f$ is a
coalgebra homomorphism, and the right-hand one commutes by the
induction
hypothesis. The inner pullback induces the desired morphism
$\bar{f}_{i+1}$
and the commutativity of the left-hand square is that
of~\eqref{diag:sq} for~${i+1}$.
Finally, for a limit ordinal $j$ let $\bar{f}_j = \colim_{i<j}
\bar{f}_i$, in other words, $\bar{f}_j$ is the unique morphism such
that the squares
\[\bfig
\square/->`<-`<-`->/<700,400>[B^*_j`A^*_j`B^*_i`A^*_i;%
\bar{f}_j`b^*_{i,j}`a^*_{i,j}`\bar{f}_i]
\efig\]
commute for all $i < j$. We need to verify that~\eqref{diag:sq}
commutes for
$\bar{f}_j$. This is the commutativity of the lower square
in~\eqref{diag:sq2}, which follows since all other parts and the
outside square commute for all $i<j$.

To complete the proof consider any ordinal~$i$ such that $B^*_i =
B^* = B$
and $A^*_i = A^*$ hold. Then $\bar{f} = \bar{f}_i\colon B \to A^*$ is
a coalgebra
homomorphism with $a^*_i \tec \bar{f} = f$ by the commutativity of the
upper and left-hand part of Diagram~\eqref{diag:sq3}.
}%
\end{proof}

For endofunctors preserving inverse images the following corollary is
Exercise~VI.16 in~\cite{Ta2}:

\begin{corollary}
\label{L-colim}
Assuming that $H$ preserves monomorphisms, the subcategory of $\Coalg H$
consisting of the well-founded coalgebras is closed under quotients
and coproducts in $\Coalg H\!$.
\end{corollary}

This follows from a general result on coreflective subcategories: the category $\Coalg H$ has a
(strong epi, mono)-factorization system (see Remark \ref{R-fact}),
and its full subcategory of well-founded coalgebras is coreflective with monomorphic coreflections
(see Proposition~\ref{P-starr}).
Consequently, it is closed under quotients and colimits.

We also have the following fact which will be used in Section~3. 

\begin{lemma} \label{wfdsub}
  If $H$ preserves finite intersections, then every subcoalgebra  of a 
well-founded coalgebra  is well-founded.
\end{lemma}
\begin{proof} 
  Given a subcoalgebra $f: (B, \beta) \to (A,\alpha)$ we prove that
  the natural transformation $\bar f_i: B_i^* \to A_i^*$
  of~\refeq{diag:nat} makes the squares in~\refeq{diag:sq} pullbacks
  for every ordinal number $i$. The base case $i = 0$ is clear. For
  the isolated step we use that $a_{i+1}^*: A_{i+1}^* \to A$ is the
  pullback of $Ha_i^*$ along $\alpha$. Thus, it suffices to show that
  $b_{i+1}^*\colon B_{i+1} \to B$ is a pullback of $Ha_i^*$ along $\alpha
  \tec f$. But, since $\alpha \tec f = Hf \tec \beta$ and since $H$
  preserves finite intersections (i.e., pullbacks of monos along
  monos), the latter pullback can be obtained by pasting two pullback
  squares as displayed below:
  \[
  \xymatrix{
    B_{i+1}^*
    \ar[d]_{b_{i+1}^*}
    \ar[r]^-{\beta[b_i^*]}
    \POS (130,-30) = "a"
    \POS (130,-100) = "b"
    \POS (60,-100)  = "c"
    \ar@{-} "a";"b"
    \ar@{-} "c";"b"
    &
    HB_i^*
    \ar[d]_{Hb_i^*}
    \ar[r]^-{H\bar f_i}
    \POS (620,-30) = "a1"
    \POS (620,-100) = "b1"
    \POS (550,-100)  = "c1"
    \ar@{-} "a1";"b1"
    \ar@{-} "c1";"b1"
    &
    HA_i^*
    \ar[d]^{Ha_i^*}
    \\
    B
    \ar[r]_-{\beta}
    &
    HB
    \ar[r]_-{Hf}
    &
    HA
    }
  \]

  Now assume that $(A, \alpha)$ is well-founded, i.\,e., some $a_i^*$
  is invertible. Then its pullback $b_i^*$ along $f$ is invertible,
  i.\,e., $(B, \beta)$ is well-founded.
\end{proof}

\begin{remark}
\label{R-Ta}
If $H$ is a set functor which also preserves inverse images, a much stronger result holds,
as proved in~\cite[Corollary 6.3.6]{Ta2}: every coalgebra from which a homomorphism into a
well-founded coalgebra exists is well-founded.
\end{remark}

\begin{example}
  Without the assumption that $H$ preserves finite intersections the
   lemma above can fail to be true. On the category $\Gra$ of graphs the
  functor $H$ of Example~\ref{E-graphh} has the well-founded coalgebra
  $1 = H1$ which has the subcoalgebra
  \[
  \xymatrix{
    *+[F-]{\bullet} 
    \ar@{x-y}[r]
    &
    *+[F-]{\bullet \quad \bullet t}
    }
  \]
  which is not well-founded: its subcoalgebra $\emptyset \to \{t\}$ is
  cartesian. 
\end{example}

\takeout{For more general results, see Theorem 16.8 and Corollary 13.20 of~\cite{AHS}.}
\takeout{%
\begin{example}
\label{E-init}
The initial coalgebra $0\to H0$ is well-founded.
\end{example}
}%

\takeout{
We present  a number of consequences of this, culminating in Lemma~\ref{L-colim} 
below.  This last fact will be used repeatedly in what follows.

 We again recall from~\cite{AR} that the assumption that $\A$ be LFP implies
 that all morphisms have (epi, strong mono)-factorizations.
}

\subsection{Recursive coalgebras}
\label{sec:recurs}\hfill

\par\smallskip\noindent
Here we recall the notion of recursive coalgebra in order to use it for our proof that initial algebras are the same as final well-founded coalgebras. ``Recursive" and ``well-founded" are closely related concepts. But whereas final recursive coalgebras are already known to be initial algebras, see \cite{CUV}, for well-founded coalgebras this is new (and a bit more involved).

\takeout{%
Osius~\cite{O} introduced the concept of a \emph{recursive coalgebra} in the special 
case where $H=\PP$.
Taylor~\cite{Ta2} studied the concept  in 
more generality, allowing  $H$ to be a functor preserving
inverse images (without using the name). 
Using different terminology, he  proved that under these
circumstances ``recursive'' is equivalent to ``well-founded''. The general 
concept of recursive coalgebra 
that we are using was first introduced by
Capretta et al.~\cite{CUV}.  }
\takeout{They hold
\lmnote{it's ok with me to drop this long quote}
 ``that recursive coalgebras
are a more useful tool in the study of structured recursion than initiality and that most results for
structured recursion and initial algebras can be recast in a clearer way in this more general framework.
In particular, the recursive coalgebras approach makes it plain that structured recursion is
not necessarily tied to initial algebras, essentially the same phenomenon can occur also in situations
when there is none around.''
}%

\begin{definition}
\label{D-rec}
A coalgebra $\alpha\colon A\to HA$ is  \textbf{recursive} if
for every algebra $\beta\colon HB\to B$ there exists a unique
coalgebra-to-algebra homomorphism
\[\bfig
\square/->`->`->`<-/<700,400>[A`HA`B`HB;\alpha`h`Hh`\beta]
\efig\]
\end{definition}

This concept was introduced by Taylor under the name ``coalgebra
obeying the recursion scheme'', the name recursive coalgebra stems
from~Capretta et al.~\cite{CUV}. 

\begin{examples}[see \cite{CUV}]
\label{E-rec}\hfill
\begin{enumerate}[(1)]
\item
  $0\to H0$ is a recursive coalgebra.

\item
  If $\alpha\colon A\to HA$ is recursive, then so is $H\alpha\colon
  HA\to HHA$.

\item
  A colimit of recursive coalgebras is recursive. Combining these
  results we see that in the initial chain (\ref{R-chain}) all
  the coalgebras
  \[w_{i,i+1}\colon H^i 0\to HH^i 0\]
  are recursive.
\end{enumerate}
\end{examples}

We are going to prove that for set functors, well-founded coalgebras
are recursive.  Before we do this, let us discuss the converse.  In
general, recursive coalgebras need not be well-founded, even for set
functors.  However for all set functors preserving inverse images
recursiveness is equivalent to well-foundedness, as shown by
Taylor~\cite{Ta1,Ta2}.

\begin{example}[see \cite{ALM}]
\label{E-Luecke}
A recursive coalgebra need not be well-founded. Let
$H\colon\Set\to\Set$ be defined on objects by
\[HX=(X\times X \setminus \Delta_X) + \{d\}\]
where  $\Delta_X$~denotes the diagonal of $X$.
  For
morphisms $f\colon X\to Y$ we take $Hf(d)=d$ and
\[
Hf(x_1,x_2)=\begin{cases}
d&\text{if $f(x_1)=f(x_2)$}\\
(fx_1,fx_2)&\text{else}
\end{cases}
\]
This functor $H$ preserves monomorphisms.
The coalgebra $A=\{0,1\}$ with the structure $\alpha$ constant
to $(0,1)$ is recursive: given an algebra $\beta\colon HB\to B$, the
unique coalgebra-to-algebra homomorphism $h\colon\{0,1\}\to B$ is
\[h(0)=h(1)=\beta(d).\]
But $A$~is not well-founded: $\emptyset$~is a cartesian subcoalgebra.
\takeout{%
\[\bfig
\square<800,400>[\emptyset`\{d\}`\{0,1\}`H\{0,1\};%
```\alpha]
\POS(100,300)
\ar@{-}+(-50,0)
\ar@{-}+(0,50)
\efig\]
}%
\end{example}

\begin{theorem}
\label{T-rec}
If $H$ preserves monomorphisms, then every well-founded coalgebra is recursive.
\end{theorem}

For functors preserving inverse images this follows from~\cite[Theorem~6.3.13]{Ta1}.

\begin{proof}  
Let $\alpha\colon A\to HA$ be well-founded. For every algebra
$e\colon HX\to X$ we prove the existence and uniqueness of a
coalgebra-to-algebra homomorphism $A\to X$.  We use the 
initial chain $(H^i 0)$
of (\ref{R-chain})
and also the chain $(A_i^*)$ from Notation~\ref{N-starr}.

(1) Existence. We prove first that there is a unique natural
transformation
\[f_i\colon A_i^*\to H^i 0\qquad(i\in\Ord)\]
such that for all ordinals~$i$ we have
\begin{equation}
\label{nrdjd}
f_{i+1} = \left( A_{i+1}^*\nsi{\alpha[a_i^*]} HA_i^*\nsi{Hf_i}
H(H^i 0)=H^{i+1} 0\right).
\end{equation}
In fact, since both of the transfinite chains $(A_i^*)$ and~$(H^i 0)$
are defined by colimits on all limit ordinals~$i$, and $f_0$ must be $\id_{\emptyset}$, we only need to
check the commutativity of the square
\begin{equation}
\label{nrdjtc}
\bfig
\square<700,400>[A^*_i`H^i 0`A^*_{i+1}`H^{i+1} 0;%
f_i`a^*_{i,i+1}`w_{i,i+1}`f_{i+1}]
\efig
\end{equation}
for every successor ordinal~$i$.  For this, the diagram below commutes by the 
induction hypothesis~\eqref{nrdjtc} and by the
commutativity of the upper inner square of~\eqref{nrdjt} in Remark \ref{R-connect}:
\[
\xymatrix@C+2pc{
  A_{+i1}^*
  \ar[d]_{a_{i+1,i+2}^*}
  \ar[r]^-{\alpha[a_i^*]}
  &
  HA_i^*
  \ar[r]^-{Hf_i}
  \ar[d]^{Ha_{i,i+1}^*}
  &
  H(H^i0)
  \ar[d]^{Hw_{i,i+1}}
  \ar@{<-} `u[l] `[ll]_{f_{i+1}} [ll]
  \\
  A_{i+2}^*
  \ar[r]_{\alpha[a_{i+1}^*]}
  &
  HA_{i+1}^*
  \ar[r]_-{Hf_{i+1}}
  &
  H(H^{i+1} 0)
  \ar@{<-} `d[l] `[ll]^-{f_{i+2}} [ll]
}
\]
\takeout{
\[\bfig
\square<800,400>[A^*_{i+1}`HA^*_i`A^*_{i+2}`HA^*_{i+1};%
{\alpha[a^*_i]}`a^*_{i+1,i+2}`Ha^*_{i,i+1}`{\alpha[a^*_{i+1}]}]
\square(800,0)<800,400>[HA^*_i`H(H^i 0)`HA^*_{i+1}`H(H^{i+1} 0);%
Hf_i``Hw_{i,i+1}`Hf_{i+1}]
\efig\]
}

Next, since the $H^i 0$~are recursive coalgebras (see
Example~\ref{E-rec}) we have unique coalge\-bra-to-algebra homomorphisms
into $X$.  These form a natural transformation into the constant
functor with value~$X$:
\[r_i\colon H^i 0\to X\qquad(i\in\Ord).\]
Consequently, we obtain a natural transformation $r_if_i\colon
A_i^*\to X$ which, for~$i$ such that $A^{*}_i=A^*_{i+1}$ (thus,
$A=A_i^*$), yields
\[h=r_if_i\colon A\to X\enspace.\]
Now consider the diagram below:
\[\bfig
\square(0,0)|arlb|/->`->`->`<-/<900,350>[H^i 0`H(H^i 0)`X`HX;w_{i,i+1}`r_i`Hr_i`e]
\square(0,350)|arlb|<900,350>[A^*_i=A^*_{i+1}`HA^*_i`H^i 0`H(H^i 0);%
{\alpha[a^*_i]}`f_i`Hf_i`]
\morphism(0,700)<900,-350>[A^*_i=A^*_{i+1}`H(H^i 0);f_{i+1}]
\morphism(0,700)|l|/{@{>}@/_1.5em/}/<0,-700>[A^*_i=A^*_{i+1}`X;h]
\morphism(900,700)|r|/{@{>}@/^2em/}/<0,-700>[HA^*_i`HX;Hh]
\efig\]
  The morphism at the top  is $\alpha^{*}$, by  (\ref{starsfso2}).
The sides are the definition of $h$,  the bottom square is the definition of $r_i$,
and the upper right-hand triangle is the definition of $f_{i+1}$.
The upper left-hand triangle is (\ref{nrdjtc}) since $a^*_{i,i+1} = \id$.
The overall outside of the figure shows that $h$ is a
coalgebra-to-algebra homomorphism as desired.

(2) Uniqueness. If $h_1,h_2\colon A\to X$ are  coalgebra-to-algebra
  homomorphisms, then we
prove $h_1=h_2$ by showing that
\[h_1\tec a_i^*=h_2\tec a_i^*\qquad\text{for all $i\in\Ord$.}\]
The case $i=0$ is clear, in the isolated step use the commutative
diagrams (with $t=1,2$):\vspace{1mm}
\[\bfig
\square/->`->`->`<-/<800,400>[A`HA`X`HX;\alpha`h_t`Hh_t`e]
\square(0,400)<800,400>[A^*_{i+1}`HA^*_i`A`HA;
{\alpha[a^*_i]}`a^*_{i+1}`Ha^*_i`]
\efig\]
and the limit steps follow from $A_j^*=\colim_{i<j}A_i^*$ for limit
ordinals~$j$.
\end{proof}

\begin{example}\label{detect}  
There is a $\pow$-algebra   $(B,\beta)$
such that for all $\pow$-coalgebras $(A,\alpha)$,  if
$(A,\alpha)$ is not well-founded, then there are at least two
coalgebra-to-algebra homomorphisms $h: A\to B$.

We take $B = \set{0,1,2}$, with 
$\beta:\pow B \to B$ defined as follows:
$$\beta(x) =  \left\{ \begin{array}{l@{\quad}l}
0 & \mbox{if $x = \emptyset$ or $x = \set{0}$}  \\
1 & \mbox{else if $1\in x$} \\
2 & \mbox{if $2\in x$ and $1\notin x$}
\end{array}
\right.
$$
If $(A,\alpha)$ is any coalgebra which is not well-founded, we show that
there are at least two coalgebra-to-algebra homomorphisms $h: A\to B$.
We can take
$$h_1(x) = \left\{
\begin{array}{l@{\quad}l}
0  & \mbox{if there are no infinite sequences
$x = x_0 \rightarrow x_1 \rightarrow x_2 \cdots$} \\
1  & \mbox{if there is an infinite sequence
$x = x_0 \rightarrow x_1 \rightarrow x_2 \cdots$} \\
\end{array}
\right.
$$
and also $h_2$ defined the same way, but using $2$ as a value instead of $1$.
The verification that $h_1$ and $h_2$ are coalgebra-to-algebra homomorphisms
hinges on two facts: first,
$h(x) = 0$ iff
 there is no infinite sequence starting from $x$;
and second,   if $h_i(x) \neq 0$, then there is some $y\in \alpha(x)$ such that
$h_i(y) \neq 0$ as well.
\end{example}

For endofunctors preserving inverse images the following theorem is
Corollary 9.9 of~\cite{Ta1}. As we mentioned in the introduction, it
is non-trivial to relax the assumption on the endofunctor, and so our
proof is different from Taylor's. As a result we obtain in
Theorem~\ref{T-main-set} below that for a set endofunctor no assumptions are needed.

\begin{remark}
\label{R-CUV}
The concepts ``initial algebra'' and ``final recursive coalgebra''
coincide for all endofunctors, as proved by Capretta et
al.~\cite{CUV}. This is not true in general for well-foundedness in
lieu of recursiveness, see Example \ref{E-gr4} below. But it is true if $H$ preserves finite intersections:
\end{remark}

\begin{theorem}
\label{T-main}
If $H$~preserves finite intersections, then
\[\text{initial algebra}=\text{final well-founded coalgebra}\]
That is, an algebra $\varphi\colon HI\to I$ is initial iff
$\varphi^{-1}\colon I\to HI$ is the final well-founded coalgebra.
\end{theorem}

\begin{proof}
Recall that since $H$ preserves finite intersections, i.e., pullbacks
of monomorphisms, it preserves monomorphisms (since $m$ is monic iff
the pullback of $m$ along itself is formed by identity morphisms).

(a)
Let $I$~be an initial algebra. 
By  Remark~\ref{R-CUV}, $I$ is
 a final
recursive coalgebra.   Also, $I$~is well-founded by
Proposition~\ref{P-init}.
Thus by Theorem~\ref{T-rec},  it is a final well-founded coalgebra.

(b)
Let $\psi:I\to HI$ be a final well-founded coalgebra.

(b1)
Factorize $\psi=m\tec e$ where $e$~is  a strong epimorphism and $m$ a 
monomorphism (Remark~\ref{R-fact}). By diagonal fill-in
\[\bfig
\square/-->`->`->`->/<800,400>[I'`HI'`HI`HHI;\psi'`m`Hm`H\psi]
\square(0,400)/->`->`->`-->/<800,400>[I`HI`I'`HI';\psi`e`He`]
\efig\]
we obtain a quotient~$(I',\psi')$ which, by Corollary~\ref{L-colim},
is well-founded. Consequently, a coalgebra
homomorphism $f\colon(I',\psi')\to(I,\psi)$ exists. Then $fe$~is
an endomorphism of the final well-founded coalgebra, hence,
$fe=\id_I$. This proves that $e$~is an isomorphism,  in other words
\[\text{$\psi$ is a monomorphism.}\]

(b2)
The coalgebra~$(HI,H\psi)$ is well-founded. Indeed, consider a
cartesian subcoalgebra~$(A',a')$
\[\bfig
\square/-->`-->`->`-->/<800,400>[J`A'`I`HI;\psi'`m'``\psi]
\square(800,0)<800,400>[A'`HA'`HI`HHI;a'`m`Hm`H\psi]
\POS(100,300)
\ar@{-}+(-50,0)
\ar@{-}+(0,50)
\POS(900,300)
\ar@{-}+(-50,0)
\ar@{-}+(0,50)
\efig\]
Form the intersection~$J$ of $m$ and~$\psi$. Since $H$~preserves
this intersection, it follows that $m$ and~$Hm'$ represent the same
subobject of~$HI$, thus, we have
\[\quad u\colon A'\to HJ, \quad\text{ with}\quad m=Hm'\tec
u.\]
This yields a cartesian subcoalgebra
\[\bfig
\morphism|b|<1400,0>[I`HI;\psi]
\morphism(0,400)|l|<0,-400>[J`I;m']
\morphism(0,400)<700,0>[J`A';\psi']
\morphism(700,400)<700,0>[A'`HJ;u]
\morphism(1400,400)|r|<0,-400>[HJ`HI;Hm']
\POS(100,300)
\ar@{-}+(-50,0)
\ar@{-}+(0,50)
\efig\]
and since $(I,\psi)$~is well-founded, we conclude that $m'$~is
invertible. Consequently, $m=Hm'\tec u$ is invertible.

(b3)
$\psi$~is invertible. Indeed, we have, by~(b2), a homomorphism
$h\colon(HI,H\psi)\to(I,\psi)$:
\[\bfig
\square<800,400>[HI`HHI`I`HI;H\psi`h`Hh`\psi]
\square(0,400)<800,400>[I`HI`HI`HHI;\psi`\psi`H\psi`]
\efig\]
Then $h\tec\psi$~is an endomorphism of~$(I,\psi)$, thus,
$h\tec\psi=\id$. And the lower square yields $\psi\tec h=
H(h\tec\psi)=\id$, whence $I \cong HI$, 

(b4)~By
Proposition~\ref{P-chain},
 the initial chain converges, and
$w^{-1}_{i,i+1}\colon HH^i 0\to H^i 0$ is an initial algebra for some
ordinal $i$. Moreover, $w_{i,i+1}\colon H^i 0\to HH^i 0$ is by~(a) a
final well-founded coalgebra, thus, isomorphic to $\psi\colon I\to
HI$.  Therefore $(I,\psi^{-1})$~is isomorphic to the initial algebra
above.
\end{proof}

\begin{example}
\label{E-gr4}\hfill
\begin{enumerate}[(a)]
\item For the identity functor on the category of rings
  the initial algebra is $\mathbb{Z}$ and the terminal well-founded coalgebra is $1$.
  This shows the importance of our assumption that the base category have a simple initial object.
        
\item Also the assumption that $H$ preserves finite intersections is
  important: The endofunctor $H$ of~$\Gra$ in  Example~\ref{E-graphh} has $1$ as
  its final well-founded coalgebra, and its initial algebra is
  infinite.
\end{enumerate}
\end{example}


\begin{remark}
\label{R-classM}

Although we have previously worked  with monomorphisms only, the whole theory can be developed for a general class 
$\M$ of monomorphisms in the base category $\A$.
 We need to assume that
 \begin{enumerate}[(a)]
 \item $\A$ is $\M$-wellpowered,
 \item $\M$ is closed under inverse images, and
 \item $\M$ is constructive in the sense of~\cite{TAKR}.  
 \end{enumerate}
The last point means that
 $\M$ is closed under composition, and for every
chain of monomorphisms in $\M$, (i) a colimit exists and is formed by monomorphisms in $\M$,
and (ii) the factorization morphism of every cocone of monomorphisms in $\M$ is again
a monomorphism in $\M$. This in particular states  that the
initial object has the property that all morphisms $0\rightarrow X$ lie in $\M$.

We then can define \emph{$\M$-well-founded coalgebra}
as one that has no proper cartesian subcoalgebra carried by an $\M$-monomorphism.

All results above hold in this generality.      In Theorem~\ref{T-rec}
 we must assume that $H$ preserves $\M$, that is, if $m$ lies
in $\M$ then so does $Hm$.       In Theorem~\ref{T-main} we need to assume that
$H$ preserves $\M$ and finite intersections of $\M$-monomorphisms.
\end{remark}

\begin{example} All LFP categories with simple initial object satisfy all the assumptions of \ref{ASS} for
$$\M=\mbox{strong monomorphisms},$$
see \cite[Propositions 1.61 and 1.62]{AR}. 
\end{example}

\begin{example}\label{graphh2}
Here we compare well-foundedness w.r.t to monomorpisms to that w.r.t. strong monomorphisms. 
Take again the category
  $\Gra$ of graphs  and graph morphisms and $H$ be the following
  endofunctor: The nodes of $HA$ are all finite independent sets
    $a\subseteq A$ (i.e., no edge lies in $a$) plus a new node
    $t$. The  coalgebra structure is the constant map to $\{t\}$, i.e., the only edges of $HA$ connect every node to $t$ ($t$ is a loop).
For a graph morphism $f: A\to B$, we take $Hf: HA \to HB$ to be 
$$
Hf(a)  = \begin{cases}
f[a] & \mbox{if $f[a]$ is independent in $B$}  \\
t & \mbox{otherwise} \\
\end{cases}
$$
This functor clearly preserves strong monomorphisms (but not monomorphisms).

By Theorem~\ref{T-main}, the initial algebra for $H$ is the same as
its final $\M$-well-founded coalgebra.
This is 
$$ I = I_0\cup\set{t}$$ where $ I_0 = \pow_f  I_0$ is 
the initial algebra of the finite power set functor on $\Set$,
taken as a discrete graph, 
and the coalgebra structure is the constant to $\{t\}$.

In contrast, $I$ is not well-founded (w.~r.~t.~\emph{all} monomorphisms).
Here is the reason.   Let $J$ be the same as $I$, except that we drop
all edges between $t$ and the elements of $I_0$.  (We keep just the loop
at $t$.)    Then $HJ = HI = I$.     The inclusion $i: J \to I$ is of course monic,
and $Hi = \id_{HI}$.   It is easy to check that this inclusion is a 
coalgebra morphism, and indeed this subcoalgebra is clearly cartesian.   This verifies that
$I$ is not well-founded.%
\end{example}

\subsection{Initial algebras of set functors}\label{Ini-alg}
\hfill

\smallskip\noindent
The main result of this section is that for all endofunctors 
$H$ of~\Set\ the
equality
\begin{equation}
\text{initial algebra}=\text{final well-founded coalgebra}
\label{01}
\end{equation}   
holds, i.\,e., for the particular case of our given LFP category being
$\A = \Set$ one can lift the assumption that $H$ preserves finite intersections in Theorem~\ref{T-main}. 
   
 \takeout{ A set functor $H$ need not preserve monomorphisms.
    Although the definitions of ``initial algebra'' and 	``well-founded coalgebra''
    make sense for $H$, essentially none of the results 
    which we have seen up to this point apply to $H$.
 However, we can apply the following
result of Trnkov\'a in order to obtain our desired result (\ref{01}):
}

\begin{proposition}[Trnkov\'a~\cite{Tr}]
\label{R-Tr}
For every endofunctor~$H$ of~\Set\ there exists an
endofunctor~$\bar{H}$ preserving finite
intersections and identical with~$H$ on all nonempty sets (and
nonempty functions).
\end{proposition}

\begin{remark}\label{R-hull}  The functor  $\bar{H}$~is unique up to natural isomorphism. We call it  \textit{the Trnkov\'{a} closure of $H$}. Let us recall how Trnkov\'a
defined~$\bar{H}$:

Denote by~$C_{01}$ the set functor $\emptyset\to/|->/\emptyset$ and
$X\to/|->/1$ for all ${X\neq\emptyset}$. Define~$\bar{H}$
as~$H$ on all
nonempty sets, and put
\[\bar{H}\emptyset=\{\tau;\tau\colon C_{01}\to H\text{ a natural
transformation}\}.\]
(To check that we have a set here and not a proper class,
note that each $\tau : C_{0,1}\to H$ is determined by 
  $\tau_1: 1\to H1$.   For a nonempty set $A$, if $k: 1\to A$ is arbitrary,
  $\tau_A = Hk\circ \tau_1$.)
\takeout{$\tau_1: 1 \to H1$, since 
for non-empty $X$, $\tau_X: 1\to HX$ takes the value $Hk\circ \tau_1$, where
$k: 1\to X$ is arbitrary.)}
Given a nonempty set~$X$, $\bar{H}$~assigns to
the empty map $q_X\colon\emptyset\to X$
the map
\[\bar{H}q_X\colon\tau\to/|->/\tau_X\qquad\text{for every
\quad$\tau\colon C_{01}\to H$,}\]
where $\tau_X\colon 1\to HX$ is simply an element of~$HX$.

Observe that there exists a map
$u\colon H\emptyset\to\bar{H}\emptyset$ such that for every set
$A\neq\emptyset$ the triangle
\begin{equation}
\label{rdo}
\bfig
\Vtriangle<500,400>[H\emptyset`\bar{H}\emptyset`{HA=\bar{H}A};u`Hq_A`\bar{H}q_A]
\efig
\end{equation}
commutes. 
Indeed, for each element $x\in H\emptyset$,   let the
natural
transformation $u(x)\colon C_{01}\to H$  have components
$u(x)_A=Hq_A(x)$ for all $A\neq\emptyset$. Then
$$\bar{H}q_A(u(x))=  (u(x))_A  =     Hq_A(x).$$

\end{remark}

\begin{lemma}
\label{810}
Let $(A,a)$ be a well-founded $H$-coalgebra with $A\neq \emptyset$,
so that $(A,a)$ is also an $\bar{H}$-coalgebra.
Then $\emptyset$ is not the carrier of any cartesian $\bar{H}$-subcoalgebra
of $(A,a)$.
\end{lemma}

\begin{proof}
Assume towards a contradiction that $q_{\bar{H}\emptyset}: \emptyset
\to \bar H \emptyset$ were 
a cartesian subcoalgebra of $(A, a)$.
We claim that the square below is a pullback:
\begin{equation}\label{811}
\bfig
\square<700,400>[\emptyset`H\emptyset`A`HA;%
q_{H\emptyset}`q_A`Hq_{A}`a]
\POS(100,300)
\ar@{-}+(-50,0)
\ar@{-}+(0,50)
\efig
\end{equation}
We show that there are no
$y\in A$ and $x\in H\emptyset$ such that
that $a(y)  = Hq_A(x)$.  
For assume that $y$ and $x$ exist with these properties.
 Then by (\ref{rdo}), $\bar{H}q_A (u(x)) = a(y)$.
This contradicts our assumption that 
$(\emptyset, q_{\bar{H}\emptyset})$ is 
a cartesian subcoalgebra of $(A, a)$. Thus, $y$ and $x$ do not exist
as assumed, and hence, the square in (\ref{811}) is indeed a
pullback. Therefore $q_A$ is an isomorphism. But $A\neq\emptyset$, and this is a contradiction.
\end{proof}

\begin{theorem}
\label{T-main-set}
For every endofunctor on~\Set\ we have:
\[\text{initial algebra}=\text{final well-founded coalgebra.}\]
\end{theorem}

\begin{proof}
Given~$H$, we know from Theorem~\ref{T-main} that the statement
holds for the Trnkov\'{a} closure~$\bar{H}$. From this we are
going to prove it for~$H$. 

(a)
If $\varphi\colon HI\to I$ is an initial algebra, we prove that
$\varphi^{-1}\colon I\to HI$ is a final well-founded coalgebra.

 This is clear when $H\emptyset=\emptyset$.
In this case $I=\emptyset$. And the only
(hence, the
final) well-founded coalgebra is the empty one: if
 $a\colon
A\to HA$ is well-founded, then the following cartesian subcoalgebra
\begin{equation}
\label{bb}
\bfig
\square<700,400>[\emptyset`\emptyset`A`HA;%
\id`q_A`Hq_A`a]
\POS(100,300)
\ar@{-}+(-50,0)
\ar@{-}+(0,50)
\efig\end{equation}
demonstrates that $q_A$ is an isomorphism, so $A = \emptyset$.

Thus we assume $H\emptyset\neq\emptyset$.
Then $\bar{H}\emptyset\neq\emptyset$
via $u$ in \eqref{rdo} above. The $\bar{H}$-algebra
$\varphi\colon\bar{H}I\to I$ is initial because every
$\bar{H}$-algebra is nonempty, hence, it also is an $H$-algebra. And
the unique homomorphism from~$I$ w.r.t.~$H$ is also a homomorphism
w.r.t.~$\bar{H}$. By Theorem~\ref{T-main},
$\varphi^{-1}\colon I\to\bar{H}I$ is a final well-founded
$\bar{H}$-coalgebra. Let us now verify that it is also well-founded
w.r.t.~$H$.   Consider a cartesian subcoalgebra
\begin{equation} 
\bfig
\square<700,400>[A'`HA'`I`HI;%
a'`m`Hm`\varphi^{-1}]
\POS(100,300)
\ar@{-}+(-50,0)
\ar@{-}+(0,50)
\efig
\label{cc}
\end{equation}
We claim that $A'$ cannot be empty.   For if it were, then since $HA' = H\emptyset\neq \emptyset$,
we take any $x\in HA'$ and consider $x$ and $(\varphi\cdot Hm)(x)$.
By the pullback property, there is some $y\in A'$  so that $a'(y) = x$.
This  contradicts $A'  = \emptyset$.

As a result, $HA'  = \bar{H}A'$, and $Hm = \bar{H}m$.   So 
(\ref{cc}) is a cartesian subcoalgebra for $\bar{H}$.  Thus $m$ is invertible, as desired.

At this point we know that $\varphi^{-1}\colon I\to HI$ is a well-founded
$H$-coalgebra;
we conclude with the verification that $\varphi^{-1}$ is \emph{final} 
among these.    This follows from the
observation that every nonempty well-founded $H$-coalgebra $a\colon
A\to HA$ is also well-founded w.r.t.~$\bar{H}$. Indeed, consider a
cartesian subcoalgebra
\begin{equation}\bfig
\square<700,400>[A'`\bar{H}A'`A`\bar{H}A;%
a'`m`\bar{H}m`a]
\POS(100,300)
\ar@{-}+(-50,0)
\ar@{-}+(0,50)
\efig\label{cdc}
\end{equation}
By Lemma~\ref{810},
 $A'\neq\emptyset$.   Thus $\bar{H}m=Hm$ and we conclude that $m$~is
invertible.
\takeout{
 If $A'=\emptyset$, then the square in (\ref{cdc})
is the same as the square in (\ref{bb}).   And so  
by well-foundedness of $(A,a)$, we  again 
have that $m = q_A$ is an isomorphism, 
contradicting the non-emptiness of $A$.
}
\takeout{OLD ARGUMENT BELOW
$m=q_A$, and since we have seen
in~\eqref{rdo} above that $\bar{H}q_A\tec u=Hq_A$, from the last
pullback we get the following one
\[\bfig
\square<700,400>[\emptyset`H\emptyset`A`HA;%
`q_A`Hq_A`a]
\POS(100,300)
\ar@{-}+(-50,0)
\ar@{-}+(0,50)
\efig\]
in contradiction to the well-foundedness of the $H$-algebra~$A$.
}

(b)
If $\psi\colon I\to HI$ is a final well-founded coalgebra, we prove that
$\psi$~is invertible and $\psi^{-1}\colon HI\to I$ is an initial
algebra. Unfortunately, we cannot use the converse implication of
what we have just proved (every nonempty well-founded
$\bar{H}$-coalgebra is also well-founded w.r.t.~$H$) since this is
false in general (see Example~\ref{ddo} below). We can assume
$H\emptyset\neq\emptyset$, since the case $H\emptyset=\emptyset$ is
trivial.

Consider first the coalgebra
\[b\colon\bar{H}\emptyset\to H\bar{H}\emptyset\]
defined by
\[b(\tau)=\tau_{\bar{H}\emptyset}\qquad\text{for all $\tau\colon
C_{01}\to H$.}\]
Let us show that
this coalgebra is well-founded for~$H$.  Consider a cartesian subcoalgebra
\begin{equation} \label{e}
\bfig
\square<700,400>[A'`HA'`\bar{H}\emptyset`H\bar{H}\emptyset;%
a'`m`Hm`b]
\POS(100,300)
\ar@{-}+(-50,0)
\ar@{-}+(0,50)
\efig
\end{equation}
It is our task to prove that $m$~is surjective (thus, invertible).
First, assume that $A'\neq \emptyset$.
Given $\tau\colon C_{01}\to H$ in~$\bar{H}\emptyset$, the
element~$\tau_{A'}$ of~$HA'$ 
\takeout{%
(recall that $H\emptyset\neq\emptyset$,
thus $HA'\neq\emptyset$ for all sets~$A'$)
}%
 fulfils
\[b(\tau)=\tau_{\bar{H}\emptyset}=Hm(\tau_{A'})\]
by the naturality of $\tau$ and the fact 
that $C_{01} m = \id_1$.
Thus, there exists an element of~$A'$ that $m$~maps to~$\tau$.
Our second case is when $A' = \emptyset$.
We show that this case leads to a contradiction. Observe that $m =
q_{\bar H \emptyset}\colon \emptyset \to \bar H\emptyset$, and let
 $x\in H\emptyset$, so that $u(x) \in \bar{H}\emptyset$, see (\ref{rdo}), and we have
$$b(u(x))  = (u(x))_{\bar{H}\emptyset}
=  H q_{\bar{H}\emptyset} (x).
$$
Thus $x$ and $u(x)$ are mapped to the same element of $H\bar{H}\emptyset$
by $Hm$ and $b$, respectively, contradicting the assumption that
$\emptyset$ is a pullback in (\ref{e}) above.

The first point of this coalgebra 
$(\bar{H}\emptyset,b)$ is that its well-foundedness and non-emptiness
implies that  the final well-founded $H$-coalgebra $(I,\psi)$ 
must also be  nonempty.
Thus $(I,\psi)$ is also a coalgebra for~$\bar{H}$. Let us prove that it is
well-founded w.r.t.~$\bar{H}$. Given a cartesian subcoalgebra
\[\bfig
\square<700,400>[A'`\bar{H}A'`I`\bar{H}I;%
a'`m`\bar{H}m`\psi]
\POS(100,300)
\ar@{-}+(-50,0)
\ar@{-}+(0,50)
\efig\]
by Lemma~\ref{810},
$A'\neq\emptyset$.   So $\bar{H}m=Hm$, hence $m$~is invertible.
\takeout{And the case $A'=\emptyset$ cannot occur since then $m=q_I$ and we
know from~\eqref{rdo} that $\bar{H}q_I\tec u=Hq_I$, so that we would
get a pullback for~$H$ in place of~$\bar{H}$, in contradiction to the
well-foundedness of~$I$.}

We next prove that $(I,\psi)$ is the final well-founded
$\bar{H}$-coalgebra. Let $a\colon A\to\bar{H}A$ be a nonempty
well-founded $\bar{H}$-coalgebra. We prove that the coproduct
\[(A,a)+(\bar{H}\emptyset,b)\qquad\text{in $\Coalg H$}\]
is a well-founded $H$-coalgebra. This will conclude the proof: we
have a unique homomorphism from that coproduct into~$(I,\psi)$
in~$\Coalg H$, hence, a unique homomorphism from~$(A,a)$
to~$(I,\psi)$. Now in oder to prove that the  coproduct above is a
well-founded $H$-coalgebra we first use that every nonempty well-founded coalgebra
for~$H$ is also well-founded for~$\bar{H}$, thus, both of the 
summands above are well-founded $\bar{H}$-coalgebras. Since coproducts of
coalgebras are formed on the level of sets, the two categories $\Coalg
H$ and~$\Coalg\bar{H}$ have the same formation of coproduct of
nonempty coalgebras. Let
\[(A,a)+(\bar{H}\emptyset,b)=(A+\bar{H}\emptyset,c)\]
be a coproduct in~$\Coalg\bar{H}$, then this coalgebra is well-founded
w.r.t.~$\bar{H}$ by Corollary~\ref{L-colim}. To prove that it is also
well-founded w.r.t.~$H$, we only need to consider the empty
subcoalgebra: we must prove that the square
\[\bfig
\square<1000,400>[\emptyset`H\emptyset`A+\bar{H}\emptyset`
H(A+\bar{H}\emptyset);a'`m`Hm`c]
\efig\]
is not a pullback. Indeed, choose an element $x\in H\emptyset$ and
put $\tau=u(x)$ (see~\eqref{rdo}). Then $m=q_{A+\bar{H}\emptyset}$
implies
\[Hm(x)=\tau_{A+\bar{H}\emptyset}.\]
We also have $\tau\in\bar{H}\emptyset$ and the coproduct injection
$v\colon\bar{H}\emptyset\to A+\bar{H}\emptyset$ fulfils $c\tec
v=Hv\tec b$ (due to the formation of coproducts in~$\Coalg\bar{H}$).
Therefore
\[c\bigl(v(\tau)\bigr)=Hv\bigl(b(\tau)\bigr)=
Hv(\tau_{\bar{H}\emptyset})=\tau_{A+\bar{H}\emptyset}=Hm(x).\]
Since we presented elements of $A+\bar{H}\emptyset$ and~$H\emptyset$
that are mapped to the same element by $c$ and~$Hm$, respectively,
the  square above is not a pullback. This finishes the proof that
$(I,\psi)$~is a final well-founded $\bar{H}$-coalgebra.

By Theorem~\ref{T-main} we conclude that $\psi$~is invertible and
$(I,\psi^{-1})$~is an initial $\bar{H}$-algebra. It is also an
initial $H$-algebra:   due to
$H\emptyset\neq\emptyset\neq\bar{H}\emptyset$, the two functors have
the same categories of algebras.
\end{proof}

\begin{example}
\label{ddo}
Let $H=C_{01}+C_1$ be the constant functor of value~2 except
$\emptyset\mapsto 1$. The functor~$\bar{H}$ in the  proof above is
the constant functor with value~${1+1}$, expressed, say as~$\{a,b\}$.
Here
\[H\emptyset=\{b\}\qquad\text{and}\qquad HA=\{a,b\}\quad\text{for
$A\neq\emptyset$.}\]
The coalgebra
\[\{a\}\to/^{ (}->/<300>\{a,b\}\]
is obviously well-founded w.r.t.~$\bar{H}$ but not w.r.t.~$H$ since we
have the pullback:
\[\bfig
\square/->`->`->`^{ (}->/<900,400>%
[\emptyset`H\emptyset=\{b\}`\{a\}`\{a,b\};%
`q_{\{a\}}`Hq_{\{a\}}`]
\POS(100,300)
\ar@{-}+(-50,0)
\ar@{-}+(0,50)
\efig\]
\end{example}

\subsection{The canonical graph and well-foundedness}

\begin{definition}
\label{D-can}   Let $H$~be a set functor preserving (wide) intersections. 
For
every coalgebra $a\colon A\to HA$ define the \textbf{canonical graph}
on~$A$: the neighbors of $x\in A$ are precisely those elements of~$A$ 
which lie in the least subset $m\colon M\to/^{ (}->/A$ with $a(x)\in
Hm[HM]$.
\end{definition}

\begin{remark}
\label{R-can}
(a)
Gumm observed in~\cite{Gu2} that if $H$~preserves intersections we
obtain a ``subnatural'' transformation from it to the power-set
functor~\PP\ by defining functions
\[\tau_A\colon HA\to<180>\PP A,\quad \tau_A(x)=\text{the least subset
$m\colon M\to/^{ (}->/<180>A$ with $x\in Hm[HM]$.}\]
The naturality squares do not commute in general, but for every
monomorphism $m\colon A'\to/^{ (}->/<200>A$ we have a commutative
square
\[\bfig
\square<800,400>%
[HA'`\PP A'`HA`\PP A;%
\tau_{A'}`Hm`\PP m`\tau_A]
\POS(100,300)
\ar@{-}+(-50,0)
\ar@{-}+(0,50)
\efig\]
which even is a pullback. The canonical graph of a coalgebra $a\colon
A\to HA$ is simply the graph $\tau_A\tec a\colon A\to\PP A$.

(b)
Recall that a graph is well-founded iff it has no infinite directed
paths. This also fully characterizes well-foundedness of $H$-coalgebras:
\end{remark}

\begin{proposition}
\label{P-can}
If a set functor preserves intersections, then a coalgebra is
well-founded iff its canonical graph is well-founded.
\end{proposition}

\begin{remb}
For functors $H$ preserving inverse images this fact is proved by
Taylor~\cite[Remark~6.3.4]{Ta2}.  Our proof is essentially
the same.
\end{remb}

\begin{proof}
Let $a\colon A\to HA$ be a well-founded coalgebra. Given a
subgraph~$(A',a')$ of the associated graph~$(A,\tau\tec a)$ forming a
pullback
\[\bfig
\morphism|b|<500,0>[A`HA;a]
\morphism(500,0)|b|<500,0>[HA`\PP A;\tau_A]
\morphism(0,400)<1000,0>[A'`\PP A';a']
\morphism(0,400)|l|<0,-400>[A'`A;m]
\morphism(1000,400)|r|<0,-400>[\PP A'`\PP A;\PP m]
\POS(100,300)
\ar@{-}+(-50,0)
\ar@{-}+(0,50)
\efig\]
we are to prove that $m$~is invertible. Use the pullback of Remark~\ref{R-can}:
 \begin{equation}\label{tau-pb}\bfig
\hSquares|aallrbb|<400>[A'`HA'`\PP A'`A`HA`\PP A;%
a''`\tau_{A'}`m`Hm`\PP m`a`\tau_A]
\POS(720,300)
\ar@{-}+(-50,0)
\ar@{-}+(0,50)
\efig\end{equation}
We get a unique $a''\colon A'\to HA'$ with $a' = \tau_{A'} \tec a''$, and
$(A', a'')$
 is a subcoalgebra of~$(A,a)$. Moreover, in the  diagram above the
outside square and the right-hand one are both pullbacks, thus, the
left-hand square is also a pullback. Consequently, $m$~is invertible
since $(A,a)$~is well-founded.

Conversely, if the graph~$(A,\tau_A\tec a)$ is well-founded, we are
prove that if the left-hand square of~\eqref{tau-pb} is a   pullback then $m$ is invertible. Indeed, in that case, by composition, the outside square is  a pullback for the subcoalgebra~$(A',\tau_{A'}\tec a'')$
of~$(A,\tau_A\tec a)$. Thus, since the last coalgebra is
well-founded, $m$~is invertible.
\end{proof}

\subsection{Initial algebras for functors on vector spaces}
\label{vectors}\hfill

\smallskip\noindent
For every field $K$,
 the category $\KVec$ of vector spaces over $K$
 also has the property that the equality (\ref{01}) holds for all
 endofunctors. This follows from the next lemma whose proof is a variation of Trnkov\'a's  proof
of Proposition~\ref{R-Tr} (cf.~\cite{Tr}):

\begin{lemma} \label{var}
 In $\KVec$, finite intersections of monomorphisms are absolute, i.e., preserved by every functor with
domain $\KVec$.
\end{lemma}

\begin{corollary} For every endofunctor of $\KVec$ we have
  \[
  \text{initial algebra} = \text{final well-founded coalgebra}.
  \]
\end{corollary}

\begin{remark}
 The existence of an initial algebra is equivalent to the existence of
 a space $X \cong HX$, see Proposition~\ref{P-chain}.
\end{remark}

\section{Well-pointed coalgebras}
\label{treti}\setcounter{equation}{0}

\subsection{Simple coalgebras}\label{simp}
\hfill

\smallskip\noindent
We arrive at the centerpiece of this paper, characterizations of 
the initial algebra, final coalgebra, and initial iterative algebra for 
  endofunctors preserving intersections. Recall from Section
    \ref{druha} that subcoalgebras are represented by homomorphisms carried by monomorphisms in~\A, and quotient coalgebras by homomorphisms carried by strong epimorphisms in~\A.  

Here we prove that an endofunctor preserving
monomorphisms has a final coalgebra iff it has only a set of
simple coalgebras (up to isomorphism). For concrete categories and
endofunctors preserving intersections we prove a stronger result: the
final coalgebra consists of all
\textit{well-pointed coalgebras} which are those pointed coalgebras
with no proper quotient and no proper subobject. And a much sharper result is obtained if the base category is an equational class of algebras. Numerous examples of
this type of description of final coalgebras are presented in Section
\ref{ctvrta}.

\begin{assumption}\label{A-coc}
Throughout this section~\A~denotes a cocomplete, wellpowered and
cowellpowered category. And  $H: \A \to \A$ is an endofunctor.\end{assumption}

Additionally,  in a number of results below we assume that $H$ preserves (wide) intersections, i.e., multiple pullbacks of monomorphisms.

\begin{examples}\label{E-coc}
In the case where~\A$=$\Set~the assumption that $H$ preserves intersections is an extremely mild condition:
examples include
\begin{enumerate}[(a)]
\item
the power-set functor, all polynomial functors, the finite
distribution functor,

\item
products, coproducts, quotients, and subfunctors of functors
preserving intersections, and

\item
``almost'' all finitary functors: if $H$~is finitary then $\bar{H}$
in Theorem~\ref{T-main-set} preserves intersections, see Lemma \ref{810}.
\item  An example of an important set functor not preserving
    intersections is the continuation monad $X\mapsto R^{(R^X)}$ for a
  fixed set $R$.
\end{enumerate}
\end{examples}

\begin{remark}\label{R-(e,m)}\hfill
\begin{enumerate}[(a)]
\item We are working with factorizations of morphisms as  strong epimorphisms   followed by monomorphisms, see Remark \ref{R-fact}. Recall from
  \cite{AHS} that every cocomplete and cowellpowered category has such
  factorizations.

\item We use Terminology \ref{T-sub} and recall from Remark \ref{R-str} that quotients and subcoalgebras 
 form a factorization system in $\Coalg H$  whenever $H$ preserves monomorphisms.  In this case, a coalgebra $(A,a)$ is simple (see Definition \ref{D-2.1}) iff every homomorphism from it is a subcoalgebra.

\end{enumerate}
\end{remark}

\begin{notation}
  From now on we will write 
  \[
  \nu H \qquad\text{and}\qquad \mu H
  \]
  for the final coalgebra and initial algebra for $H$,
  respectively, whenever they exist. 
\end{notation}

\begin{examples}\label{E-simp}\hfill
\begin{enumerate}[(1)]
\item If $H$ has a final coalgebra $\nu H$, then $\nu H$ is simple. Indeed, the terminal object of every category is (clearly) simple.

\item If a set functor $H$ has an initial algebra, then the corresponding coalgebra
	is simple (see Theorem \ref{T-main-set}). More generally, let ~\A~be an LFP category with a simple initial object. If $H$ preserves finite intersections and has an inital algebra
$\mu H$, then $\mu H$ is (as a coalgebra) simple. Indeed, by Theorem
\ref{T-main}, $\mu H$ is a final well-founded coalgebra. Since
well-founded coalgebras are closed under quotients (see Lemma
\ref{L-colim}), it follows that $\mu H$ is simple (in $\Coalg H$).

\item A deterministic automaton considered as a coalgebra of
$$HX=X^I \times \{0,1\}\qquad \qquad (I = \text{the set of inputs})$$
  is simple iff it is \textit{observable}. That is, every pair of
  distinct states accept distinct languages.

\item A graph, considered as a coalgebra for $\PP$, is simple iff it has pairwise
  non-bisimilar vertices.
\end{enumerate}
\end{examples}

\begin{observation}\label{ord}
  Simple coalgebras form an ordered class (up to isomorphism), i.e.,
  between two simple coalgebras there exists at most one homomorphism.

  Indeed, given a parallel pair $h_1,h_2\colon (A,a)\to (B,b)$, their
  coequalizer is a quotient of $(B,b)$, hence it is invertible and we
  conclude $h_1=h_2$.
\end{observation}

\begin{proposition}[Gumm \cite{Gu}] \label{P-simp}
Every coalgebra has a unique simple quotient represented by the wide pushout
$$e_{(A,a)}\colon (A,a)\to (\bar{A},\bar{a})$$
of all quotients. If $H$ preserves monomorphisms, this is the reflection of $(A,a)$ in the full subcategory of all simple coalgebras.
\end{proposition}

Gumm worked with $\A=\Set$, but his argument extends without problems: for every coalgebra homomorphism $f\colon (A,a)\to (B,b)$ there exists a unique coalgebra homomorphism $\bar{f}\colon (\bar{A},\bar{a})\to (\bar{B},\bar{b})$ with $\bar{f}\cdot e_{(A,a)}=e_{(B,b)}\cdot f$ by the universal property of wide pushouts.

\begin{corollary}\label{C-sub}
Every subcoalgebra of a simple coalgebra is simple.
\end{corollary}

Indeed, every full (strong epi)-reflective subcategory is closed under subobjects.

\begin{theorem}\label{T-set-s}
 For every  endofunctor $H$ the existence of $\nu H$ implies
   that $H$ has only a set of simple coalgebras (up to
   isomorphism). If $H$ preserves monomorphisms, the converse implication also holds.
\end{theorem}

\begin{remb}  Moreover, if $(A_i,a_i)$, $i\in I$, is a set of representatives of all simple coalgebras, then $\nu H$ is the simple quotient of their coproduct:
$$\nu H=(\bar{A},\bar{a})\qquad \text{ where }\qquad (A,a)=\coprod_{i\in I} (A_i,a_i).$$
The theorem is a consequence of Freyd's Adjoint Functor Theorem. We
include a (short) proof for the convenience of the reader.
\end{remb}

\begin{proof} (1) Let $H$ have a set $(A_i,a_i)$, $i\in I$, of
  representative simple coalgebras. Proposition~\ref{P-simp} implies that this set is weakly final: for every coalgebra $(B,b)$ choose $i\in I$ with $(\bar{B},\bar{b})\simeq (A_i,a_i)$ and obtain a homomorphism $(B,b)\to (A_i,a_i)$. Consequently, the  coproduct $(A,a)$
  above is a weakly final object, hence, so is its quotient $(\bar{A},\bar{a})$. For every parallel pair of morphisms with codomain $(\bar{A},\bar{a})$ their coequalizer is invertible (since the codomain is simple, see Remark \ref{R-(e,m)}). Hence, $(\bar{A},\bar{a})$ is final.

  (2)  Let $\nu H$ exist. Then for every simple coalgebra the unique
  homomorphism into  $\nu H$ is  monic. Therefore, since~\A~is wellpowered by assumption,  $H$ has only a set of simple coalgebras up
  to isomorphism.  
\end{proof}

 \begin{example}\label{E-count} If $H$ does not preserve monomorphisms, then it can have both a final coalgebra and  a proper class of simple coalgebras which are pairwise non-isomorphic.

On the category $\Gra$ of graphs and graph morphisms define an endofunctor, based on the power-set functor $\PP$, as follows:
$$HX=\left\{  \begin{array}{ll}\PP X\; \mbox{(no edges)}& \mbox{if $X$ has no edges}\\
1&\mbox{else}
\end{array}\right.$$
For morphisms between graphs without edges put $Hf = \PP f$.
Then $1=H1$  is the final coalgebra.

Now $\PP$  as an endofunctor of $\Set$ has, since no final coalgebra exists, a proper class of simple, pairwise non-isomorphic coalgebras $a_i:A_i\to \PP A_i\, (i\in I)$. Consider $A_i$ as a graph without edges, then $(A_i,a_i)$ is a coalgebra for $H$. And this coalgebra is simple because if a coalgebra homomorphism $e:(A_i,a_i)\to (B,b)$ is carried by a strong epimorphism  $e:A_i\to B$ of $\Gra$,  then the fact that $A_i$  has no edge implies that neither has $B$. Then $e$ is a homomorphism in $\Coalg \PP$ which implies that it is invertible (in $\Set$, hence, in $\Gra$). Thus, we obtain a proper class of simple $H$-coalgebras $(A_i,a_i)$.
\end{example}

\subsection{Well-pointed coalgebras}\label{well-p}

\begin{remark}\label{R-concr}
In the rest of Section \ref{treti} we assume that the base category~\A~is \textit{concrete}, i.e., a faithful functor
$$U\colon \A\to \Set$$
is given. We require that $U$
\begin{enumerate}[(a)]
\item preserves intersections,
\item is \textit{fibre-small}, i.e., for every set $X$ there exists
  up to isomorphism only a set of objects $A\in \A$ with $UA=X$, and
\item is \emph{uniquely transportable}, i.e., for every object $A$ and every
  bijection $b\colon UA\to X$ in $\Set$ there exists a unique object
  $\bar{A}$ with $U\bar{A}=X$ and $A\simeq \bar{A}$ where the
  isomorphism is carried by $b$.
\end{enumerate}
Condition (c) is harmless: every concrete category is equivalent to a
uniquely transportable one, see \cite[Proposition 5.36]{AHS}. Also (a) and (b) are
conditions fulfilled by all ``everyday'' concrete categories: usually
$U$ is the hom-functor of an object $G$ which is a generator, and then
(a) and (b) hold. More generally:
\end{remark}

\begin{example}\label{E-gener}\hfill
  \begin{enumerate}[(1)]
  \item  Let $G_i\; (i\in I)$ be a generating set of~\A, i.e.,
    for every parallel pair $f_1,\, f_2\colon A\to B$ of distinct
    morphisms there exists $i\in I$ and $h\colon G_i\to A$ with
    $f_1\cdot h\neq f_2\cdot h$. Then the functor
    \[
    U=\coprod_{i\in I}\A(G_i,-)\colon \A\to \Set
    \]
    is faithful, fibre-small, and preserves intersections. Indeed,
    faithfulness is equivalent to $\{G_i: i \in I\}$ forming a generating set. Each
    $\A(G_i,-)$ preserves limits, and connected limits commute with
    coprodutcs in $\Set$, thus, $U$ preserves connected
    limits. Fibre-smallness follows from~\A~being cocomplete and
    cowellpowered: for every object $A$ the canonical morphism
    $$e\colon \coprod_{i\in I}\A(G_i,A)\bullet G_i\to A$$
    where $- \bullet G_i$ denotes copowers of $G_i$ (and the
    $f$-component of $e$ is $f$ for every $f\in \A(G_i,A)$) is an
    epimorphism. This is also equivalent to $\{G_i : i \in I\}$ forming a
    generating set. For every set $X$ all objects $A$ with $UA=X$ are
    thus quotients of $\coprod_{i\in I}X_i\bullet G_i$ where
    $X_i\subseteq X$. Since~\A~is cowellpowered, all these quotients
    form a set of objects up to isomorphism.
  \item Every LFP category~$\A$ is concrete as described in the previous
    point when one chooses as generating set any set of
    representatives of all finitely presentable objects up to
    isomorphism. 
  \end{enumerate}
\end{example}

\begin{definition}\label{D-point}
By a \textbf{pointed coalgebra} is meant a triple $(A,a,x)$ consisting of a coalgebra $a: A\to HA$ and an element $x$ of $UA$. The category
$$\Coalg_p H$$
of pointed coalgebras has as morphisms from $(A,a,x)$ to $(B,b,y)$ those coalgebra homomorphisms $f\colon (A,a)\to (B,b)$ which preserve the point:
\[\bfig
\Atriangle<500,400>[1`UA`UB;x`y`Uf]
\efig
\]
\end{definition}

\begin{remark}
  As for $\Coalg H$, the
  quotients of a pointed coalgebra $(A,a,x)$ are precisely the
  morphisms with this domain carried by strong epimorphisms of~\A. And
  subcoalgebras are precisely the morphisms with codomain $(A,a,x)$
  carried by monomorphisms of~\A. Moreover, Remark~\ref{R-str}
  immediately extends to $\Coalg_p H$.
\end{remark}

\begin{definition}\label{D-well}
A \textbf{well-pointed coalgebra} is a pointed coalgebra with no proper quotient and no proper subobject.
\end{definition}

\begin{remb}\hfill
  \begin{enumerate}[(a)]
  \item To say that a pointed coalgebra $(A,a,x)$ has no proper subobject means precisely that $x$ generates the coalgebra $(A,a)$: whenever a  subcoalgebra $m\colon (B,b)\to (A,a)$ contains $x$ (in the image of $Um$) then $m$ is invertible. We call such coalgebras \textit{reachable}. Thus:
 \[
  \text{well-pointed} = \text{simple} + \text{reachable}.
  \]
  
\item It is easy to see that if $f: (A,a,x)\to (B,b,y)$ is a morphism of pointed coalgebras, and if $(A,a)$ is simple
and $(B,b,y)$ is reachable, then $f$ is an isomorphism.

\item In the case where $\A=\Set$ and $H=\PP$, reachability of a pointed graph means that every vertex can be reached from the chosen one. Suppose $H$ is an arbitrary set functor preserving intersections. Then reachability of coalgebras can be translated to
reachability of its canonical graph, see Definition~\ref{D-can}:
\end{enumerate}
\end{remb}

\begin{lemma}
  \label{L-wp} Let $H$ be a set functor preserving intersections. Then
  a pointed coalgebra~$(A,a,x)$ is reachable iff its pointed
  canonical graph is, i.e., every vertex can be reached from~$x$ by a
  directed path.
\end{lemma}

\begin{proof}
Recall $\tau_A\colon HA\to\PP A$ from Remark~\ref{R-can}. Take a
subcoalgebra~$(A',a')$ containing~$x$:
\[\bfig
\dtriangle|llb|/<--`>`>/<600,400>[A'`1`A;x'`m`x]
\hSquares(600,0)|aallrbb|<400>[A'`HA'`\PP A'`A`HA`\PP A;%
a'`\tau_{A'}``Hm`\PP m`a`\tau_A]
\efig\]
Then $A'$~is a subcoalgebra of the canonical graph~$(A,\tau_A\tec a)$ (as
a pointed coalgebra of~\PP). And vice versa: if $m\colon A'\to A$ is
a subobject of the pointed canonical graph then, since the square in
Remark~\ref{R-can} is a pullback, we have a unique structure
$a'\colon A'\to HA'$ of a subobject of~$(A,a,x)$. Therefore,
$(A,a,x)$~is reachable w.r.t.~$H$ iff $(A,\tau_A\tec a,x)$~is
reachable w.r.t.~\PP. 
\takeout{%
It is easy to see that the latter means that
every element can be reached from~$x$ by a directed path.
}%
\end{proof}

\begin{examples}
\label{E-wp}\hfill
\begin{enumerate}[(a)]
\item A deterministic automaton with a given initial state is a pointed
  coalgebra for $HX=X^I\times\{0,1\}$. Reachability means that every
  state can be reached (in finitely many steps) from the initial state. 
  The usual terminology is that reachability and observability (see
    Example~\ref{E-simp} (3)) together are called \emph{minimality.
    Thus, well-pointed coalgebras are precisely the minimal automata.}

\item For the power-set functor the pointed coalgebras are the pointed
  graphs. Well-pointed means reachable and  simple  (Example \ref{E-simp} (4)). See
  Subsection~\ref{subsF} for more details.
\end{enumerate}
\end{examples}

\begin{notation}
\label{N-wp}
If $H$~preserves intersections, then there is a canonical process of
turning an arbitrary  pointed  coalgebra~$(A,a,x)$ into a well-pointed one:
form the  simple quotient (see Proposition~\ref{P-simp}) pointed by
$Ue_{(A,a)}\tec x\colon1\to U\bar{A}$, then form the least subcoalgebra
containing that point: \[\bfig
\morphism|b|<400,0>[1`UA;x]
\morphism(400,0)|b|<500,0>[UA`U\bar{A};Ue_{(A,a)}]
\morphism|l|<900,400>[1`U\bar{A}_0;x_0]
\morphism(900,400)|b|<0,-400>[U\bar{A}_0`U\bar{A};Um]
\square(1200,0)<700,400>[\bar{A}_0`H\bar{A}_0`\bar{A}`H\bar{A};%
\bar{a}_0`m`Hm`\bar{a}]
\efig\]
Then $(\bar{A}_0,\bar{a}_0,x_0)$~is well-pointed by
Corollary~\ref{C-sub}. We denote the well-pointed coalgebra $(\bar{A}_0,\bar{a}_0,x_0)$~(unique up to isomorphism
) by 
$$\wellp(A,a,x)$$
and call it the \textit{well-pointed modification of $(A,a,x)$}.
\end{notation}

\begin{example}
\label{E-stat}
For deterministic automata
our process $A\to/|->/\bar{A}_0$ above means that we first merge the
states that are observably equivalent and then discard the states
that are not reachable. A more efficient way is first discarding the
unreachable states and then merging observably equivalent pairs.
Both ways are possible if our functor preserves inverse images:
\end{example}

\begin{remark}
\label{R-inv}
Let $H$ and $U$ preserve inverse images. Then a quotient of a reachable
pointed coalgebra is reachable. Indeed, given such a quotient~$e$ and
its subcoalgebra~$m$ containing the given point~$x$, form the inverse
image~$m'$ of~$m$ along~$e$, and apply $U$ to this pullback:
\[\bfig
\square<1000,600>[UA'`UA`U\bar{A}_0`U\bar{A};Um'`Ue'`Ue`Um]
\morphism(500,300)|b|<500,300>[1`UA;x]
\morphism(500,300)/-->/<-500,300>[1`UA';]
\morphism(500,300)|r|<-500,-300>[1`U\bar{A}_0;x_0]
\efig\]
Since $H$~preserves inverse images, $m'\colon A'\to A$ is a 
subcoalgebra of~$A$, and, since $U$ preserves inverse images too,  the universal property of pullbacks implies
that $A'$~contains the given point~$x$. Consequently, $m'$~is
invertible, thus, $m\tec e'$~is strongly epic, therefore $m$~ is invertible.

Thus, we have an alternative procedure of forming
well-pointed coalgebras from pointed ones, $(A,a,x)$: first form the
least pointed subcoalgebra~$(A_0,a_0,x)$. Then form the simple
quotient of~$(A_0,a_0)$.
\end{remark}

\begin{notation}
\label{N-nu}
The collection of all well-pointed coalgebras up to isomorphism is
denoted by $T$.
  For every coalgebra $a\colon A\to HA$ we have a function
\[
a^+\colon UA\to T
\qquad
\text{defined by}
\qquad
a^+(x)=\wellp(A,a,x).
\]
(Notice that the well-pointed modification $\text{wp}(A,a,x)$ is unique up to
isomorphism. Thus we have precisely one choice in $T$.)
\end{notation}

\begin{lemma}\label{L-plus} Let $H$ preserve monomorphisms. For every coalgebra homomorphism $h:(A,a)\to(B,b)$ the triangle 
\[\bfig
\Vtriangle[UA`UB`T;Uh`a^+`b^+]
\efig\]
commutes.
\end{lemma}

\begin{proof}
(a) Assume that both coalgebras above are simple. In particular, $h$ is a
monomorphism by   simplicity  of $(A,a)$. 
For every element $x\in UA$ we know that $a^+(x)$ is the subcoalgebra $m:(A_0,a_0)\to (A,a)$ generated by $x$. Therefore $h\cdot m:(A_0,a_0)\to (B,b)$ is a subcoalgebra of $(B,b,Uh(x))$, and since $(A_0, a_0, x_0)$, with $Um(x_0)=x$, is well-pointed, we conclude that it is  isomorphic to $b^+(Uh(x))$. Now $T$ contains just one representative of every well-pointed coalgebra up to isomorphism, consequently, $b^+(Uh(x))=a^+(x)$.

(b)
If the two coalgebras are arbitrary, form the reflection $\bar
{h}$ of $h$ (see Proposition \ref{P-simp}):
\[\bfig
\square<700,400>[(A,a)`(B,b)`(\bar
{A}, \bar{a})`(\bar{B}, \bar
{b});h`e_{(A,a)}`e_{(B,b)}`\bar
{h}]
\efig\]
Then for every element $x\in UA$ we have that $a^+(x)$ is the subcoalgebra of $(\bar
{A}, \bar
{a})$ generated by $\bar
{x}= Ue_{(A,a)}(x)$, thus $a^+(x)=\bar
{a}^+(\bar
{x} )$; analogously for $b^+(Uh(x))$. By applying (a) to $\bar
{h}$ in lieu of $h$ we conclude $a^+(x)=\bar
{a}^+(\bar
{x})=\bar
{b}^+(U\bar
{h}(\bar
{x}))=b^+(Uh(x))$.
\end{proof}

\begin{lemma} 
\label{plusinj}
If $(A,a)$ is a simple coalgebra, then $a^+ :UA \to T$ is injective.
\end{lemma}
\proof
Suppose that $\text{wp}(A,a,x^1) = \text{wp}(A,a,x^2)$.
Let $m_i:(A_i,a_i, x^i_0)\to (A,a, x^i)$ denote the smallest  subcoalgebra containing $x^i$ ($i=1,2$) which is isomorphic to  $\text{wp}(A,a,x^i)$. Let $$u\colon(A_1,a_1,x_0^1)\to (A_2, a_2, x_0^2)$$ be an isomorphism. Then since   $(A,a)$ is
 simple, we have $m_1=m_2\cdot u$ due to Observation \ref{ord}. From
 $Uu(x_0^1)=x_0^2$ we
 get 
 \[
 x^2=Um_2(x_0^2)=U(m_2u)(x_0^1)=Um_1(x_0^1)=x^1. \eqno{\qEd}
 \]

\subsection{Final coalgebras}\label{final-coal}
\hfill

\smallskip\noindent
We remind the reader that in this section, we assume that
the endofunctor $H$ preserves intersections.

\begin{theorem}
\label{T-final}
$H$ has a final coalgebra iff it has only a set of well-pointed
coalgebras up to isomorphism. Moreover, a set $T$ of representatives of well-pointed coalgebras carries the final coalgebra:
$$U(\nu H)=T.$$
\end{theorem}

\begin{remark}
  \label{R-final}
The final coalgebra for $H$ is, as we will also prove, characterized up to isomorphism as a coalgebra
$\bar \tau: \bar{T} \to H \bar T$ 
 with two properties:  $U\bar{T} = T$, and
for every coalgebra $(A,a)$ the function $a^+\colon UA\to T$ carries a
coalgebra homomorphism from $(A,a)$ to $(\bar T, \bar \tau)$.
\end{remark}

\begin{proof} 
The necessity follows from Theorem \ref{T-set-s}.
 For the sufficiency,  fix a set $T$ of representative well-pointed coalgebras.
We  also use Theorem~\ref{T-set-s}  to show that
 $H$ has a final coalgebra. Indeed, if $(A,a)$ is a simple coalgebra, then by Lemma~\ref{plusinj}, 
  $UA$ has cardinality at most  $\card T$. Since $\A$ is small-fibred
  and uniquely transportable, it has up to isomorphism of $\A$  only a
  set of objects whose underlying sets have cardinality at most $\card
  T$. Consequently, $H$ has up to isomorphism of $\Coalg H$ only a set
  of simple coalgebras: given an object $A$ with $\A(A,HA)$ of
  cardinality $\alpha$, there are at most $\alpha$ pairwise
  non-isomorphic coalgebras $b\colon B \to HB$ with $A\cong B$ in $\A$.

\takeout{
(1)  We first prove that every simple coalgebra $(A,a)$ has the property that given elements $x^1,\, x^2\in UA$ such that the pointed subcoalgebras generated (in the sense of  Remark \ref{D-well}) by $x^1$ and $x^2$ are isomorphic in $\Coalg_p H$, then $x^1=x^2$. Let $m_i:(A_i,a_i, x^i_0)\to (A,a, x^i)$ denote the smallest subcoalgebra containing $x^i$ ($i=1,2$). Let $$u\colon(A_1,a_1,x_0^1)\to (A_2, a_2, x_0^2)$$ be an isomorphism. Then since $(A_1,a_1)$ and $(A,a)$ are simple, we have $m_1=m_2\cdot u$ due to Observation \ref{ord}. From $Uu(x_0^1)=x_0^2$ we get $$x^2=Um_2(x_0^2)=U(m_2u)(x_0^1)=Um_1(x_0^1)=x^1.$$

$H$ has a final coalgebra. To prove this we apply Theorem
  \ref{T-set-s}. In order to see that there is (up to isomorphism)
only a set of simple coalgebras for $H$, observe that for every simple coalgebra $(A,a)$ we have  that the function 
$$a^{+}:UA\to T\, ,$$
assigning to every element $x\in UA$ the well-pointed modification of $(A,a,x)$ (see Notation \ref{N-wp}), is monic, by (1). Thus, $UA$ has cardinality at
most  $\card T$. Since $\A$ is small-fibred and uniquely
transportable, it has up to isomorphism of $\A$  only a set of objects
whose underlying sets have cardinality at most $\card
T$. Consequently, $H$ has up to isomorphism of $\Coalg H$ only a set
of simple coalgebras: given an object $A$ with $\A(A,HA)$ of
cardinality $\alpha$, there are at most $\alpha$ pairwise
non-isomorphic coalgebras $b\colon B \to HB$ with $A\cong B$ in $\A$.
}

Given the coalgebra structure
$$\tau:\nu H\to H(\nu H)$$
of the final coalgebra, we now prove that the map $\tau^+:U(\nu H)\to
T$ is a bijection. 
Indeed, $\tau^+$ is monic due to 
the simplicity of   $\nu H$  (see Example~\ref{E-simp} (1)) and
Lemma~\ref{plusinj}. 
To  check the surjectivity, let $a^+(x) \in T$, where $(A,a)$ is a coalgebra and $x\in UA$.
Then by Lemma~\ref{L-plus}, $a^+(x) = \tau^+(Uh(x))$, where $h: A\to \nu H$ is the coalgebra homomorphism.
This shows that the image of $\tau^+$ contains $a^+(x)$.

 \takeout{
Firstly, $\tau^+$ is monic due to (1): recall that
$\nu H$ is simple. And secondly, $\tau^+$ is epic because for every
member $(A,a,x)$ of $T$ the unique homomorphism $f:(A,a)\to (\nu H,
\tau)$ is a  monomorphism by Corollary \ref{C-simp}. Consequently, for
$$y=Uf(x)\qquad \text{in}\qquad U(\nu H)$$
we have
$$\tau^+(y)=(A,a,x).$$
Indeed, $\tau^+(y)$ is, by definition, the subcoalgebra of $(\nu H, \tau, y)$ generated by $y$. Now $f:(A,a,x)\to (\nu H, \tau, y)$ is a subcoalgebra containing $y$ ($=Uf(x)$), and since $(A,a,x)$ is reachable, we conclude that it is isomorphic, in $\Coalg_p H$, to $\tau^+(y)$. However, $T$ contains no distinct members isomorphic in $\Coalg_p H$, thus, $\tau^+(y)=(A,a,x)$. 
}

Since $\A$ is uniquely transportable, there exists a unique object $\bar{T}$ of $\A$ and a unique isomorphism $i
:\nu H\to \bar{T}$ with $\tau^+=Ui
$. Define a coalgebra $\bar{\tau}: \bar{T}\to H \bar{T}$ so that $i
$ is a coalgebra isomorphism:
$\bar{\tau}=Hi
\cdot \tau\cdot i^{-1}.$ The coalgebra $(\bar{T},\bar{\tau})$ is final because 
for every  coalgebra $(A,a)$ we have a unique coalgebra homomorphism 
$a^*\colon(A,a)\to (\nu H, \tau)$, hence a unique coalgebra
homomorphism $i\cdot a^*\colon (A,a)\to (\bar{T}, \bar{\tau})$:
\[\bfig
\square/->`->`->`->/<800,400>[\nu H`H(\nu H)`\bar{T}`H(\bar{T});`i
`Hi
`\bar{\tau}]
\square(0,400)/->`->`->`->/<800,400>[A`HA`\nu H`H(\nu H);
a`a^*`Ha^*`\tau]
\morphism(50,80)|r|/->/<0,310>[`;i^{-1}]
\efig\]

We conclude with the verification of Remark~\ref{R-final}.
 First
 we show that $\bar{\tau}^+=\mbox{id}$. To see this apply
  Lemma~\ref{L-plus} to $i$ in order to get $\tau^+=\bar{\tau}^+\cdot
  Ui$. But since $\tau^+=Ui$, we get $Ui=\bar{\tau}^+\cdot Ui$, hence
  $\bar{\tau}^+=\mbox{id}$ because $Ui$ is an isomorphism. 

  To see that $a^+  = U(i\cdot a^*)$, we use Lemma \ref{L-plus} again: 
\[
a^+=\bar{\tau}^+\cdot U(i\cdot a^*)=\mbox{id}\cdot U(i\cdot
a^*)=U(i\cdot a^*).
\]
For the uniqueness, suppose that $(T',\tau': T'\to HT')$  also fulfils 
$UT' = T$ and 
for all coalgebras $(A,a)$,  the map $a^+: UA\to T$
is  $U(b)$ for some coalgebra homomorphism $b: (A,a)\to (T',\tau')$.
We apply this with $(A, a) = (\bar{T},\bar{\tau})$, and so $\id = \bar{\tau}^+:  T \to T$ is $Uf$ for some 
coalgebra morphism $f: (\bar{T},\bar{\tau}) \to (T',\tau')$.
However, by unique transportability, there is some isomorphism $g: \bar{T}\to T'$ such that $Ug = \id$.
And by faithfulness, $f = g$.  Thus the coalgebras 
 $(\bar{T}, \bar{\tau})$ and $(T', \tau')$ are isomorphic.
\end{proof}

\begin{example}\label{E-sets} Let $H$ be a set functor preserving intersections. If $T$ is a set of representatives of all well-pointed coalgebras, then $T$ is a final coalgebra. Its coalgebra structure assigns to every member $(A,a,x)$ of $T$ the following member of $HT$:
\begin{equation}
  \label{coalg}
  1\to^x A\to^a HA\to^{Ha^+}HT.
\end{equation}
See Section \ref{ctvrta} for numerous concrete examples. 
\end{example}

\begin{example}\label{E-fail}  
  If $H$ does not preserve intersections the theorem can fail: the functor in Example \ref{E-count}
 has a proper class of well-pointed coalgebras.
 \end{example}

\begin{example}\label{E-seq-T} For the set functor
$$HX=X^I\times\{0,1\}$$
presenting deterministic automata the well-pointed coalgebras are precisely the minimal (i.e., reachable and observable $=$ simple) automata. Since every language $L\subseteq I^{*}$ is accepted by a minimal automaton, unique up to isomorphism, we get the more usual description of the final coalgebra
$$\begin{array}{rl}\nu \, X\, .\,  X^I\times \{0,1\}&= \mbox{all minimal automata} \\
&\cong \PP I^{*}\;\mbox{(all languages).}
\end{array}$$
\end{example}

\begin{remark}\label{R-class}
Actually \textit{every} set functor $H$ has a final coalgebra, but this can be a proper class. More precisely, $H$  has an extension $H^{*}$ to the category of classes and functions unique up to a natural isomorphism, and $\nu H^{*}$ exists, see \cite{AMV0}.
\end{remark}

\begin{corollary}\label{C-class} 
For every set functor $H$ preserving intersections a class of representative well-pointed $H$-coalgebras with the coalgebra structure
given by the formula \eqref{coalg} is a final coalgebra for $H^{*}$.
\end{corollary}

The proof is completely analogous to that of Theorem \ref{T-final}.

\begin{example}\label{E-class}  The final coalgebra of the power set functor is the class of all well-pointed graphs (up to isomorphism).
\end{example}

\begin{construction}\label{C-mon}  Now let $\A$ be a variety
  of algebras determined by a set $E$ of equations, for a fixed
  signature $\Sigma$. Given a set
  $T$ representing all well-pointed coalgebras up to isomorphism, we
  turn it into a final coalgebra of $H$ as follows.

  \begin{enumerate}[(a)]
  \item $T$ as a $\Sigma$-algebra. For every $n$-ary symbol $\sigma \in \Sigma$  define $\sigma^T:T^n \to
    T$ as follows: Given an $n$-tuple of elements $(A_i, a_i,x_i)$ of $T$ form a coproduct 
    $$A=\coprod_{i<n}A_i\; \; \quad \text{in}\; \A$$
    and obtain a coproduct $(A,a)=\coprod_{i<n}(A_i,a_i)$ in $\Coalg H$
    together with elements $\hat{x}_i \in A$ corresponding to $x_i\in A_i$. 
      For the element
      \[
      y=\sigma^A (\hat{x}_i)_{i<n}\quad \text{of $A$}
      \]
      we define the result of $\sigma^T$ as the well-pointed modification of $(A,a,y)$:
      \begin{equation}\label{form}\sigma^T(A_i,a_i, x_i)_{i<n}
        = \wellp(A,a,y).
      \end{equation}

  \item $T$ will be proved to satisfy all equations in $E$, i.e., $T$
    is an object of $\A$. And all $a^+$ in Notation \ref{N-nu} are
    $\Sigma$-homomorphisms. (See Lemma \ref{C_lem1}.)

  \item $T$ as a coalgebra. We have a function 
    \[
    \tau\colon T \to  HT
    \]
    defined precisely as in Example \ref{E-sets}:
    \[
    \tau(A,a,x)=Ha^+(a(x))
    \]
    We prove that $\tau$ is a $\Sigma$-homomorphism. (See Proposition \ref{P-mon}.)

  \item We derive that $(T,\tau)$ is a final coalgebra for $H$.
  \end{enumerate}
\end{construction}

\begin{lemma}\label{C_lem1}
  The $\Sigma$-algebra $T$ lies in $\A$ and for every coalgebra $(A,a)$ we have a $\Sigma$-homomorphism
  $a^+:A \to T$.
\end{lemma}

\begin{proof}
Recall the final coalgebra $(\bar{T}, \bar{\tau})$ from the proof of Theorem~\ref{T-final} whose underlying set is $T$. All we need to prove is that the operations $\sigma^{\bar{T}}$ of the $\Sigma$-algebra $\bar{T}$ are given by the formula 
\eqref{form} above. Indeed, given a well-pointed coalgebra $(A_i, a_i,x_i)$ we have $a^+_i(x_i)=(A_i,a_i,x_i)$. Let us apply Lemma  \ref{L-plus} to the coproduct injection $v_i:(A_i,a_i) \to (A,a)$: since $a^+_i=a^+\cdot v_i$ and $\hat{x}_i=v_i(x_i)$, we conclude 

$$(A_i, a_i, x_i)=a^+_i(x_i)=a^+(\hat{x}_i).$$
Since $a^+:A \to \bar{T}$ is (by Remark~\ref{R-final}) a $\Sigma$-homomorphism, we obtain 
\[
\sigma^{\bar{T}}\left((A_i,a_i,x_i)_{i<n}\right)
=
\sigma^{\bar{T}}\left((a^+(\hat{x}_i))_{i<n}\right)
=
a^+(\sigma^A((\hat{x}_i)_{i<n}))
=
a^+(y) 
=
\wellp(A,a,y)
\]
as required. 
\end{proof}
\begin{proposition}\label{P-mon}
The function $\tau(A,a,x)=Ha^+(a(x))$ is a $\Sigma$-homomorphism from $T$ to $HT$, and the coalgebra
 $(T,\tau)$ is final.
\end{proposition}

\proof
  For the final coalgebra$\bar{\tau}\colon \bar{T}\to H\bar{T}$ of the
  proof of Theorem \ref{T-final} we already know that $T=\bar{T}$. It
  remains to prove that $\bar{\tau}=\tau$. For every element $(A,a,x)$
  of $T$ we have $a^+(x)=(A,a,x)$ and, since $a^+:(A,a) \to
  (T,\bar{\tau})$ is by Remark \ref{R-final} a coalgebra homomorphism,
  \[
  \bar{\tau}(A,a,x)
  =
  \bar{\tau}(a^+(x))
  =
  Ha^+(a(x))
  =
  \tau(A,a,x).\eqno{\qEd}
  \]
\medskip

\begin{remark}\label{R-new-minim}\hfill
 \begin{enumerate}[(a)]
\item Generalizing deterministic automata, see Example \ref{E-seq-T},
every pointed coalgebra $(A,a,x)$ can be viewed as a realization of the
corresponding element $t = a^+(x)$ of the terminal coalgebra of $H$.
The well-pointed coalgebras are the minimal realizations of $t$.
       Then every element of $\nu H$ has a minimal realization,
unique up-to isomorphism.
       \item If the algebra $A$ above is finite, then minimality is
equivalent to state-minimality:
 \begin{enumerate}[(b1)]
       \item Every realization of $t$ has cardinality at least that of $A$,
and
       \item if it has the same cardinality as $A$, it is isomorphic
            to $(A,a,x)$.
            \end{enumerate}
This follows from the fact that every well-pointed coalgebra $(A,a,x)$ is,
as we have seen above, isomorphic to the subcoalgebra of $\nu H$ generated
by $a^+(x)$.
       \item  For non-deterministic automata we obtain minimization w.r.t. bisimilarity (i.e. w.r.t. to the branching behavior) but this is not minimization in the classical sense. The reason is that the terminal
coalgebra of the corresponding functor $(P_f X)^I \times \{0,1\}$ is not the
set of all languages over $I$.
\end{enumerate}
\end{remark}

\begin{example}\label{E-new-Bool} In the variety $\Bool$ of boolean algebras consider the functor
$HX = X^I \times \mathbbm{2}$ where $\mathbbm{2}$ is the $2$-element boolean algebra. Its coalgebras are deterministic automata with a boolean algebra structure
on the states such that (1) final states form an ultrafilter, and
(2) transitions preserve the boolean operations. The terminal coalgebra is
the boolean algebra of all well-pointed coalgebras, and this is
isomorphic to the boolean algebra of all languages over $I$.

	For every regular language $L$ the unique minimal realization $A$ (i.e. the corresponding well-pointed H-coalgebra) is finite, but possibly larger than the minimal automaton in $\Set$. However,
by restricting ourselves to the atoms of the boolean algebra $A$, one
obtains a nondeterministic automaton which Brzozowski and Tamm \cite{BT}
call the \'atomaton for $L$ and which in some cases is the state-minimal
nondeterministic realization of $L$ (in the classical sense in $\Set$).
\end{example}
\subsection{Initial algebras}
\hfill

\begin{assumption}
  In the rest of this section $\A$ denotes an LFP category with a
  simple initial object, and $H$ an endofunctor preserving
  intersections.
\end{assumption}

 Just as the final coalgebra $\nu H$ for a set
functor $H$ consists of all well-pointed coalgebras (up to
isomorphism), we now prove that the initial algebra $\mu H$ consists
of all well-founded, well-pointed coalgebras.  In more detail: the
well-founded coalgebras in $\nu H$ form a subcoalgebra, and we prove
that this is a final well-founded coalgebra which by Theorem
\ref{T-main} is $\mu H$. In Section 4 we then present numerous
examples of initial algebras described in this manner.

\begin{notation}\label{N-mu} Recall the concept of well-founded
    coalgebra from Section \ref{druha}.   The collection of all well-founded, well-pointed coalgebras (up to
  isomorphism) is denoted by $I$. For every well-founded coalgebra
  $a:A\to HA$ we have a function
  \[
  a^{+}: \bigcup A\to I
  \]
  assigning to every element $x\colon1\to \bigcup A$ the well-founded,
  well-pointed coalgebra of Notation~\ref{N-wp}:
  \[a^+(x)=(\bar{A}_0,\bar{a}_0,x_0).\]
  Indeed, $(\bar{A}_0,\bar{a}_0)$~is well-founded due to
  Corollary~\ref{L-colim} and Lemma~\ref{wfdsub}. 
\end{notation}

\begin{remark}  \label{proh}
Observe that for a pointed coalgebra to be
well-founded and well-pointed two types of proper
subcoalgebras are prohibited: the cartesian ones, and those
containing the chosen point.
\end{remark}

\begin{theorem}
\label{T-init}  $H$
has an initial algebra iff it has only a set of well-founded,
well-pointed coalgebras up to isomorphism. Moreover a set $I$ of
representatives of well-founded, well-pointed coalgebras carries the
initial algebra: $\bigcup(\mu H)=I$.
\end{theorem}

\takeout{ 
\begin{remb} The initial algebra of $H$ is, as we will also prove, characterized as the only coalgebra on the set $I$ such that for every well-founded coalgebra $(A,a)$ the function $a^{+}:UA\to I$ carries a coalgebra homomorphism from $(A,a)$.
\end{remb}
}

\begin{proof}
(1)~If $H$~has an initial algebra~$\mu H$, then by
Theorem~\ref{T-main} this is a final well-founded coalgebra. 
Every well-founded, well-pointed coalgebra, being in particular
simple, is a subcoalgebra of~$\mu H$, since the unique homomorphism
into $\mu H$ is carried by a monomorphism. Consequently, $I$~is a set.

(2)
Let $H$~have a set~$I$ of representatives of well-founded,
well-pointed coalgebras. $I$ carries a canonical coalgebra structure
\[\bar{\psi}\colon I \to HI.\]
As in Theorem~\ref{T-final}, this structure
assigns to every member~$(A,a,x)$ of~$I$ the following
element of~$\bigcup HI$:
\[1\to^x \bigcup A\to^{\bigcup a}\bigcup HA\to^{\bigcup Ha^+}\bigcup HI.\]
We prove below that this is a final well-founded coalgebra. Thus, by
Theorem~\ref{T-main}, $I$~is an initial algebra
w.r.t. the inverse of~$\bar{\psi}$.

The proof that for every well-founded coalgebra $(A,a)$ the map $a^{+}:\bigcup A\to I$ carries a unique coalgebra homomorphism into $\bar{\psi}\colon I\to
H(I)$ is  completely analogous
to the proof of Theorem~\ref{T-final}.
 Just recall that subcoalgebras and quotients
of a well-founded coalgebra are all well-founded (by Corollary
\ref{L-colim} and Lemma \ref{wfdsub}).

It remains to prove that $(I, \bar{\psi})$ is a well-founded coalgebra. To this end notice that for every well-pointed, well-founded coalgebra $(A,a,x)$ in $I$ we have that $$a^{+}(x)=(A,a,x).$$
Now take the coproduct (in $\Coalg H$) of all $(A,a)$ for which there is an $x\in A$ such that $(A,a,x)$ lies in $I$. This coproduct is a well-founded coalgebra by Corollary \ref{L-colim}, and, as we have just  seen,  the unique induced homomorphism from this coproduct into $(I,\bar{\psi})$ is epimorphic, whence $I$ is a quotient coalgebra of the coproduct. Thus, another application of Corollary \ref{L-colim} shows that $(I, \bar{\psi})$ is a well-founded coalgebra as desired.
\end{proof}

\begin{example}\label{E-seq-I}
The initial algebra for $HX=X^I\times \{0,1\}$, and more generally, for any set functor $H$ with $H\emptyset=\emptyset$, is empty. No non-empty coalgebra is well-founded (due to the cartesian subcoalgebra $\emptyset$) and thus no pointed coalgebra is well-founded.
\end{example}

\begin{remark}
\label{R-large}
Analogously to Corollary~\ref{C-class}, every set functor~$H$ has a,
possibly large, initial algebra. That is,  the extension~${H}^{*}$
of~$H$ to classes always has an initial algebra: 
\[\mu H^{*}=\text{all well-founded, well-pointed algebras}\]
(up to isomorphism). This is  a
subcoalgebra of~$\nu H$ of Remark~\ref{R-class}. And as an algebra
for~${H}^{*}$ it is initial:
\end{remark}

\begin{corollary}
\label{C-mu}
For every intersection preserving set functor~$H$ the large
coalgebra~$\mu H^{*}$ is the final well-founded coalgebra for~${H}^{*}$. Thus,
the large initial algebra is~$\mu H^{*}$ w.r.t. the inverse of~$\bar{\psi}$.
\end{corollary}

The first statement follows from the Small Subcoalgebra Lemma
of~\cite{AM} and the fact that subcoalgebras of well-founded
coalgebras are well-founded (Corollary~\ref{C-sub}). The second
statement is proved precisely as Theorem~\ref{T-main}.

\begin{example}
\label{E-init-j}
The initial algebra for~\PP\ consists of all well-founded,
well-pointed graphs.
\end{example}

\begin{remark}\label{R-new-variet} The above theorem generalizes to endofunctors of finitary varieties having a simple initial object (and thus satisfying Assumptions \ref{ASS}). Let H be an endofunctor
preserving intersections. Given a set $I$ representing
well-founded, well-pointed coalgebras, we turn $I$ into a coalgebra
of $H$ as in Construction \ref{C-mon}, using the fact that (by Proposition \ref{P-starr})
well-founded coalgebras are closed under coproducts. The rest of the proof
is, due to Theorem \ref{T-main}, completely analogous to Proposition \ref{P-mon}.
\end{remark}

\subsection{Initial iterative algebras}
\hfill

\smallskip\noindent In this subsection, we study another subcoalgebra
of the final coalgebra for a set functor: all finite well-pointed
coalgebras. We prove that this is the initial iterative algebra (also
known as the rational fixed point). Before doing so we recall what
completely iterative and iterative algebras are. Once again, there is
no problem in generalizing the results below to locally finitely
presentable base categories with a simple initial object and which are
concrete via a given $U: \A \to \Set$.

\begin{remark}
\label{R-cit}
We know, from Theorem~\ref{T-main}  and \ref{T-init}, that $\mu H$~has a double
role: an initial algebra and a final well-founded coalgebra. Also
$\nu H$~has a double role. Recall from~\cite{M1} that an algebra
$a\colon HA\to A$ is \textit{completely iterative} if for every
(equation) morphism $e\colon X\to HX+A$ there exists a unique
\emph{solution}, i.e., a unique morphism $e^\dag\colon X\to A$ such that the
square below commutes:
\[\bfig
\Square/>`>`<-`>/<400>[X`A`HX+A`HA+A;e^\dag`e`{[a,A]}`He^\dag+A]
\efig\]
\end{remark}

\begin{theorem}[see \cite{M1}]
\label{T-cit}
For every endofunctor
\[\text{final coalgebra}=\text{initial completely iterative
algebra.}\]
\end{theorem}

\begin{remark}
\label{R-it}
(a)
Let $H$~be a \textit{finitary} set functor, i.e., every element $x\in
HA$ lies, for some finite subset $m\colon A'\to A$,  in the image
of~$Hm$. Then an algebra $a\colon HA\to A$ is called
\textit{iterative} provided that for every equation morphism $e\colon
X\to HX+A$ with
$X$~finite, there exists a unique solution $e^\dag\colon
X\to A$.

This concept was studied for classical \Sig-algebras by
Nelson~\cite{N} and Tiurin~\cite{Ti}, and for $H$-algebras in general
in~\cite{AMV}.

(b)
Form the colimit~$C$, in~\Set, of the diagram of all finite
coalgebras $a\colon A\to HA$ with the colimit cocone $a^+\colon A\to
C$. Then there exists a unique morphism $c\colon C\to HC$ with $c\tec
a^+=H a^+\tec a$. It was proved in~\cite{AMV} that $c$~is invertible
and the resulting algebra is the initial iterative algebra for~$H$.
\end{remark}

\begin{example}[see \cite{AMV}]
\label{E-iter}
(a)
The initial iterative algebra of $HX=X^I\times\{0,1\}$ consists of
all finite minimal automata. This is isomorphic to its description as
all regular languages.

(b)
The initial iterative algebra of the finite power-set functor
consists of all finite well-pointed graphs. See Section~\ref{ctvrta}
for a description using rational trees.
\end{example}

\begin{definition}[see \cite{M2}]
\label{D-locfin}
A coalgebra is called \textbf{locally finite} if every element lies
in a finite subcoalgebra.
\end{definition}

\begin{theorem}[see \cite{M2}]
\label{T-locfin}
Let $H$~be a finitary set functor. Then
\[\text{initial iterative algebra}=\text{final locally finite
coalgebra.}\]
Moreover, the final locally finite coalgebra is the colimit of all
finite coalgebras in~$\Coalg H$.
\end{theorem}

\begin{remark}
\label{R-inter-fin}
We prove below that given a finitary set functor, the set of all
finite well-pointed coalgebras forms the initial iterative algebra.
For this result we do not need to assume (unlike the rest of this
section) that the functor preserves intersections. This can be
deduced from the following
\end{remark}

\begin{lemma}
\label{E-inter}
For every finitary set functor~$H$ the Trnkov\'a closure~$\bar{H}$ (see Remark \ref{R-hull}) preserves (wide) intersections.
\end{lemma}

\begin{proof}
The functor~$\bar{H}$ of Proposition~\ref{R-Tr} is obviously also
finitary. It preserves finite intersections, and we deduce that it
preserves all intersections. Given subobjects $m_i\colon A_i\to B$
($i\in I$) with an intersection $m\colon A\to B$, let $x\in\bar{H}B$
lie in the image of each~$Hm_i$; we are to prove that $x$~lies in the
image of~$\bar{H}m$. Choose a subset $n\colon C\to B$ of the smallest
(finite) cardinality with $x$~lying in the image of~$\bar{H}n$. Since
$\bar{H}$~preserves the intersection of $n$ and~$m_i$, the minimality
of~$C$ guarantees that $n\subseteq m_i$ (for every $i\in I$). Thus,
$n\subseteq m$, proving that $x$~lies in the image of~$\bar{H}m$.
\end{proof}

\begin{notation}
\label{N-rho}
For every finitary set functor denote by
\[\varrho H\]
the set of all finite well-pointed coalgebras up to isomorphism.

Given a finite coalgebra $a\colon A\to HA$ we again define a function
\[a^+\colon A\to\varrho H\]
by assigning to every element $x\colon1\to A$ the well-pointed
coalgebra of Notation~\ref{N-wp}:
\[a^+(x)=\wellp(A, a, x).\]
This is well-defined due to Corollary
\ref{L-colim} and Lemma \ref{wfdsub}
since
$H$ and~$\bar{H}$
have the same pointed coalgebras.
\end{notation}

\begin{theorem}
\label{T-rho}
Every finitary set functor~$H$ has the initial iterative
algebra~$\varrho H$.
\end{theorem}

\begin{remb}
$\varrho H$~has the canonical coalgebra structure
\[{\psi}^{*}\colon\varrho H\to H(\varrho H).\]
It assigns, analogously to \eqref{coalg},  to every element~$(A,a,x)$ of~$\varrho H$ the following
element of~$H(\varrho H)$:
\[1\to^x A\to^a HA\to^{Ha^+} H(\varrho H).\]
We prove below that this is the final locally finite coalgebra. Thus,
$\varrho H$~is the initial iterative algebra
w.r.t. the inverse of~${\psi}^{*}$, by Theorem~\ref{T-locfin}.
\end{remb}

\proof
Analogously to the proof of Theorem~\ref{T-final} one verifies that
the morphisms
\[ a^+\colon(A,a)\to (\varrho H,{\psi}^{*})\qquad\text{($A$
finite)}\]
are coalgebra homomorphisms forming a cocone. By Remark~\ref{R-it}(b)
it remains
to prove that this is a
colimit in~$\Coalg H$.  We verify that  all $a^+$'s form
a colimit cocone in~\Set. That is:

\begin{enumerate}[(i)]
\item Every element of~$\varrho H$ has the form~$a^+(x)$ for some finite
  coalgebra~$(A,a)$ and some $x\in A$. Indeed, for every
  element~$(A,a,x)$ of~$\varrho H$ we have $a^+(x)=(A,a,x)$.

\item Whenever
  \[a^+(x)=b^+(y)\]
  holds for two finite coalgebras $(A,a)$ and~$(B,b)$ and for  elements
  $x\in A$, $y\in B$ (turning them into pointed coalgebras), there
  exists a zig-zag of homomorphisms of finite pointed coalgebras
  connecting~$(A,a,x)$ with~$(B,b,y)$. For that recall
  $a^+(x)=(\bar{A}_0,\bar{a}_0,x_0)$ in the notation~\ref{N-wp}. Here
  is the desired zig-zag:
  \[\bfig
  \Vtriangle/`>`>/<600,400>[(A,a,x)`(\bar{A}_0,\bar{a}_0,x_0)`%
  (\bar{A},\bar{a},e_{(A,a)}\tec x);`e_{(A,a)}`m_{(A,a)}]
  \Vtriangle(1200,0)/`>`>/<600,400>[(\bar{A}_0,\bar{a}_0,x_0)`(B,b,y)`%
  (\bar{B},\bar{b},e_{(B,b)}\tec y);`e_{(B,b)}`m_{(B,b)}]
  \efig
  \]\vspace{-2\baselineskip}

\qed
\end{enumerate}

\takeout{
\begin{example}
\label{E-seq-R}
The functor $HX=X^I\times \{0,1\}$ representing deterministic automata has the initial iterative algebra 
\begin{align*}
\varrho \, X .\, XÎ\times \{0,1\} & = \text{all finite minimal
  automata}
\\
&\cong\text{all regular languages in $\PP I^{*}$.}
\end{align*}
\end{example}
}

\begin{remark}
\label{R-fini}
For non-finitary set functors~$H$ the set~$\varrho H$ also carries
the  structure above of a coalgebra. But this is in general not a
fixed point of~$H$. For example, the functor $HX=X^{\nn}+1$ has the
final coalgebra consisting of all countably branching trees. And $\varrho H$~is the set of all rational
trees, i.e., those having only finitely many subtrees (up to
isomorphism), see Example~\ref{E-rat-2}. This is a subcoalgebra of
the final coalgebra, but not a fixed point of~$H$.
\end{remark}

\begin{remark}\label{R-new-loc-fin}
Theorem \ref{T-rho} generalizes immediately to varieties of algebras that are locally finite, i.e. free algebras on finitely many variables are finite (for example, boolean algebras or semilattices).

	Let $H$ be a finitary endofunctor preserving intersections. The set
	$$\rho H = \mbox{ all finite  well-pointed coalgebras}$$
forms a subcoalgebra of the coalgebra of Construction \ref{C-mon}. This is the initial iterative algebra for $H$. (The proof is entirely analogous to that
of Proposition \ref{P-mon}, based on Theorem \ref{T-locfin} and
the fact that a finite coproduct of finite coalgebras is finite.)
\end{remark}



\section{Examples of well-pointed coalgebras}
\label{ctvrta}\setcounter{equation}{0}
For a number of important set functors~$H$ we are going to apply the
results of Section~\ref{treti} and compare them to the well-known
description of the three fixed points of interest: the final
coalgebra, the initial algebra, and the initial iterative algebra
($=$~final locally finite coalgebra).
The last one is also called the rational fixed point of $H$.
 Throughout this section pointed
coalgebras are considered up to (point-preserving) isomorphism.
Recall that
\begin{align*}
\nu H&=\text{all well-pointed coalgebras}\\
\mu H&=\text{all well-founded, well-pointed coalgebras}\\
\intertext{and if $H$ is a finitary functor}
\varrho H&=\text{all finite well-pointed coalgebras.}
\end{align*}
We are using various types of labeled trees throughout this section.
Trees, too, are considered up to (label-preserving) isomorphism.
Unless explicitly stated, trees are ordered, i.e., a linear ordering
on the children of every node is always given.

In all our examples the endofunctors~$H$ used preserves intersections
and weak pullbacks. Recall from Rutten~\cite{R1} that this implies
that
\begin{enumerate}[(a)]
\item congruences on a coalgebra~$A$ are precisely the kernel equivalences
 of homomorphisms $f\colon A\to B$, and 
\item for every coalgebra the largest congruence is precisely the
  bisimilarity equivalence. 
\end{enumerate}

Also recall from Remark~\ref{R-inv} that, for these functors, every pointed coalgebra yields
a well-pointed one by first forming the ``reachable part'' and then the simple coreflection.

In pictures of pointed coalgebras the choice of the point~$q_0$ is
depicted by
\[\bfig
\morphism/^{(}->/<100,0>[`;]
\POS(160,0)
*+{q_0}*\cir{}
\efig\]

\subsection{Moore automata}
\label{subsA}
\hfill

\smallskip\noindent
Given a set~$I$ of inputs and a set~$J$ of outputs, a \textit{Moore
automaton} on a set~$Q$ (of states) is given by a next-state function
$\delta\colon Q\times I\to Q$ curried as
\[\curry\delta\colon Q\to Q^I\]
an output function
\[\out\colon Q\to J\]
and an initial state $q_0\in Q$. The first two items form a coalgebra for
\[HX=X^I\times J,\]
thus we work with pointed coalgebras for this functor, with $q_0$ as the chosen point.  The \textit{behavior} of
an automaton is the function
\[\beta\colon I^*\to J\]
which to every input word $w\in I^*$ assigns the output of the state
reached from~$q_0$ by applying the inputs in~$w$. A function
$\beta\colon I^*\to J$ is called \textit{regular} if the set of all
functions $\beta(w{-})\colon I^*\to J$ for $w\in I^*$ is finite.

\begin{lemma}
\label{L-Moore}
The largest congruence on a Moore automaton merges states $q$
and~$q'$ iff by applying an arbitrary finite sequence of inputs to
each of them, we obtain states with the same output.
\end{lemma}

This is well-known and easy to prove. Automata satisfying this
condition are called \textit{simple}. Another well-known fact is
the following

\begin{theorem}
\label{T-Moore}
For every function $\beta\colon I^*\to J$ there exists a reachable
and simple Moore automaton with the behavior~$\beta$. This
automaton is unique up to isomorphism. It is finite iff $\beta$~is
regular.
\end{theorem}

\begin{corollary}
\label{C-Moore}
For Moore automata, $HX=X^I\times J$, we have
$$\begin{array}{ll}
\nu H\, \cong J^{I^*},&\text{all functions $\beta\colon I^*\to
J$;}\\
\varrho H\, \cong&\text{all regular functions $\beta\colon I^*\to J$;}\\
\mu H\, =\emptyset.&
\end{array}$$
The coalgebra structure of~$\nu H$ (and~$\varrho H$) assigns to
every~$\beta\colon I^*\to J$
the pair in $(\nu H)^I\times J$ consisting of the
function $i\to/|->/\beta(i{-})$ for $i\in I$ and the
element~$\beta(\varepsilon)$ of~$J$.
\end{corollary}

Indeed, the isomorphism between~$\nu H$, the set of all reachable
and simple automata, and~$J^{I^*}$ is given by the  theorem above.
And the structure map of Example \ref{E-sets} is easily seen
to correspond to the   map above, taking~$\beta$ to
$(i\to/|->/\beta(i{-}),\beta(\varepsilon))$. Also the isomorphism
of~$\varrho H$ and all regular functions follows from the  theorem above;
 from Theorem~\ref{T-rho} we know that $\varrho H$~is a
subcoalgebra of~$\nu H$.

Finally, $\mu H=\emptyset$ since
no well-pointed coalgebra~$(A,a)$ is well-founded due to the
cartesian subcoalgebra
\[\bfig
\Square<400>[\emptyset`H\emptyset=\emptyset`A`A^I\times J;%
\id`m`Hm`a]
\efig\]

\begin{example}
\label{E-Moore}
If $J=\{0,1\}$ we get $\nu H=\PP I^*$ and $\varrho
H=\mbox{}$regular languages, see Examples \ref{E-seq-T}
and~\ref{E-iter} (a).
\end{example}

\subsection{Mealy automata}
\label{subsB}\hfill

\smallskip\noindent
For \textit{Mealy automata} the next-state function has the form
$\delta\colon Q\times I\to Q\times J$ and in curried form this is a
coalgebra for
\[HX=(X\times J)^I.\]
Given a state~$q$ of a Mealy automaton~$Q$, its \textit{response
function~$f_q$} is the function $f_q\colon I^\omega\to
J^\omega$ assigning to  an infinite word of input symbols the
infinite word of output symbols (delayed by one time unit) consisting of the
outputs given by the transitions (as the computations of the inputs are
performed, starting in~$q$). Observe that $f_q$~is a \textit{causal
function}, i.\,e., for every infinite word~$w$ the \mbox{$n$-th} component
of~$f_q(w)$ depends only on the first  $n$~components of~$w$.

\begin{remark}
\label{R-Mealy}
Given a causal function $f\colon I^\omega\to J^\omega$ the 
property above with $n=0$ tells us that the component~0 of~$f(w)$ only
depends on~$w_0$. We thus obtain a derived
function
\[f^0\colon I\to J\]
with $f(iw)=f^0(i)w'$ (for convenient $w'$) for all $w\in I^\omega$.
\end{remark}

\begin{lemma}
\label{L-Mealy}
For every Mealy automaton the largest congruence merges precisely the
pairs of states with the same response function.
\end{lemma}

\begin{proof}
Let $Q$~be a Mealy automaton, then the equivalence $q\sim q'$
iff $f_q=f_{q'}$ is obviously a congruence. We have a structure of a
Mealy automaton~$\bar{\delta}$ on~$Q/{\sim}$ derived from that
of~$Q$: Given a state $[q]\in Q/{\sim}$ and an input $i\in I$, the
pair $\delta(q,i)=(q',j)$ yields $\bar{\delta}([q],i)=([q'],j)$. It is
easy to verify that the canonical map $c\colon Q\to Q/{\sim}$ is a
coalgebra homomorphism
$c\colon(Q,\delta,q_0)\to(Q/{\sim},\bar{\delta},[q_0])$. Conversely,
every congruence is contained in~$\sim$ because given a coalgebra
homomorphism $h\colon Q\to\bar{Q}$ then for every state $q\in Q$ we
have $f_q=f_{h(q)}$. Thus, the kernel congruence of~$h$ is contained
in~$\sim$.
\end{proof}

\begin{corollary}
\label{C-Mealy}
The well-pointed Mealy automata are precisely those with an initial
state~$q_0$ such that the automaton is
\begin{enumerate}[\rm (a)]
\item reachable: every state can be reached from~$q_0$, and 
\item simple: different states have different response functions.
\end{enumerate}
\end{corollary}

The automata satisfying (a) and~(b) together are
called ``minimal''. The following theorem can be found in
Eilenberg \cite[Theorem XII.4.1]{E}:

\begin{theorem}
\label{T-Mealy}
For every causal function~$f$ there exists a unique well-pointed
coalgebra whose initial state has the response function~$f$.
\end{theorem}

\begin{remark}
\label{R-Mealy-2}
Eilenberg also proves that a minimal Mealy automaton is finite iff
$f$~has the property that the set of all functions~$f(w{-})$ where
$w\in I^*$ is finite. Let us call such causal functions
\textit{regular}.
\end{remark}

\begin{corollary}
\label{C-Mealy-2}
For Mealy automata, $HX=(X\times J)^I$, we have
\begin{align*}
\nu H&\cong\text{all causal functions from $I^\omega$ to $J^\omega$}\\
\varrho H&\cong\text{all regular causal functions}\\
\mu H&=\emptyset.
\end{align*}
The coalgebra structure of~$\nu H$ (and that of~$\varrho H$) assigns
to every causal function $f\colon I^\omega\to J^\omega$ the map
\[I\to\nu H\times J,\qquad i\to/|->/\bigl(f(i{-}),f^0(i)\bigr)\]
for $f^0\colon I\to J$ in Remark~\ref{R-Mealy}.
\end{corollary}

Indeed, the first two statements follow from the  theorem above,
and
the last one follows again from $H\emptyset=\emptyset$. The 
description of the final coalgebra is due to Rutten~\cite{R2}. 
Eilenberg works with functions $f\colon I^*\to J^*$ preserving length
and prefixes, but it is immediate that these are just another way of
coding all causal functions between infinite streams.

\begin{remark}
\label{R-Mealy-3}
An alternative description of the final coalgebra for $HX=(X\times
J)^I$ is:
\[\nu H\cong J^{I^+},\qquad\text{all functions $\beta\colon I^+\to
J$.}\]
Here and below, $I^+$ is the set of finite non-empty words on the set $I$.
The coalgebra structure assigns to every~$\beta$ the mapping from~$I$
to $\nu H\times J$ given by
\[i\to/|->/\bigl(\beta(i{-}),\beta(i)\bigr)\qquad\text{for $i\in
I$.}\]
\end{remark}

Indeed, this coalgebra is isomorphic to that of all causal functions
$f\colon I^\omega\to J^\omega$: to every function $\beta\colon I^+\to
J$ assign the causal function
\[
f(i_0i_1i_2\dots)=(\beta(i_0),\beta(i_0i_1),\beta(i_0i_1i_2),\ldots).
\]

\subsection{Streams}
\label{subsC}\hfill

\smallskip\noindent
Consider the coalgebras for
\[HX=X\times I+1.\]
Rutten~\cite{R1} interprets them as dynamical systems with outputs in~$I$
and with terminating states (where no next state is given). Every state~$q$
yields a stream, finite or infinite, over~$I$ by starting in~$q$ and
traversing the dynamical system as long as possible. We call it the
\textit{response} of~$q$. It is an element of~$I^\omega+I^*$.

\begin{lemma}
\label{L-stream}
For a dynamical system the largest congruence merges two states iff
they yield the same response.
\end{lemma}

\begin{proof}
Let $\sim$~be the equivalence from the statement of the lemma. Then we have an obvious dynamic
system on~$Q/{\sim}$, thus, $\sim$~is a congruence. Every coalgebra
homomorphism $h\colon Q\to\bar{Q}$ fulfils: the response of $q$
and~$h(q)$ is always the same. Therefore, $\sim$~is the largest
congruence.
\end{proof}

\begin{corollary}
\label{C-stream}
A well-pointed coalgebra is a dynamical system with an initial
state~$q_0$ such that the system is
\begin{enumerate}[\rm (a)]
\item reachable: every state can be reached from~$q_0$, and
\item simple: different states yield different responses.
\end{enumerate}
\end{corollary}

\begin{example}
\label{E-stream}\hfill
\begin{enumerate}[(a)]
\item For every word $s_1\dots s_n$ in~$I^*$ we have a well-founded dynamic
system
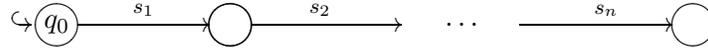
\[\bfig
\POS(0,0) *+{q_0}*\cir<8pt>{}
\ar@{->}+(600,0)*\cir<8pt>{}^{s_1}
\morphism(-160,0)/^{(}->/<100,0>[`;]
\POS(600,0) *\cir<8pt>{}
\ar@{->}+(620,0)*+{}^{s_2}
\place(1400,20)[\dots]
\POS(1600,0)
\ar@{->}+(600,0)*\cir<8pt>{}^{s_n}
\efig\]

\item
For every \textit{eventually periodic} stream in~$I^\omega$,
\[w=uv^\omega\qquad\text{for $u,v\in I^*$},\]
we have a pointed dynamical system
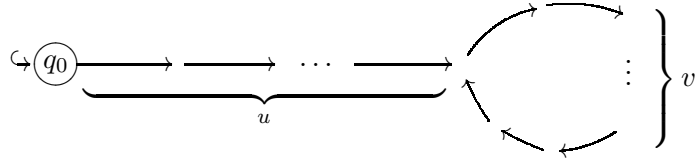
\[\bfig
\POS(0,300) *+{q_0}*\cir<8pt>{}
\ar@{->}+(400,0)
\morphism(400,300)<400,0>[\rule{3pt}{0pt}`\rule{3pt}{0pt};]
\place(900,320)[\dots]
\morphism(1000,300)<400,0>[\rule{1pt}{0pt}`\rule{1pt}{0pt};]
\morphism(1500,100)/{@{>}@/^1pt/}/<-100,200>%
[\rule{3pt}{0pt}`\rule{3pt}{0pt};]
\morphism(1700,0)/{@{>}@/^1pt/}/<-200,100>%
[\rule{3pt}{0pt}`\rule{3pt}{0pt};]
\morphism(1980,100)/{@{>}@/^3pt/}/<-280,-100>%
[\rule{3pt}{0pt}`\rule{3pt}{0pt};]
\morphism(1403,300)/{@{>}@/^5pt/}/<300,200>%
[\rule{0pt}{3pt}`\rule{3pt}{0pt};]
\morphism(1670,500)/{@{>}@/^2.5pt/}/<300,-30>%
[\rule{1pt}{0pt}`\rule{3pt}{0pt};]
\POS (1980,300) *+{\vdots}
\morphism(-160,300)/^{(}->/<100,0>[`;]
\place(720,250)[\underbrace{\rule{12.5em}{0em}}_u]
\place(2130,250)[\left.\rule{0em}{2.7em}\right\}v]
\efig\]
If we choose, for the given stream~$w$, the words $u$~and~$v$ of
minimum lengths, then this system is well-pointed.
\end{enumerate}
\end{example}

The following was already proved by Arbib and Manes~\cite[Theorem 10.2.5]{MA}.

\begin{corollary}
\label{C-stream-2}
For $HX=X\times I+1$ we have
$$\begin{array}{ll}
\nu H \, \cong I^*+I^\omega,&\text{all finite and infinite
streams,}\\
\varrho H \, \cong
&\text{all finite and eventually periodic streams,}\\
\mu H\, \cong I^*,&\text{all finite streams.}
\end{array}$$
The coalgebra structure assigns to every nonempty stream~$w$ the
pair
\[(\tail w,\head w)\qquad\text{in}\quad \nu H\times I\]
and to the
empty stream
the right-hand summand of $H(\nu H)=\nu H\times I+1$.
\end{corollary}

Indeed, the description of $\nu H$ follows from Corollary~\ref{C-stream}
since
by forming the response of~$q_0$ we get a bijection between
well-pointed coalgebras and streams
in~$I^*+I^\omega$.
For the description of $\varrho H$ observe that a well-pointed system yields a
finite
or eventually periodic response iff it has finitely many states. The point in our statement is that 
  $\mu H$
follows from the observation that a dynamical system is
well-founded iff every run of a state is finite. Indeed, given a
coalgebra $a\colon A\to A\times I+1$, form the subset $m\colon A'\to
A$ of all states with finite runs. We obtain a cartesian subcoalgebra
\[\bfig
\square/>```>/<800,400>[A'`A'\times I+1`A`A\times I+1;%
a'`m`m\times I+1`a]
\morphism(0,400)/{_(->}/<0,-400>[\strut`A;]
\morphism(800,400)/{_(->}/<0,-400>[\strut`A\times I+1;]
\efig\]
Thus,
well-founded, well-pointed coalgebras are precisely those of
Example~\ref{E-stream}(a).

\subsection{Binary trees}
\label{subsD}\hfill

\smallskip\noindent
Coalgebras for the functor
\[HX=X\times X+1\]
are given, as observed by Rutten~\cite{R1}, by a set~$Q$ of states
which are either terminating or have precisely two next states according to
a binary input, say~$\{l,r\}$. Every state $q\in Q$ yields an ordered
binary tree~$T_q$ (i.e, nodes that are not leaves have a left-hand
child and
a right-hand one) by \textit{tree expansion}: the root is~$q$
and a node is either a leaf, if it is a terminating state, or has the two
next states as children (left-hand for input~$l$, right-hand for
input~$r$). Binary trees are considered up to isomorphism.

\begin{lemma}
\label{L-tree}
For every system the largest congruence merges precisely the pairs of
states having the same tree expansion.
\end{lemma}

\begin{proof}
Let $\sim$~be the equivalence with $q\sim q'$ iff $T_q=T_{q'}$.
There is an obvious structure of a coalgebra on~$Q/{\sim}$ showing
that $\sim$~is a congruence. For every coalgebra homomorphism
$h\colon Q\to\bar{Q}$ the tree expansion of $q\in Q$ is always the
same as the tree expansion of~$h(q)$ in~$\bar{Q}$. Thus $\sim$~is the
largest congruence.
\end{proof}

\begin{corollary}
\label{C-tree}
A well-pointed system is a system with an initial state~$q_0$ which
is
\begin{enumerate}[\rm (a)]
\item
reachable: every state can be reached from~$q_0$, and
\item simple: different states have different tree expansions.
\end{enumerate}
\end{corollary}

Moreover, tree expansion is a bijection between well-pointed
coalgebras and binary trees (see Proposition~\ref{P-Sig} below). For instance, the
dynamical system
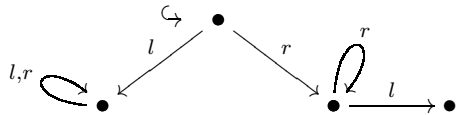
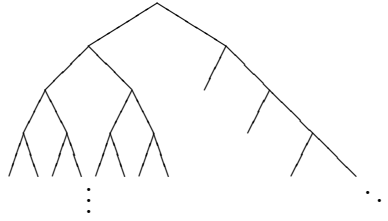
\[\bfig
\Loop(0,0){\bullet}(l,ul)^{l,r}
\morphism(400,300)|l|<-400,-300>[\bullet`\bullet;l]
\morphism(400,300)|r|<400,-300>[\bullet`\bullet;r]
\morphism(800,0)<400,0>[\bullet`\bullet;l]
\Loop(800,0){\bullet}(u,ur)^r
\morphism(200,300)/^{(}->/<100,0>[`;]
\efig\]
defines the tree
\[\bfig
\POS(50,150)
\ar@{-}+(-50,-150)
\ar@{-}+(50,-150)
\POS(200,150)
\ar@{-}+(-50,-150)
\ar@{-}+(50,-150)
\POS(350,150)
\ar@{-}+(-50,-150)
\ar@{-}+(50,-150)
\POS(500,150)
\ar@{-}+(-50,-150)
\ar@{-}+(50,-150)
\POS(125,300)
\ar@{-}+(-75,-150)
\ar@{-}+(75,-150)
\POS(425,300)
\ar@{-}+(-75,-150)
\ar@{-}+(75,-150)
\POS(275,450)
\ar@{-}+(-150,-150)
\ar@{-}+(150,-150)
\POS(512.5,600)
\ar@{-}+(-237.5,-150)
\ar@{-}+(237.5,-150)
\POS(750,450)
\ar@{-}+(-75,-150)
\ar@{-}+(450,-450)
\POS(900,300)
\ar@{-}+(-75,-150)
\POS(1050,150)
\ar@{-}+(-75,-150)
\place(275,-100)[\vdots]
\place(1280,-100)[\ddots]
\efig\]
Observe that this tree has only 4 subtrees (up to isomorphism):
this
follows from the fact that the dynamical systems has 4~states. In
general, the finite dynamical systems correspond to the
\textit{rational trees}, i.e., trees having (up to isomorphism) only
finitely many subtrees. This description is due to Ginali~\cite{Gi}.

\begin{corollary}
\label{C-tree-2}
For the functor $HX = X \times X + 1$ we have
\begin{align*}
\nu H & \cong\text{all binary trees,}\\
\varrho H & \cong\text{all rational binary trees,}\\
\mu H & \cong\text{all finite binary trees.}
\end{align*}
The coalgebra structure is, in each case, the inverse of tree tupling: it assigns to the
root-only tree the right-hand summand of $\nu H\times\nu H+1$ and to
any other tree the pair of its maximum subtrees.
\end{corollary}

Indeed, we only need to explain the last item $\mu H$. Given a coalgebra
$a\colon A\to A\times A+1$, let $m\colon A'\to/^{ (}->/A$ be the set
of all states defining a finite subtree. This is a
cartesian subcoalgebra
\[\bfig
\square/>```>/<800,400>[A'`A'\times A'+1`A`A\times A+1;%
a'`m`m\times m+\id`a]
\morphism(0,400)/{_(->}/<0,-400>[\strut`A;]
\morphism(800,400)/{_(->}/<0,-400>[\strut`A\times A+1;]
\POS(100,300)
\ar@{-}+(-50,0)
\ar@{-}+(0,50)
\efig\]
i.e., this square is a pullback: whenever a state $q\in A$ has both
next states in~$A'$ or whenever $q$~is final, then $q\in A'$. Thus,
if $A$~is well-founded, then $A=A'$. The converse implication is
easy: recall the subsets~$A^*_i$ of Notation~\ref{N-starr}. Here
$A^*_i$~is the set of all states whose binary tree has depth at
most~$i$. Thus, if $A=A_i$ for some~$i$, the initial state defines a
tree of depth at most~$i$.

\subsection{\texorpdfstring{\Sig}{Sigma}-Algebras and \texorpdfstring{\Sig}{Sigma}-coalgebras}
\label{subsE}\hfill

\smallskip\noindent
All the examples above (and a number of other interesting cases) are
subsumed in the following general case. Let \Sig~be a signature,
i.e., a set of operation symbols with given arities~$\arit(s)$ of symbols $s\in\Sig$; the
arity is a natural number. The classical \Sig-algebras
are the algebras for the corresponding \textit{polynomial functor}
\[H_{\Sig}X=\coprod_{\sigma \in \Sigma} X^{\arit(\sigma)}.\]
Coalgebras for~$H_{\Sig}$ are called \textit{\Sig-coalgebras}.

\begin{example}
\label{E-Sig} Let $I$ be a set of cardinality $n$.
Deterministic automata $HX=X^I\times\{0,1\}=X^I+X^I$ are given by two
$n$-ary operations. Streams, $HX=X\times I+1$, are
given by~$n$ unary operations and a
constant. Binary trees $HX=X\times X+1$ are given by one binary
operation and one constant.
\end{example}

\begin{definition}
\label{D-Sig}
A \textbf{\Sig-tree} is an ordered tree with nodes labeled in~\Sig\
so that every node with $n$~children has a label of arity~$n$. We
consider \Sig-trees up to isomorphism.
\end{definition}

Observe that every \Sig-tree~$T$ is a coalgebra: the function
$a\colon T\to H_{\Sig}T$ takes every node~$x$ labelled by a symbol
$\sigma\in\Sig$ (of arity~$n$) to the $n$-tuple~$(x_i)_{i<n}$ of its
children, an element of the $\sigma$-summand~$T^n$ of~$H_{\Sig}T$.

In general a \Sig-coalgebra $a\colon Q\to H_{\Sig}Q$ can be viewed
as a system with a state set~$Q$ labeled in~\Sig:
\[\bar{a}\colon Q\to\Sig\]
and such that every state $q\in Q$ with $n$-ary label has ``next
states'' forming an $n$-tuple
\[ a^{*}(q)\in Q^n.\]
Indeed, to give a function $a\colon Q\to H_{\Sig}Q$ means precisely
to given a pair~$(\bar{a}, a^{*})$ of functions as above.

\begin{definition}
\label{D-exp}
Let $a\colon Q\to H_{\Sig}Q$ be a $\Sigma$-coalgebra.
\begin{enumerate}[\rm (a)]
\item A \textbf{computation} of length~$n$ is a word $i_0\cdots i_{n-1}$
  in $\nn^*$ for which there are states~$q_0,\cdots, q_n$ in~$Q$ with
  \[q_{k+1}=\text{the $i_k$-component of $ a^{*}(q_k)$}\qquad
  (k=0,\dots,n-1).\]
\item The \textbf{tree expansion} of a state~$q$ is the \Sig-tree
  \[T_q\]
  of all computations with initial state~$q$. The label of a
  computation is~$\bar{a}(q_n)$, where $q_n$~is its last state. And the
  children are all one-step extensions of that computation, i.e., all
  words~$i_0\dots i_{n-1}j$ with  $j<\arit(\bar{a}(q_n))$.
\end{enumerate}
\end{definition}

\begin{lemma}
\label{L-Sig}
The greatest congruence on a \Sig-coalgebra merges precisely the
pairs of states with the same tree expansion.
\end{lemma}

\begin{proof}
Let $(Q,\bar{a}, a^{*})$~be a \Sig-coalgebra and put $q\sim q'$
iff $T_q=T_{q'}$. Then we have a coalgebra structure on~$Q/{\sim}$:
the label of~$[q]$ is~$\bar{a}(q)$, independent of the
representative. The next-state $n$-tuple is~$([q_i])_{i<n}$ where
$ a^{*}(q)=(q_i)$. It is easy to see that this is
independent of the choice of representatives. And the quotient map is
a coalgebra homomorphism from~$Q$ to~$Q/{\sim}$. Thus, $\sim$~is a
congruence.

To prove that $\Sig$~is the largest congruence, observe that given a
coalgebra homomorphism $h\colon Q\to Q'$, then for every state $q\in
Q$ we have $T_q=T_{h(q)}$. Indeed, an isomorphism $i\colon T_q\to
T_{h(q)}$ is easy to define by induction on the depth of nodes
of~$T_q$.
\end{proof}

\begin{corollary}
\label{C-Sig}
Well-pointed \Sig-coalgebras are the \Sig-coalgebras with an initial
state~$q_0$ which are
\begin{enumerate}[\rm (a)]
\item reachable: every state can be reached from~$q_0$ by a
  computation, and
\item simple: different states have different tree expansions.
\end{enumerate}
\end{corollary}

\begin{example}
\label{E-Sig2}
For every \Sig-tree~$T$ the equivalence on the nodes of~$T$ given by
\begin{equation}
\label{r-cj}
x\sim y\qquad\text{iff}\qquad T_x\cong T_y,
\end{equation}
where $T_x$~is the subtree of $T$ rooted at node~$x$, is a congruence. And
$T/{\sim}$~carries an obvious structure of a \Sig-coalgebra.
Let
$[r]$~be the congruence class of the root, then the pointed
\Sig-coalgebra $(T/{\sim},[r])$ is well-pointed.

Indeed, this pointed coalgebra is reachable: given a node~$q$ of~$T$
let $i_0\cdots i_{n-1}$ be the unique path from~$r$ to~$q$, then
$i_0\cdots i_{n-1}$~is a computation in~$T/{\sim}$ with initial
state~$[r]$ and terminal state~$[q_n]$.

The simplicity of~$T/{\sim}$ follows from Lemma~\ref{L-Sig} and the
observation that the tree expansion of a state~$[q]$ of~$T/{\sim}$ is
the subtree~$T_q$ of~$T$.

These are all well-pointed \Sig-coalgebras:
\end{example}

\begin{proposition}
\label{P-Sig}
Every well-pointed coalgebra is isomorphic to~$(T/{\sim},[r])$ for a
unique \Sig-tree~$T$.
\end{proposition}
\begin{proof}
  It is well-known that the coalgebra $\tau\colon B \to H_\Sigma B$ of
  all $\Sigma$-labeled trees where $\tau$ is given by
  \begin{align*}
    \bar \tau(T) &= \text{label of the root of $T$, and} \\
      \tau^*(T) &= (T_i)_{i<n},
  \end{align*}
  where $T_i$ is the $i$-th maximum subtree is final. Indeed, for
  every coalgebra $(Q, a)$ the unique coalgebra homomorphism $h: Q \to
  B$ is given by tree expansion (see Definition~\ref{D-exp}): $h(q) =
  T_q$. 

  Now from Theorem~\ref{T-final} we have the final coalgebra $\nu
  H_\Sigma$ of all well-pointed coalgebras. The tree expansion map
  $h: \nu H_\Sigma \to B$ is then an isomorphism. The inverse $h^{-1}$
  takes a tree $T$ to the coalgebra $(T/\mathord{\sim}, [r])$ above:
  this is immediate from the fact that the tree expansion of $[r]$ in
  $T/\mathord{\sim}$ is $T$. 
\end{proof}

\takeout{ 
\begin{proof}
Let $(Q,a,q_0)$~be a well-pointed \Sig-coalgebra.

(a)
Existence.
Let $T$~denote the tree expansion of~$q_0$. For the 
equivalence~\eqref{r-cj} above, we prove that two equivalent computations
always terminate in the same state. This follows from the simplicity
of~$(Q,a)$: Denote by~$\approx$ the equivalence with $q\approx q'$ iff
there exist two equivalent (under~$\sim$) computations with terminal
states $q$ and~$q'$ (respectively). This is clearly a congruence
on~$(Q,a)$, so $q=q'$. We thus obtain a function
\[t\colon T/{\sim}\to Q,\qquad [i_0\cdots
i_{n-1}]\to/|->/\text{terminal state of $i_0\cdots i_{n-1}$.}\]
It is easy to see that this is a coalgebra homomorphism. Since it
takes~$[q_0]$ to~$q_0$, it is surjective (use the reachability). And
it is an isomorphism since two computations $p$ and~$p'$ with the
same last state~$q_n$ fulfil $T_p=T_{p'}$. (Indeed, the subtree
of~$T$ at the node~$p$ is precisely the tree~$T_{q_n}$). Hence, if
$t$~merges the equivalence classes $[p]$ and~$[p']$, then $p\sim p'$.

(b)
Uniqueness.
This follows from the observation that two \Sig-trees $T$ and~$T'$
are isomorphic whenever the coalgebras of Example~\ref{E-Sig2} are.
Indeed, given an isomorphism $f\colon (T/{\sim},[r])\to
(T'/{\sim},[r'])$, define an isomorphism $g\colon T\to T'$ from top
down. Since $f$~preserves labels, $r$ and~$r'$ are labeled by the
same $n$-ary label. We put $g(r)=r'$. Let $ a^{*}(r)=(x_i)_{i<n}$ and
$\widetilde{a'}(r')=(x'_i)_{i<n}$.
Since $f$~preserves~$ a^{*}$ we have $f[x_i]=[x'_i]$ for
all~$i$. We define~$g$ on level~1 by $g(x_i)=x'_i$, $i<n$, and
proceed recursively.
\end{proof}
}

\begin{proposition}
\label{P-tree}
A \Sig-coalgebra is well-founded iff all its tree-expansions are
\textbf{well-foun\-ded \Sig-trees}, i.e., \Sig-trees with no infinite
path.
\end{proposition}

\begin{proof}
Given a \Sig-coalgebra~$A$
let $m\colon A'\to/^{ (}->/A$ be the subset of all states $q\in A$
with $T_q$~well-founded. This is, obviously, a subcoalgebra. And it
is cartesian
\[\bfig
\square/>```>/<800,400>[A'`H_{\Sig}A'`A`H_{\Sig}A;%
a'`m`H_{\Sig}m`a]
\morphism(0,400)/{_(->}/<0,-400>[\strut`A;]
\morphism(800,400)/{_(->}/<0,-400>[\strut`H_{\Sig}A;]
\POS(100,300)
\ar@{-}+(-50,0)
\ar@{-}+(0,50)
\efig\]
Indeed, if a state~$q$ has the property that all components
of~$ a^{*}$ lies in~$A'$, the $q$~lies in~$A'$. Thus $A$~is
well-founded iff $A=A'$.
\end{proof}

\begin{definition}[see \cite{Gi}]
\label{D-rat-2}
A \Sig-tree is called \textbf{rational} if it has up to isomorphism
only finitely many subtrees.
\end{definition}

\begin{example}
\label{E-rat-2}
Given a finite \Sig-coalgebra, all tree expansions of its states are
rational.

Indeed, if $Q=\{q_1,\dots,q_n\}$ is the state set, then every subtree
of~$T_{q_i}$ (given by a computation with initial state~$q_i$) has
the form~$T_{q_j}$: take~$q_j$ to be the terminal state of the
computation.
\end{example}

\begin{corollary}
\label{C-Sig-2}
For every finitary signature~\Sig\ we have
\begin{align*}
\nu H_{\Sig}&\cong\text{all \Sig-trees},\\
\varrho H_{\Sig}&\cong\text{all rational \Sig-trees},\\
\mu H_{\Sig}&\cong\text{all finite \Sig-trees.}
\end{align*}
The coalgebra structure is in each case inverse to tree-tupling.
\end{corollary}

Indeed, the isomorphism between $\nu H_{\Sig}$ and all \Sig-trees is
given by Proposition~\ref{P-Sig}. And the coalgebra structure of
Remark~\ref{R-final}
corresponds to the inverse of tree-tupling,
i.e., it assigns to a \Sig-tree~$T$ with $ a^{*}(r)=(x_1,\ldots,
x_n)$ the $n$-tuple $(T_{x_1},\ldots,T_{x_n})$ in the $\sigma$-summand
of~$H_{\Sig}(\nu H_{\Sig})$ where $\sigma$~is the label of the
root.

Finally, the isomorphism between $\varrho H_{\Sig}$ (all finite
well-pointed coalgebras) and rational \Sig-trees follows from
Proposition~\ref{P-Sig} and Example~\ref{E-rat-2}. The last item
follows from K\"onig's Lemma: every well-founded finitely branching
tree is finite.

\begin{example}
For the functor $HX=X^*$ we can use nonlabeled trees: we have
\begin{align*}
\nu H&\cong\text{all finitely branching trees}\\
\varrho H&\cong\text{all rational finitely branching trees}\\
\mu H&\cong\text{all finite trees.}
\end{align*}
\end{example}

Indeed, let \Sig~be the signature with one $n$-ary operation for
every $n\in\nn$. Then $H_{\Sig}X\cong X^*$. And \Sig-trees need no
labeling, since operations already differ by arities.

\takeout{%
\begin{example}
All previous examples \mbox{\ref{subsA}--\ref{subsD}} are
special cases of
Corollary~\ref{C-Sig-2}:

\begin{enumerate}[(a)]
\item
Moore automata. The functor
\[ HX=X^I\times J\]
corresponds to a signature~\Sig\ of~$|J|$~operations, all of arity~$|I|$.
Since no nullary operation is given, every \Sig-tree is a complete
$I$-ary tree labeled by~$J$. Now the usual representation of the
complete $I$-ary tree is by~$I^*$: the root is the empty word
and the children
of~$i_1\cdots i_n$ are all~$i_1\cdots i_nj$ for $i\in I$. Thus, a
complete binary tree labeled by~$J$ is nothing else than a function
from~$I^*$ to~$J$:
\[\nu H=J^{I^*}\qquad\text{(see Corollary \ref{C-Moore}).}\]
The subcoalgebra~$\varrho H$ is then given by the regular functions.
And since we have no leaves,
\[\mu H=\emptyset\]
because a well-founded tree always has leaves.

\item
Mealy automata. The functor
\[HX=(X\times J)^I=X^I\times J^I\]
corresponds to $J^I$~operations, all of arity~$I$. Thus, a \Sig-tree
is a function from~$I^*$ to~$J^I$. Or, by uncurrying, a function from
$I^*\times I=I^+$ to~$J$:
\[\nu H=J^{I^+}\qquad\text{(see Remark \ref{R-Mealy-3}).}\]

\item
Streams. The functor
\[HX=X\times I+1\]
corresponds to one constant and $I$~unary operations. A \Sig-tree is
then a labeled tree consisting either of a path of length
$n\in\nn$, or an infinite path. The nodes are labeled by~$I$. Thus,
\[\nu H\cong1+I+I^2+\dots+I^\omega\qquad\text{(see Corollary
\ref{C-stream-2}).}\]

\item
Binary trees. Here
\[HX=X\times X+1\]
is given by a binary operation and a constant. Thus
\[\nu H=\text{binary trees\qquad(see Corollary \ref{C-tree-2}).}\]
\end{enumerate}
\end{example}
}%

\subsection{Graphs}
\label{subsF}\hfill

\smallskip\noindent
Here we investigate coalgebras for the power-set functor~\PP. In the
rest of Section~\ref{ctvrta} all trees are understood to be
non-ordered. That is, a tree is a directed graph with a node (root)
from which every node can be reached by a unique path.

Recall the concept of a \textit{bisimulation} between graphs $X$
and~$Y$: it is a relation $R\subseteq X\times Y$ such that whenever
$x\mathrel{R}y$ then every child of~$x$ is related to a child of~$y$,
and vice versa. Two nodes of a graph~$X$ are called
\textit{bisimilar} if they are related by a bisimulation $R\subseteq
X\times X$.

\begin{lemma}
\label{L-graph}
The greatest congruence on a graph merges precisely the bisimilar
pairs of states.
\end{lemma}

This follows, since \PP~preserves weak pullbacks, from general
results of Rutten~\cite{R1}.

\begin{corollary}
\label{C-graph}
A pointed graph~$(G,q_0)$ is well-pointed iff it is
\begin{enumerate}[\rm (a)]
\item reachable: every vertex can be reached from~$q_0$ by a directed
  path, and
\item simple: all distinct pairs of states are non-bisimilar.
\end{enumerate}
\end{corollary}

\begin{example}
\label{E-graph}
Aczel~\cite{Ac} introduced the \textit{canonical picture} of a
(well-founded) set~$X$. It is the graph with vertices all sets~$Y$
such that a sequence
\[Y=Y_0\in Y_1\in\dots\in Y_n=X\]
of sets exists. The neighbors of a vertex~$Y$ are all of its
elements. When pointed by~$X$, this is a well-pointed graph. Indeed,
reachability is clear. And suppose $R$~is a bisimulation and
$Y\mathrel{R}Y'$, then we prove $Y=Y'$. Assuming the contrary, there
exists $Z_0\in Y$ with $Z_0\notin Y'$, or vice versa. Since $R$~is a
bisimulation, from $Z_0\in Y$ we deduce that $Z'_0\in Y'$ exists with
$Z_0\mathrel{R}Z'_0$. Clearly $Z_0\neq Z'_0$. Thus, we
substitute~$(Y,Y')$ by~$(Z_0,Z'_0)$ and obtain $Z_1\in Z_0$ and
$Z'_1\in Z'_0$ with $Z_1\mathrel{R}Z'_1$ but $Z_1\neq Z'_1$ etc. This
is a contradiction to the well-foundedness of $X$: we get an infinite sequence~$Z_n$ with
\[\cdots\, Z_2\in Z_1\in Z_0\in Y.\]
\end{example}

Here are some concrete examples of canonical pictures and their
corresponding tree expansions (cf.~Remark~\ref{R-exp} below):

\[\begin{array}{lcl}
\text{Set:}&\text{Canonical picture:}&\text{Tree expansion}\\ \\
0=\emptyset&\bfig
\place(60,0)[\scriptstyle\bullet]
\place(0,0)[0]
\morphism(-150,0)/^{(}->/<100,0>[`;]
\efig&\bfig
\place(60,0)[\scriptstyle\bullet]
\place(0,0)[0]
\efig\\ \\
1=\{0\}&\bfig
\place(500,0)[1]
\morphism(350,0)/^{(}->/<100,0>[`;]
\morphism(560,0)<250,0>[\scriptstyle\bullet`\scriptstyle\bullet;]
\efig
&
\bfig
\place(0,100)[1]
\POS(60,100)
\ar@{-}+(0,-100)
\place(60,100)[\scriptstyle\bullet]
\place(60,0)[\scriptstyle\bullet]
\efig\\ \\
2=\{0,1\}&\bfig
\place(250,0)[2]
\morphism(100,0)/^{(}->/<100,0>[`;]
\morphism(310,0)<250,0>[\scriptstyle\bullet`\scriptstyle\bullet;]
\morphism(560,0)<250,0>[\scriptstyle\bullet`\scriptstyle\bullet;]
\morphism(310,0)/{@{->}@/^7pt/}/<500,0>%
[\scriptstyle\bullet`\scriptstyle\bullet;]
\efig&\bfig
\place(0,100)[\scriptstyle\bullet]
\place(200,0)[\scriptstyle\bullet]
\place(100,200)[\scriptstyle\bullet]
\place(200,100)[\scriptstyle\bullet]
\POS(100,200)
\ar@{-}+(-100,-100)
\ar@{-}+(100,-100)
\POS(200,100)
\ar@{-}+(0,-100)
\place(30,200)[2]
\efig\\ \\
3=\{0,1,2\}&\rule{.6em}{0em}\bfig
\place(0,0)[3]
\morphism(-150,0)/^{(}->/<100,0>[`;]
\morphism(60,0)<250,0>[\scriptstyle\bullet`\scriptstyle\bullet;]
\morphism(310,0)<250,0>[\scriptstyle\bullet`\scriptstyle\bullet;]
\morphism(560,0)<250,0>[\scriptstyle\bullet`\scriptstyle\bullet;]
\morphism(310,0)/{@{->}@/^7pt/}/<500,0>%
[\scriptstyle\bullet`\scriptstyle\bullet;]
\morphism(60,0)/{@{->}@/^-6pt/}/<500,0>%
[\scriptstyle\bullet`\scriptstyle\bullet;]
\morphism(60,0)/{@{->}@/^-12pt/}/<750,0>%
[\scriptstyle\bullet`\scriptstyle\bullet;]
\efig&\bfig
\place(700,100)[\scriptstyle\bullet]
\place(900,0)[\scriptstyle\bullet]
\place(800,200)[\scriptstyle\bullet]
\place(900,100)[\scriptstyle\bullet]
\POS(800,200)
\ar@{-}+(-100,-100)
\ar@{-}+(100,-100)
\POS(900,100)
\ar@{-}+(0,-100)
\place(630,300)[3]
\place(700,300)[\scriptstyle\bullet]
\place(600,200)[\scriptstyle\bullet]
\POS(700,300)
\ar@{-}+(-100,-100)
\ar@{-}+(100,-100)
\efig\\ \\
\omega&\bfig
\place(-50,10)[\dots]
\morphism(60,0)<250,0>[\scriptstyle\bullet`\scriptstyle\bullet;]
\morphism(310,0)<250,0>[\scriptstyle\bullet`\scriptstyle\bullet;]
\morphism(560,0)<250,0>[\scriptstyle\bullet`\scriptstyle\bullet;]
\morphism(310,0)/{@{->}@/^12pt/}/<500,0>%
[\scriptstyle\bullet`\scriptstyle\bullet;]
\morphism(60,0)/{@{->}@/^-6pt/}/<500,0>%
[\scriptstyle\bullet`\scriptstyle\bullet;]
\morphism(60,0)/{@{->}@/^-12pt/}/<750,0>%
[\scriptstyle\bullet`\scriptstyle\bullet;]
\morphism(560,300)/-/<0,-180>[\scriptstyle\bullet`;]
\morphism(560,0)/<-/<0,150>[\scriptstyle\bullet`;]
\morphism(560,300)/{@{->}@/^4pt/}/<250,-300>%
[\scriptstyle\bullet`\scriptstyle\bullet;]
\morphism(560,300)/{@{->}@/^-2pt/}/<-250,-300>%
[\scriptstyle\bullet`\scriptstyle\bullet;]
\morphism(560,300)/{@{->}@/^-2pt/}/<-500,-300>%
[\scriptstyle\bullet`\scriptstyle\bullet;]
\place(500,340)[\omega]
\morphism(350,340)/^{(}->/<100,0>[`;]
\efig&\bfig
\place(700,100)[\scriptstyle\bullet]
\place(900,0)[\scriptstyle\bullet]
\place(800,200)[\scriptstyle\bullet]
\place(900,100)[\scriptstyle\bullet]
\place(500,100)[\scriptstyle\bullet]
\place(300,200)[\scriptstyle\bullet]
\place(500,200)[\scriptstyle\bullet]
\place(400,300)[\scriptstyle\bullet]
\place(300,400)[\scriptstyle\bullet]
\place(200,300)[\scriptstyle\bullet]
\place(200,200)[\scriptstyle\bullet]
\place(50,300)[\scriptstyle\bullet]
\POS(800,200)
\ar@{-}+(-100,-100)
\ar@{-}+(100,-100)
\POS(900,100)
\ar@{-}+(0,-100)
\place(810,320)[\dots]
\place(260,440)[\omega]
\place(700,300)[\scriptstyle\bullet]
\place(600,200)[\scriptstyle\bullet]
\POS(400,300)
\ar@{-}+(-100,-100)
\ar@{-}+(100,-100)
\POS(700,300)
\ar@{-}+(-100,-100)
\ar@{-}+(100,-100)
\POS(500,200)
\ar@{-}+(0,-100)
\POS(300,400)
\ar@{-}+(-100,-100)
\ar@{-}+(100,-100)
\ar@{-}+(-250,-100)
\ar@{-}+(400,-100)
\POS(200,300)
\ar@{-}+(0,-100)
\efig
\end{array}\]

\begin{remark}
\label{R-exp}
Given a vertex~$q$ of a graph, its \textit{tree expansion} is
(similarly to the ordered case, see Definition~\ref{D-exp}) the
non-ordered tree
\[T_q\]
whose nodes are all finite directed paths from~$q$.

The children of a node~$p$ are all one-step extensions of the
path~$p$. The root is~$q$ (considered as the path of length~0).

For every pointed graph~$(G,x)$ the \textit{tree expansion}
is the
tree~$T_x$. In the previous example we saw tree expansions of the
given pointed graphs.
\end{remark}

\begin{definition}[Worrell~\cite{W}]
\label{D-bisim}
By a \textbf{tree-bisimulation} between trees $T_1$ and~$T_2$ is
meant a graph bisimulation $R\subseteq T_1\times T_2$ which
\begin{enumerate}[\rm (a)]
\item relates the roots, 
\item $x_1\mathrel{R}x_2$ implies that $x_1$ and $x_2$ are the roots
  or have related  parents, and
\item $x_1\mathrel{R}x_2$ implies that the depths of $x_1$ and $x_2$
  are equal.
\end{enumerate}
A tree~$T$ is called \textbf{strongly extensional} iff every tree
bisimulation $R\subseteq T\times T$ is trivial: $R\subseteq\Delta_T$.
\end{definition}

\begin{example}
\label{T-bisim}
The tree expansion of a well-pointed graph~$(G,q_0)$ is strongly
extensional. Indeed, given a tree bisimulation $R\subseteq
T_{q_0}\times
T_{q_0}$, we obtain a graph bisimulation $\bar{R}\subseteq G\times G$
consisting of all pairs~$(q_1,q_2)$ of vertices for which paths~$p_i$
from~$q_0$ to~$q_i$ exist, $i=1,2$, with $p_1\mathrel{R}p_2$. Since
$G$~is simple, $\bar{R}\subseteq\Delta$. Thus, for all pairs $(p_1,p_2)$ of
paths:
\[\text{if $p_1\mathrel{R}p_2$ then the last vertices of $p_2$ and
$p_1$ are equal.}\]
We prove $p_1\mathrel{R}p_2$ implies $p_1=p_2$ by induction on the
maximum~$k$ of the lengths of $p_1$ and~$p_2$. For $k=0$ we have
$p_1=q_0=p_2$. For ${k+1}$ we have $p'_1\mathrel{R}p'_2$ where
$p'_i$~is the trimming of~$p_i$ by one edge (since $R$~is a tree
bisimulation). Then $p'_1=p'_2$ implies $p_1=p_2$ because the last
vertices are equal.
\end{example}

Furthermore, there are no other extensional trees:

\begin{proposition}
\label{P-bisim}
Every strongly extensional tree is the tree expansion of a unique  (up to isomorphism)
well-pointed graph.
\end{proposition}

\begin{proof}
Let $T$~be a strongly extensional tree with root~$r$, considered as a coalgebra for \PP.

(a) Existence.  The coalgebra~$(T/{\sim},[r])$ where $\sim$~merges
bisimilar vertices of~$T$ is well-pointed by Lemma~\ref{L-graph}. Its
tree expansion $T'=(T/{\sim})_{[r]}$ is (isomorphic to) the given
tree~$T$. Indeed, the relation $R\subseteq T\times T'$ of all
pairs~$(x,p)$ where $x$~is a node of~$T$ and $p$~is the equivalence
class of the unique path from~$r$ to~$x$ is clearly a tree
bisimulation. Since \PP~preserves weak pullbacks, it follows that the
composite~${R\circ R^{-1}}$ of $R$ and~$R^{-1}$ is also a tree
bisimulation, see \cite{R1}. But $T$~is strongly extensional, thus
$R\circ R^{-1}\subseteq\Delta$. Also $T'$~is strongly extensional, see
Example~\ref{T-bisim}, thus $R^{-1}\circ R\subseteq\Delta$. Since for
every $x$ there is a pair $(x,p)$ in $R$, we conclude that $R$~is (the
graph of) an isomorphism from~$T$ to~$T'$.

(b)
Uniqueness:
If well-pointed graphs $(G,q_0)$ and~$(G',q'_0)$ have isomorphic tree
expansions, then they are isomorphic. Arguing analogously to~(a) we
only need to find a graph bisimulation $R\subseteq G\times G'$ and
use the simplicity of $G$ and~$G'$. For that, we just observe that
there is a graph bisimulation between $(G,q_0)$ and~$T_{q_0}$: the
relation $R\subseteq G\times T_{q_0}$ of all pairs~$(q,p)$ where
$q\in G$ is the last vertex of the path~$p$ from~$q_0$ to~$q$.
\end{proof}

\begin{corollary}
\label{C-bisim}
$\nu\PP=\mbox{}$all strongly extensional trees.
\end{corollary}

We must be careful here: \PP~has no fixed points. But recall the
extension of set functors to classes in Remark~\ref{R-class}. For~\PP\
this is the functor ${\PP}^{*}X=\{A;A\text{ is a set with }A\subseteq
X\}$.
Its (large) final coalgebra is the coalgebra of all (small) strongly
extensional trees.

\begin{notation}
\label{N-lambda}
Let $\PP_\lambda$~be the subfunctor of all subsets of cardinality
less than~$\lambda$. (Thus $\PP_\omega$~is the finite power-set
functor.) Then by precisely the same argument as above one proves
\end{notation}

\begin{corollary}
\label{C-lambda}
For every cardinal~$\lambda$,
\[
\nu\PP_\lambda=\text{all $\lambda$-branching strongly extensional
trees.}
\]
\end{corollary}

This was proved for $\lambda=\omega$ by Worrell~\cite{W} and for
general~$\lambda$ by Schwencke~\cite{S}. Our proof is entirely
different.\medskip

We know from Example~\ref{E-wf}(1) that the well-founded graphs
are precisely the graphs without an infinite directed path.
Now strong extensionality can, in the case of well-founded trees, be
simplified to \textit{extensionality} which says that for every node
different children define non-isomorphic subtrees. Thus we get

\begin{corollary}
\label{C-PP}
$
\begin{array}[t]{r@{\ }c@{\ }l}
  \mu\PP & = & \text{all well-founded, extensional trees;} \\
  \mu\PP_\lambda & = & \text{all $\lambda$-branching, well-founded, extensional trees.}
\end{array}
$
\end{corollary}

Analogously to Example~\ref{E-rat-2} the rational  fixed
point of the finite-powerset functor~$\PP_\omega$ consists of all rational strongly
extensional trees, i.e., those with finitely many subtrees up to
isomorphism:

\begin{corollary}
\label{C-finite}
For the finite power-set functor~$\PP_\omega$ we have
\begin{align*}
\nu\PP_\omega&=\text{all finitely branching, strongly extensional
trees,}\\
\varrho\PP_\omega&=\text{all finitely branching, rational, strongly
extensional trees, and} \\
\mu\PP_\omega&=\text{all finite extensional trees.}
\end{align*}
\end{corollary}

\subsection{Sets and non-well-founded sets}
\label{subsG}\hfill

\smallskip\noindent
We revisit $\mu\PP$ and~$\nu\PP$ here from a set-theoretic
perspective. Before coming to the non-well-founded sets, let us
observe that Example~\ref{E-graph} has the following strengthening:

\begin{lemma}
\label{L-non}
Well-founded, well-pointed graphs are precisely the canonical
pictures of well-founded sets.
\end{lemma}

This follows from the standard fact from set theory
that every well-pointed graph~$G$ has a unique
\textit{Mostowski collapse}, also called a
 \textit{decoration} in Aczel \cite[see Introduction]{Ac}, i.e.,
coalgebra homomorphism~$d$ to the class~\Set\ of sets considered as a graph with $\in$~as
the neighborhood relation. That is, $d$~assigns to every vertex~$x$
a set~$d(x)$ as follows:
\[d(x)=\bigl\{d(y);y\in G\text{ a neighbor of }x\bigr\}.\]
Observe that the kernel of~$d$ is clearly a congruence on~$G$. Thus,
given a well-pointed, well-founded graph~$(G,q_0)$, we know that
$d$~is monic. From that it follows that the canonical picture of the
set~$d(q_0)$ is isomorphic to~$(G,q_0)$.

\begin{corollary}
\label{C-non}
$\mu\PP=\mbox{}$ the class of all sets.
\end{corollary}

This was proved by Rutten and Turi in~\cite{RT}. The bijection
between well-founded, well-pointed graphs and sets (given by the
canonical picture) takes the finite graphs to the \textit{hereditarily
finite sets}, i.e., finite sets with finite elements which also
have finite elements, etc. More precisely: a set $X$ is hereditarily finite
if
all sets in the canonical
picture of~$X$ are finite:

\begin{corollary}
\label{C-her}
$\mu\PP_\omega=\mbox{}$all hereditarily finite sets.
\end{corollary}

In order to describe the final coalgebra for~\PP\ in a similar
set-theoretic manner, we must move from the classical theory to the
non-well-founded set theory of Aczel~\cite{Ac}. 
\takeout{%
Recall that a {\bf decoration} of a graph is a coalgebra homomorphism
from this graph into the large coalgebra $({\Set},{\in})$.
}%
Non-well-founded set theory is obtained by swapping the axiom of
foundation, telling us that $({\Set},{\in})$~is well-founded, with the following

\medskip
\noindent
{\bf Anti-foundation axiom.} Every graph has a unique decoration.

\begin{example}
\label{E-non}
The decoration of a single loop is a set~$\Omega$ such that
$\Omega=\{\Omega\}$.

\vspace{1mm}

The coalgebra~$({\Set},{\in})$ where now \Set~is the class of all
non-well-founded sets, is of course final for \PP: the
decoration of any graph~$G$ is the unique homomorphism $d\colon G\to\Set$.
\end{example}

\begin{corollary}
\label{C-all}
In the non-well-founded set theory
\[\nu\PP=\text{all sets.}\]
\end{corollary}

Let us turn to the finite power-set functor~$\PP_\omega$. Its final
coalgebra consists of all sets whose canonical picture is finitely
branching. They are called $1$-\textit{hereditarily finite},
notation~$HF^1[\emptyset]$, in the monograph of Barwise and
Moss~\cite{BM}. The
rational fixed point of~$\PP_\omega$ consists of all sets whose
canonical picture is finite, they are called \textit{$1/2$-hereditary} in \cite{BM}.
The collection of these sets is denoted by  
$HF^{1/2}[\emptyset]$.
For well-founded
sets (with canonical picture
well-founded) the two collections coincide.

\begin{corollary}
\label{C-her2}
In the non-well-founded set theory
\[
\begin{array}{r@{\ }c@{\ }ll}
\nu\PP_\omega & = & HF^1[\emptyset], 
& \text{the $1$-hereditarily finite sets,}
\\
\varrho\PP_\omega & = & HF^{1/2}[\emptyset], 
& \text{the $1/2$-hereditarily finite sets, and}
\\
\mu\PP_\omega& = &\multicolumn{2}{@{}l}{\text{the well-founded, hereditarily finite sets.}}
\end{array}
\]
\end{corollary}

\subsection{Labeled transition systems}\hfill

\smallskip\noindent
Here we consider, for a set~$A$ of actions, labeled transition systems (LTS) as coalgebras for~$\PP({-}{\times} A)$.
A \textit{bisimulation} between two labeled transition systems 
$G$ and~$G'$ is a relation $R\subseteq G\times G'$ such that
\[\text{if}\quad x\mathrel{R}x'\quad\parbox[t]{84mm}{then for every
transition $x\to^a x'$ in $G$ there exists
$y' \in G'$  and a transition
$y\to^a y'$ with $x'\mathrel{R}y'$, and vice versa.}\]
States $x,y$ of an LTS are called \textit{bisimilar} if
$x\mathrel{R}y$ for some bisimulation $R\subseteq G\times G$.

\begin{lemma}
\label{L-LTS}
For every LTS the greatest congruence merges precisely the bisimilar
pairs of states.
\end{lemma}

This, again, follows from general results of Rutten~\cite{R1} since
$\PP({-}{\times}A)$~preserves weak pullbacks.

\begin{corollary}
\label{P-LTS}
An LTS together with an initial state $q_0$ is well-pointed iff it
 is
\begin{enumerate}[\rm (a)]
\item reachable: every state can be reached from~$q_0$ (by a
    sequence of actions), and
\item simple: distinct states are non-bisimilar.
\end{enumerate}
\end{corollary}

The \textit{tree expansion} of a state~$q$ is a (non-ordered) tree
with edges labeled in~$A$, shortly, an $A$-labeled tree. For
$A$-labeled trees we modify Definition~\ref{D-bisim} and speak about
\textit{tree bisimulation} if a bisimulation $R\subseteq T_1\times
T_2$ also fulfils~\mbox{(a)--(c)} of Definition~\ref{D-bisim}. An $A$-labeled
tree~$T$ is \textit{strongly extensional} iff every tree bisimulation
$R\subseteq T\times T$ is trivial.

\begin{proposition}
\label{P-LTS2}
Tree expansion is a bijection between well-pointed LTS and strongly
extensional $A$-labeled trees.
\end{proposition}

The proof is analogous to that of Proposition~\ref{P-bisim}.
Also the rest is analogous to the case of~\PP\ above:

\begin{corollary}
\label{C-LTS}
$
\begin{array}[t]{r@{\ }c@{\ }l}
  \nu\PP({-}{\times}A) & \cong & 
  \text{all strongly extensional $A$-labeled trees,}
  \\
  \nu\PP_\lambda({-}{\times}A) & \cong & 
  \text{all $\lambda$-branching, strongly extensional $A$-labeled trees.}
\end{array}
$
\end{corollary}

\begin{corollary}
\label{C-end}
For the finitely branching LTS we have
\begin{align*}
\nu\PP_\omega({-}{\times}A)&\cong\text{all finitely branching, strongly
extensional $A$-labeled trees,}\\
\varrho\PP_\omega({-}{\times}A)&\cong\text{all rational, finitely
  branching strongly extensional $A$-labeled trees,}\\ 
\mu\PP_\omega({-}{\times}A)&\cong\text{all finite extensional $A$-labeled
trees.}
\end{align*}
\end{corollary}

\section{Conclusion}

For functors~$H$ satisfying the (mild) assumption of
preservation of intersections
we described (a)~the final coalgebra as the set of all
well-pointed coalgebras, (b)~the initial algebra as the set of all well-pointed
coalgebras that are well-founded, and (c)~in the case where $H$~is finitary,
the initial iterative algebra as the set of all finite well-pointed
coalgebras. This is based on the observation that given an element of
a final coalgebra, the subcoalgebra it generates has no proper
subcoalgebras nor proper quotients---shortly, this subcoalgebra is
well-pointed. And different elements define non-isomorphic
well-pointed subcoalgebras. We then combined this with our result that for all set functors
the initial algebra is precisely the final well-founded coalgebra.
 (For set functors preserving inverse images this was proved by Taylor \cite{Ta1}.) More generally, for endofunctors of varieties preserving
	intersections we proved that the final coalgebra is carried
	by the sets of all well-pointed coalgebras, and the initial
	algebra is carried by the set of all well-founded, well-pointed
	coalgebras, and we presented  a concrete description. Numerous examples demonstrate that this
view of final coalgebras and initial algebras is useful in
applications. 

More generally, for functors preserving finite intersections the fact that initial algebras coincide with final well-founded coalgebras was
proved in locally finitely presentable categories. The
description of the final coalgebra was formulated concretely only in varieties of algebras. In future research we intend to generalize
this result to a wider class of base categories.

\end{document}

%% file: ldiagxy.tex
 \def\dated#1{\def\thedate{#1}}
 \dated{2004-03-08}

\newcount\atcode \atcode=\catcode`\@%
\catcode`\@=12
\input xy
\xyoption{arrow}
\xyoption{curve}

\newdir{ >}{{ }*!/-.9em/@{>}}
\newdir{ (}{{ }*!/-.5em/@{(}}
\newdir^{ (}{{ }*!/-.5em/@^{(}}
\newdir{< }{!/.9em/@{<}*{ }}

\newdimen\high%
\newdimen\ul%
\newcount\deltax%
\newcount\deltay%
\newcount\deltaX%
\newcount\deltaY%

\newdimen\wdth
\newcount\xend%
\newcount\yend%
\newcount\Xend
\newcount\Yend
\newcount\xpos%
\newcount\ypos%
\newcount\default \default=500%
\newcount\defaultmargin \defaultmargin=150
\newcount\topw%
\newcount\botw%
\newcount\Xpos%
\newcount\Ypos%
\def\ratchet#1#2{\ifnum#1<#2\global #1=#2\fi}%

\catcode`\@=11
\expandafter\ifx\csname @ifnextchar\endcsname\relax
\def\ifnextchar#1#2#3{\let\@tempe
#1\def\@tempa{#2}\def\@tempb{#3}\futurelet
    \@tempc\@ifnch}%
\def\@ifnch{\ifx \@tempc \@sptoken \let\@tempd\@xifnch
      \else \ifx \@tempc \@tempe\let\@tempd\@tempa\else\let\@tempd\@tempb\fi
      \fi \@tempd}%
\def\:{\let\@sptoken= } \:  
\def\:{\@xifnch} \expandafter\def\: {\futurelet\@tempc\@ifnch}%
\else
\let\ifnextchar\@ifnextchar
\fi
\ifx\check@mathfonts\undefined
\else \check@mathfonts
\fi
\newdimen\axis \axis=\fontdimen22\textfont2
\def\scalefactor#1{\ul=#1\ul \X@xbase=#1\X@xbase \Y@ybase=#1\Y@ybase}%
\catcode`\@=12%

\def\fontscale#1{%
\if#1h\relax
\font\xydashfont=xydash10 scaled \magstephalf
\font\xyatipfont=xyatip10 scaled \magstephalf
\font\xybtipfont=xybtip10 scaled \magstephalf
\font\xybsqlfont=xybsql10 scaled \magstephalf
\font\xycircfont=xycirc10 scaled \magstephalf
\else
\font\xydashfont=xydash10 scaled \magstep#1%
\font\xyatipfont=xyatip10 scaled \magstep#1%
\font\xybtipfont=xybtip10 scaled \magstep#1%
\font\xybsqlfont=xybsql10 scaled \magstep#1%
\font\xycircfont=xycirc10 scaled \magstep#1%
\fi}

\def\bfig{\vcenter\bgroup\xy}
\def\efig{\endxy\egroup}

\def\car#1#2\nil{#1}%

\def\morphism{\ifnextchar({\morphismp}{\morphismp(0,0)}}%
\def\morphismp(#1){\ifnextchar|{\morphismpp(#1)}{\morphismpp(#1)|a|}}%
\def\morphismpp(#1)|#2|{\ifnextchar/{\morphismppp(#1)|#2|}%
    {\morphismppp(#1)|#2|/>/}}%
\def\morphismppp(#1)|#2|/#3/{%
    \ifnextchar<{\morphismpppp(#1)|#2|/#3/}%
    {\morphismpppp(#1)|#2|/#3/<\default,0>}}%

\def\morphismpppp(#1,#2)|#3|/#4/<#5,#6>[#7`#8;#9]{%
\xend#1\advance \xend by #5%
\yend#2\advance \yend by #6%
\domorphism(#1,#2)|#3|/#4/<#5,#6>[{#7}`{#8};{#9}]}

\def\domorphism(#1,#2)|#3|/#4/<#5,#6>[#7`#8;#9]{%
\def\next{\car#4.\nil}%
\if@\next\relax
 \if#3l%
  \ifnum #6>0%
   \POS(#1,#2)*+!!<0ex,\axis>{#7}\ar#4^-{#9} (\xend,\yend)*+!!<0ex,\axis>{#8}%
  \else%
   \POS(#1,#2)*+!!<0ex,\axis>{#7}\ar#4_-{#9} (\xend,\yend)*+!!<0ex,\axis>{#8}%
  \fi%
 \else \if#3m%
    \setbox0\hbox{$#9$}%
   \ifdim \wd0=0pt
     \POS(#1,#2)*+!!<0ex,\axis>{#7}\ar#4 (\xend,\yend)*+!!<0ex,\axis>{#8}%
   \else
     \POS(#1,#2)*+!!<0ex,\axis>{#7}\ar#4|-*+<1pt,4pt>{\labelstyle#9}
       (\xend,\yend)*+!!<0ex,\axis>{#8}%
   \fi
 \else \if#3r%
  \ifnum #6<0%
   \POS(#1,#2)*+!!<0ex,\axis>{#7}\ar#4^-{#9} (\xend,\yend)*+!!<0ex,\axis>{#8}%
  \else%
   \POS(#1,#2)*+!!<0ex,\axis>{#7}\ar#4_-{#9} (\xend,\yend)*+!!<0ex,\axis>{#8}%
  \fi%
 \else \if#3a%
  \ifnum #5>0%
   \POS(#1,#2)*+!!<0ex,\axis>{#7}\ar#4^-{#9} (\xend,\yend)*+!!<0ex,\axis>{#8}%
  \else%
   \POS(#1,#2)*+!!<0ex,\axis>{#7}\ar#4_-{#9} (\xend,\yend)*+!!<0ex,\axis>{#8}%
  \fi%
 \else \if#3b%
  \ifnum #5<0%
   \POS(#1,#2)*+!!<0ex,\axis>{#7}\ar#4^-{#9} (\xend,\yend)*+!!<0ex,\axis>{#8}%
  \else%
   \POS(#1,#2)*+!!<0ex,\axis>{#7}\ar#4_-{#9} (\xend,\yend)*+!!<0ex,\axis>{#8}%
  \fi%
 \else
   \POS(#1,#2)*+!!<0ex,\axis>{#7}\ar#4 (\xend,\yend)*+!!<0ex,\axis>{#8}%
 \fi\fi\fi\fi\fi%
\else%
 \if#3l%
  \ifnum #6>0%
   \POS(#1,#2)*+!!<0ex,\axis>{#7}\ar@{#4}^-{#9} (\xend,\yend)*+!!<0ex,\axis>{#8}%
  \else%
   \POS(#1,#2)*+!!<0ex,\axis>{#7}\ar@{#4}_-{#9} (\xend,\yend)*+!!<0ex,\axis>{#8}%
  \fi%
 \else \if#3m%
    \setbox0\hbox{$#9$}%
   \ifdim \wd0=0pt
     \POS(#1,#2)*+!!<0ex,\axis>{#7}\ar@{#4} (\xend,\yend)*+!!<0ex,\axis>{#8}%
   \else
     \POS(#1,#2)*+!!<0ex,\axis>{#7}\ar@{#4}|-*+<1pt,4pt>{\labelstyle#9}
         (\xend,\yend)*+!!<0ex,\axis>{#8}%
   \fi
 \else \if#3r%
  \ifnum #6<0%
   \POS(#1,#2)*+!!<0ex,\axis>{#7}\ar@{#4}^-{#9} (\xend,\yend)*+!!<0ex,\axis>{#8}%
  \else%
   \POS(#1,#2)*+!!<0ex,\axis>{#7}\ar@{#4}_-{#9} (\xend,\yend)*+!!<0ex,\axis>{#8}%
  \fi%
 \else \if#3a%
  \ifnum #5>0%
   \POS(#1,#2)*+!!<0ex,\axis>{#7}\ar@{#4}^-{#9} (\xend,\yend)*+!!<0ex,\axis>{#8}%
  \else%
   \POS(#1,#2)*+!!<0ex,\axis>{#7}\ar@{#4}_-{#9} (\xend,\yend)*+!!<0ex,\axis>{#8}%
  \fi%
 \else \if#3b%
  \ifnum #5<0%
   \POS(#1,#2)*+!!<0ex,\axis>{#7}\ar@{#4}^-{#9} (\xend,\yend)*+!!<0ex,\axis>{#8}%
  \else%
   \POS(#1,#2)*+!!<0ex,\axis>{#7}\ar@{#4}_-{#9} (\xend,\yend)*+!!<0ex,\axis>{#8}%
  \fi%
 \else
   \POS(#1,#2)*+!!<0ex,\axis>{#7}\ar@{#4} (\xend,\yend)*+!!<0ex,\axis>{#8}%
 \fi\fi\fi\fi\fi
\fi\ignorespaces}%

\def\vector(#1,#2)/#3/<#4,#5>{%
 \xend#1 \yend#2 \advance\xend by #4 \advance\yend by #5
     \POS(#1,#2)\ar#3 (\xend,\yend)}

\def\squarepppp(#1,#2)|#3|/#4`#5`#6`#7/<#8>[#9]{%
\xpos#1\ypos#2%
\def\next|##1##2##3##4|{%
 \def\xa{##1}\def\xb{##2}\def\xc{##3}\def\xd{##4}\ignorespaces}%
\next|#3|%
\def\next<##1,##2>{\deltax=##1\deltay=##2\ignorespaces}%
\next<#8>%
\def\next[##1`##2`##3`##4;##5`##6`##7`##8]{%
    \def\nodea{##1}\def\nodeb{##2}\def\nodec{##3}\def\noded{##4}%
    \def\labela{##5}\def\labelb{##6}\def\labelc{##7}\def\labeld{##8}\ignorespaces}%
\next[#9]%
\morphism(\xpos,\ypos)|\xd|/{#7}/<\deltax,0>[\nodec`\noded;\labeld]%
\advance \ypos by \deltay%
\morphism(\xpos,\ypos)|\xb|/{#5}/<0,-\deltay>[\nodea`\nodec;\labelb]%
\morphism(\xpos,\ypos)|\xa|/{#4}/<\deltax,0>[\nodea`\nodeb;\labela]%
 \advance \xpos by \deltax%
\morphism(\xpos,\ypos)|\xc|/{#6}/<0,-\deltay>[\nodeb`\noded;\labelc]%
\ignorespaces}%

\def\square{\ifnextchar({\squarep}{\squarep(0,0)}}%
\def\squarep(#1){\ifnextchar|{\squarepp(#1)}{\squarepp(#1)|alrb|}}%
\def\squarepp(#1)|#2|{\ifnextchar/{\squareppp(#1)|#2|}%
    {\squareppp(#1)|#2|/>`>`>`>/}}%
\def\squareppp(#1)|#2|/#3`#4`#5`#6/{%
    \ifnextchar<{\squarepppp(#1)|#2|/#3`#4`#5`#6/}%
    {\squarepppp(#1)|#2|/#3`#4`#5`#6/<\default,\default>}}%

\def\ptrianglepppp(#1,#2)|#3|/#4`#5`#6/<#7>[#8]{%
\xpos#1\ypos#2%
\def\next|##1##2##3|{\def\xa{##1}\def\xb{##2}\def\xc{##3}}%
\next|#3|%
\def\next<##1,##2>{\deltax=##1\deltay=##2\ignorespaces}%
\next<#7>%
\def\next[##1`##2`##3;##4`##5`##6]{%
    \def\nodea{##1}\def\nodeb{##2}\def\nodec{##3}%
    \def\labela{##4}\def\labelb{##5}\def\labelc{##6}}%
\next[#8]%
\advance\ypos by \deltay%
\morphism(\xpos,\ypos)|\xa|/{#4}/<\deltax,0>[\nodea`\nodeb;\labela]%
\morphism(\xpos,\ypos)|\xb|/{#5}/<0,-\deltay>[\nodea`\nodec;\labelb]%
\advance\xpos by \deltax%
\morphism(\xpos,\ypos)|\xc|/{#6}/<-\deltax,-\deltay>[\nodeb`\nodec;\labelc]%
\ignorespaces}%

\def\qtrianglepppp(#1,#2)|#3|/#4`#5`#6/<#7>[#8]{%
\xpos#1\ypos#2%
\def\next|##1##2##3|{\def\xa{##1}\def\xb{##2}\def\xc{##3}}%
\next|#3|%
\def\next<##1,##2>{\deltax=##1\deltay=##2\ignorespaces}%
\next<#7>%
\def\next[##1`##2`##3;##4`##5`##6]{%
    \def\nodea{##1}\def\nodeb{##2}\def\nodec{##3}%
    \def\labela{##4}\def\labelb{##5}\def\labelc{##6}}%
\next[#8]%
\advance\ypos by \deltay%
\morphism(\xpos,\ypos)|\xa|/{#4}/<\deltax,0>[\nodea`\nodeb;\labela]%
\morphism(\xpos,\ypos)|\xb|/{#5}/<\deltax,-\deltay>[\nodea`\nodec;\labelb]%
\advance\xpos by \deltax%
\morphism(\xpos,\ypos)|\xc|/{#6}/<0,-\deltay>[\nodeb`\nodec;\labelc]%
\ignorespaces}%

\def\dtrianglepppp(#1,#2)|#3|/#4`#5`#6/<#7>[#8]{%
\xpos#1\ypos#2%
\def\next|##1##2##3|{\def\xa{##1}\def\xb{##2}\def\xc{##3}}%
\next|#3|%
\def\next<##1,##2>{\deltax=##1\deltay=##2\ignorespaces}%
\next<#7>%
\def\next[##1`##2`##3;##4`##5`##6]{%
    \def\nodea{##1}\def\nodeb{##2}\def\nodec{##3}%
    \def\labela{##4}\def\labelb{##5}\def\labelc{##6}}%
\next[#8]%
\morphism(\xpos,\ypos)|\xc|/{#6}/<\deltax,0>[\nodeb`\nodec;\labelc]%
\advance\ypos by \deltay\advance \xpos by \deltax%
\morphism(\xpos,\ypos)|\xa|/{#4}/<-\deltax,-\deltay>[\nodea`\nodeb;\labela]%
\morphism(\xpos,\ypos)|\xb|/{#5}/<0,-\deltay>[\nodea`\nodec;\labelb]%
\ignorespaces}%

\def\btrianglepppp(#1,#2)|#3|/#4`#5`#6/<#7>[#8]{%
\xpos#1\ypos#2%
\def\next|##1##2##3|{\def\xa{##1}\def\xb{##2}\def\xc{##3}}%
\next|#3|%
\def\next<##1,##2>{\deltax=##1\deltay=##2\ignorespaces}%
\next<#7>%
\def\next[##1`##2`##3;##4`##5`##6]{%
    \def\nodea{##1}\def\nodeb{##2}\def\nodec{##3}%
    \def\labela{##4}\def\labelb{##5}\def\labelc{##6}}%
\next[#8]%
\morphism(\xpos,\ypos)|\xc|/{#6}/<\deltax,0>[\nodeb`\nodec;\labelc]%
\advance\ypos by \deltay%
\morphism(\xpos,\ypos)|\xa|/{#4}/<0,-\deltay>[\nodea`\nodeb;\labela]%
\morphism(\xpos,\ypos)|\xb|/{#5}/<\deltax,-\deltay>[\nodea`\nodec;\labelb]%
\ignorespaces}%

\def\Atrianglepppp(#1,#2)|#3|/#4`#5`#6/<#7>[#8]{%
\xpos#1\ypos#2%
\def\next|##1##2##3|{\def\xa{##1}\def\xb{##2}\def\xc{##3}}%
\next|#3|%
\def\next<##1,##2>{\deltax=##1\deltay=##2\ignorespaces}%
\next<#7>%
\def\next[##1`##2`##3;##4`##5`##6]{%
    \def\nodea{##1}\def\nodeb{##2}\def\nodec{##3}%
    \def\labela{##4}\def\labelb{##5}\def\labelc{##6}}%
\next[#8]%
\multiply\deltax by 2%
\morphism(\xpos,\ypos)|\xc|/{#6}/<\deltax,0>[\nodeb`\nodec;\labelc]%
\divide\deltax by 2
\advance\ypos by \deltay\advance\xpos by \deltax%
\morphism(\xpos,\ypos)|\xa|/{#4}/<-\deltax,-\deltay>[\nodea`\nodeb;\labela]%
\morphism(\xpos,\ypos)|\xb|/{#5}/<\deltax,-\deltay>[\nodea`\nodec;\labelb]%
\ignorespaces}%

\def\Vtrianglepppp(#1,#2)|#3|/#4`#5`#6/<#7>[#8]{%
\xpos#1\ypos#2%
\def\next|##1##2##3|{\def\xa{##1}\def\xb{##2}\def\xc{##3}}%
\next|#3|%
\def\next<##1,##2>{\deltax=##1\deltay=##2\ignorespaces}%
\next<#7>%
\def\next[##1`##2`##3;##4`##5`##6]{%
    \def\nodea{##1}\def\nodeb{##2}\def\nodec{##3}%
    \def\labela{##4}\def\labelb{##5}\def\labelc{##6}}%
\next[#8]%
\advance\ypos by \deltay%
\morphism(\xpos,\ypos)|\xb|/{#5}/<\deltax,-\deltay>[\nodea`\nodec;\labelb]%
\multiply\deltax by 2%
\morphism(\xpos,\ypos)|\xa|/{#4}/<\deltax,0>[\nodea`\nodeb;\labela]%
\advance\xpos by \deltax \divide \deltax by 2
\morphism(\xpos,\ypos)|\xc|/{#6}/<-\deltax,-\deltay>[\nodeb`\nodec;\labelc]%
\ignorespaces}%

\def\Ctrianglepppp(#1,#2)|#3|/#4`#5`#6/<#7>[#8]{%
\xpos#1\ypos#2%
\def\next|##1##2##3|{\def\xa{##1}\def\xb{##2}\def\xc{##3}}%
\next|#3|%
\def\next<##1,##2>{\deltax=##1\deltay=##2\ignorespaces}%
\next<#7>%
\def\next[##1`##2`##3;##4`##5`##6]{%
    \def\nodea{##1}\def\nodeb{##2}\def\nodec{##3}%
    \def\labela{##4}\def\labelb{##5}\def\labelc{##6}}%
\next[#8]%
\advance \ypos by \deltay%
\morphism(\xpos,\ypos)|\xc|/{#6}/<\deltax,-\deltay>[\nodeb`\nodec;\labelc]%
\advance\ypos by \deltay \advance \xpos by \deltax%
\morphism(\xpos,\ypos)|\xa|/{#4}/<-\deltax,-\deltay>[\nodea`\nodeb;\labela]%
\multiply\deltay by 2%
\morphism(\xpos,\ypos)|\xb|/{#5}/<0,-\deltay>[\nodea`\nodec;\labelb]%
\ignorespaces}%

\def\Dtrianglepppp(#1,#2)|#3|/#4`#5`#6/<#7>[#8]{%
\xpos#1\ypos#2%
\def\next|##1##2##3|{\def\xa{##1}\def\xb{##2}\def\xc{##3}}%
\next|#3|%
\def\next<##1,##2>{\deltax=##1\deltay=##2\ignorespaces}%
\next<#7>%
\def\next[##1`##2`##3;##4`##5`##6]{%
    \def\nodea{##1}\def\nodeb{##2}\def\nodec{##3}%
    \def\labela{##4}\def\labelb{##5}\def\labelc{##6}}%
\next[#8]%
\advance\xpos by \deltax \advance\ypos by \deltay%
\morphism(\xpos,\ypos)|\xc|/{#6}/<-\deltax,-\deltay>[\nodeb`\nodec;\labelc]%
\advance\xpos by -\deltax \advance\ypos by \deltay%
\morphism(\xpos,\ypos)|\xb|/{#5}/<\deltax,-\deltay>[\nodea`\nodeb;\labelb]%
\multiply \deltay by 2%
\morphism(\xpos,\ypos)|\xa|/{#4}/<0,-\deltay>[\nodea`\nodec;\labela]%
\ignorespaces}%

\def\ptrianglep(#1){\ifnextchar|{\ptrianglepp(#1)}{\ptrianglepp(#1)|alr|}}%
\def\ptrianglepp(#1)|#2|{\ifnextchar/{\ptriangleppp(#1)|#2|}%
    {\ptriangleppp(#1)|#2|/>`>`>/}}%
\def\ptriangleppp(#1)|#2|/#3`#4`#5/{%
    \ifnextchar<{\ptrianglepppp(#1)|#2|/#3`#4`#5/}%
    {\ptrianglepppp(#1)|#2|/#3`#4`#5/<\default,\default>}}%

\def\qtrianglep(#1){\ifnextchar|{\qtrianglepp(#1)}{\qtrianglepp(#1)|alr|}}%
\def\qtrianglepp(#1)|#2|{\ifnextchar/{\qtriangleppp(#1)|#2|}%
    {\qtriangleppp(#1)|#2|/>`>`>/}}%
\def\qtriangleppp(#1)|#2|/#3`#4`#5/{%
    \ifnextchar<{\qtrianglepppp(#1)|#2|/#3`#4`#5/}%
    {\qtrianglepppp(#1)|#2|/#3`#4`#5/<\default,\default>}}%

\def\dtriangle{\ifnextchar({\dtrianglep}{\dtrianglep(0,0)}}%
\def\dtrianglep(#1){\ifnextchar|{\dtrianglepp(#1)}{\dtrianglepp(#1)|lrb|}}%
\def\dtrianglepp(#1)|#2|{\ifnextchar/{\dtriangleppp(#1)|#2|}%
    {\dtriangleppp(#1)|#2|/>`>`>/}}%
\def\dtriangleppp(#1)|#2|/#3`#4`#5/{%
    \ifnextchar<{\dtrianglepppp(#1)|#2|/#3`#4`#5/}%
    {\dtrianglepppp(#1)|#2|/#3`#4`#5/<\default,\default>}}%

\def\btrianglep(#1){\ifnextchar|{\btrianglepp(#1)}{\btrianglepp(#1)|lrb|}}%
\def\btrianglepp(#1)|#2|{\ifnextchar/{\btriangleppp(#1)|#2|}%
    {\btriangleppp(#1)|#2|/>`>`>/}}%
\def\btriangleppp(#1)|#2|/#3`#4`#5/{%
    \ifnextchar<{\btrianglepppp(#1)|#2|/#3`#4`#5/}%
    {\btrianglepppp(#1)|#2|/#3`#4`#5/<\default,\default>}}%

\def\Atriangle{\ifnextchar({\Atrianglep}{\Atrianglep(0,0)}}%
\def\Atrianglep(#1){\ifnextchar|{\Atrianglepp(#1)}{\Atrianglepp(#1)|lrb|}}%
\def\Atrianglepp(#1)|#2|{\ifnextchar/{\Atriangleppp(#1)|#2|}%
    {\Atriangleppp(#1)|#2|/>`>`>/}}%
\def\Atriangleppp(#1)|#2|/#3`#4`#5/{%
    \ifnextchar<{\Atrianglepppp(#1)|#2|/#3`#4`#5/}%
    {\Atrianglepppp(#1)|#2|/#3`#4`#5/<\default,\default>}}%

\def\Vtriangle{\ifnextchar({\Vtrianglep}{\Vtrianglep(0,0)}}%
\def\Vtrianglep(#1){\ifnextchar|{\Vtrianglepp(#1)}{\Vtrianglepp(#1)|alb|}}%
\def\Vtrianglepp(#1)|#2|{\ifnextchar/{\Vtriangleppp(#1)|#2|}%
    {\Vtriangleppp(#1)|#2|/>`>`>/}}%
\def\Vtriangleppp(#1)|#2|/#3`#4`#5/{%
    \ifnextchar<{\Vtrianglepppp(#1)|#2|/#3`#4`#5/}%
    {\Vtrianglepppp(#1)|#2|/#3`#4`#5/<\default,\default>}}%

\def\Ctrianglep(#1){\ifnextchar|{\Ctrianglepp(#1)}{\Ctrianglepp(#1)|arb|}}%
\def\Ctrianglepp(#1)|#2|{\ifnextchar/{\Ctriangleppp(#1)|#2|}%
    {\Ctriangleppp(#1)|#2|/>`>`>/}}%
\def\Ctriangleppp(#1)|#2|/#3`#4`#5/{%
    \ifnextchar<{\Ctrianglepppp(#1)|#2|/#3`#4`#5/}%
    {\Ctrianglepppp(#1)|#2|/#3`#4`#5/<\default,\default>}}%

\def\Dtrianglep(#1){\ifnextchar|{\Dtrianglepp(#1)}{\Dtrianglepp(#1)|alb|}}%
\def\Dtrianglepp(#1)|#2|{\ifnextchar/{\Dtriangleppp(#1)|#2|}%
    {\Dtriangleppp(#1)|#2|/>`>`>/}}%
\def\Dtriangleppp(#1)|#2|/#3`#4`#5/{%
    \ifnextchar<{\Dtrianglepppp(#1)|#2|/#3`#4`#5/}%
    {\Dtrianglepppp(#1)|#2|/#3`#4`#5/<\default,\default>}}%

\def\Atrianglepairpppp(#1)|#2|/#3`#4`#5`#6`#7/<#8>[#9]{%
\def\next(##1,##2){\xpos##1\ypos##2}%
\next(#1)%
\def\next|##1##2##3##4##5|{\def\xa{##1}\def\xb{##2}%
\def\xc{##3}\def\xd{##4}\def\xe{##5}}%
\next|#2|%
\def\next<##1,##2>{\deltax=##1\deltay=##2\ignorespaces}%
\next<#8>%
\def\next[##1`##2`##3`##4;##5`##6`##7`##8`##9]{%
 \def\nodea{##1}\def\nodeb{##2}\def\nodec{##3}\def\noded{##4}%
 \def\labela{##5}\def\labelb{##6}\def\labelc{##7}\def\labeld{##8}\def\labele{##9}}%
\next[#9]%
\morphism(\xpos,\ypos)|\xd|/{#6}/<\deltax,0>[\nodeb`\nodec;\labeld]%
\advance\xpos by \deltax%
\morphism(\xpos,\ypos)|\xe|/{#7}/<\deltax,0>[\nodec`\noded;\labele]%
\advance\ypos by \deltay%
\morphism(\xpos,\ypos)|\xa|/{#3}/<-\deltax,-\deltay>[\nodea`\nodeb;\labela]%
\morphism(\xpos,\ypos)|\xb|/{#4}/<0,-\deltay>[\nodea`\nodec;\labelb]%
\morphism(\xpos,\ypos)|\xc|/{#5}/<\deltax,-\deltay>[\nodea`\noded;\labelc]%
\ignorespaces}%

\def\Vtrianglepairpppp(#1)|#2|/#3`#4`#5`#6`#7/<#8>[#9]{%
\def\next(##1,##2){\xpos##1\ypos##2}%
\next(#1)%
\def\next|##1##2##3##4##5|{\def\xa{##1}\def\xb{##2}%
\def\xc{##3}\def\xd{##4}\def\xe{##5}}%
\next|#2|%
\def\next<##1,##2>{\deltax=##1\deltay=##2\ignorespaces}%
\next<#8>%
\def\next[##1`##2`##3`##4;##5`##6`##7`##8`##9]{%
 \def\nodea{##1}\def\nodeb{##2}\def\nodec{##3}\def\noded{##4}%
 \def\labela{##5}\def\labelb{##6}\def\labelc{##7}\def\labeld{##8}\def\labele{##9}}%
\next[#9]%
\advance\ypos by \deltay%
\morphism(\xpos,\ypos)|\xa|/{#3}/<\deltax,0>[\nodea`\nodeb;\labela]%
\morphism(\xpos,\ypos)|\xc|/{#5}/<\deltax,-\deltay>[\nodea`\noded;\labelc]%
\advance\xpos by \deltax%
\morphism(\xpos,\ypos)|\xb|/{#4}/<\deltax,0>[\nodeb`\nodec;\labelb]%
\morphism(\xpos,\ypos)|\xd|/{#6}/<0,-\deltay>[\nodeb`\noded;\labeld]%
\advance\xpos by \deltax%
\morphism(\xpos,\ypos)|\xe|/{#7}/<-\deltax,-\deltay>[\nodec`\noded;\labele]%
\ignorespaces}%

\def\Ctrianglepairpppp(#1)|#2|/#3`#4`#5`#6`#7/<#8>[#9]{%
\def\next(##1,##2){\xpos##1\ypos##2}%
\next(#1)%
\def\next|##1##2##3##4##5|{\def\xa{##1}\def\xb{##2}%
\def\xc{##3}\def\xd{##4}\def\xe{##5}}%
\next|#2|%
\def\next<##1,##2>{\deltax=##1\deltay=##2\ignorespaces}%
\next<#8>%
\def\next[##1`##2`##3`##4;##5`##6`##7`##8`##9]{%
 \def\nodea{##1}\def\nodeb{##2}\def\nodec{##3}\def\noded{##4}%
 \def\labela{##5}\def\labelb{##6}\def\labelc{##7}\def\labeld{##8}\def\labele{##9}}%
\next[#9]%
\advance\ypos by \deltay%
\morphism(\xpos,\ypos)|\xe|/{#7}/<0,-\deltay>[\nodec`\noded;\labele]%
\advance\xpos by -\deltax%
\morphism(\xpos,\ypos)|\xc|/{#5}/<\deltax,0>[\nodeb`\nodec;\labelc]%
\morphism(\xpos,\ypos)|\xd|/{#6}/<\deltax,-\deltay>[\nodeb`\noded;\labeld]%
\advance\ypos by \deltay%
\advance\xpos by \deltax%
\morphism(\xpos,\ypos)|\xa|/{#3}/<-\deltax,-\deltay>[\nodea`\nodeb;\labela]%
\morphism(\xpos,\ypos)|\xb|/{#4}/<0,-\deltay>[\nodea`\nodec;\labelb]%
\ignorespaces}%

\def\Dtrianglepairpppp(#1)|#2|/#3`#4`#5`#6`#7/<#8>[#9]{%
\def\next(##1,##2){\xpos##1\ypos##2}%
\next(#1)%
\def\next|##1##2##3##4##5|{\def\xa{##1}\def\xb{##2}%
\def\xc{##3}\def\xd{##4}\def\xe{##5}}%
\next|#2|%
\def\next<##1,##2>{\deltax=##1\deltay=##2\ignorespaces}%
\next<#8>%
\def\next[##1`##2`##3`##4;##5`##6`##7`##8`##9]{%
 \def\nodea{##1}\def\nodeb{##2}\def\nodec{##3}\def\noded{##4}%
 \def\labela{##5}\def\labelb{##6}\def\labelc{##7}\def\labeld{##8}\def\labele{##9}}%
\next[#9]%
\advance\ypos by \deltay%
\morphism(\xpos,\ypos)|\xc|/{#5}/<\deltax,0>[\nodeb`\nodec;\labelc]%
\morphism(\xpos,\ypos)|\xd|/{#6}/<0,-\deltay>[\nodeb`\noded;\labeld]%
\advance\ypos by \deltay%
\morphism(\xpos,\ypos)|\xa|/{#3}/<0,-\deltay>[\nodea`\nodeb;\labela]%
\morphism(\xpos,\ypos)|\xb|/{#4}/<\deltax,-\deltay>[\nodea`\nodec;\labelb]%
\advance\ypos by -\deltay%
\advance\xpos by \deltax%
\morphism(\xpos,\ypos)|\xe|/{#7}/<-\deltax,-\deltay>[\nodec`\noded;\labele]%
\ignorespaces}%

\def\Atrianglepairp(#1){\ifnextchar|{\Atrianglepairpp(#1)}%
{\Atrianglepairpp(#1)|lmrbb|}}%
\def\Atrianglepairpp(#1)|#2|{\ifnextchar/{\Atrianglepairppp(#1)|#2|}%
    {\Atrianglepairppp(#1)|#2|/>`>`>`>`>/}}%
\def\Atrianglepairppp(#1)|#2|/#3`#4`#5`#6`#7/{%
    \ifnextchar<{\Atrianglepairpppp(#1)|#2|/#3`#4`#5`#6`#7/}%
    {\Atrianglepairpppp(#1)|#2|/#3`#4`#5`#6`#7/<\default,\default>}}%

\def\Vtrianglepairp(#1){\ifnextchar|{\Vtrianglepairpp(#1)}%
{\Vtrianglepairpp(#1)|aalmr|}}%
\def\Vtrianglepairpp(#1)|#2|{\ifnextchar/{\Vtrianglepairppp(#1)|#2|}%
    {\Vtrianglepairppp(#1)|#2|/>`>`>`>`>/}}%
\def\Vtrianglepairppp(#1)|#2|/#3`#4`#5`#6`#7/{%
    \ifnextchar<{\Vtrianglepairpppp(#1)|#2|/#3`#4`#5`#6`#7/}%
    {\Vtrianglepairpppp(#1)|#2|/#3`#4`#5`#6`#7/<\default,\default>}}%

\def\Ctrianglepairp(#1){\ifnextchar|{\Ctrianglepairpp(#1)}%
{\Ctrianglepairpp(#1)|lrmlr|}}%
\def\Ctrianglepairpp(#1)|#2|{\ifnextchar/{\Ctrianglepairppp(#1)|#2|}%
    {\Ctrianglepairppp(#1)|#2|/>`>`>`>`>/}}%
\def\Ctrianglepairppp(#1)|#2|/#3`#4`#5`#6`#7/{%
    \ifnextchar<{\Ctrianglepairpppp(#1)|#2|/#3`#4`#5`#6`#7/}%
    {\Ctrianglepairpppp(#1)|#2|/#3`#4`#5`#6`#7/<\default,\default>}}%

\def\Dtrianglepairp(#1){\ifnextchar|{\Dtrianglepairpp(#1)}%
{\Dtrianglepairpp(#1)|lrmlr|}}%
\def\Dtrianglepairpp(#1)|#2|{\ifnextchar/{\Dtrianglepairppp(#1)|#2|}%
    {\Dtrianglepairppp(#1)|#2|/>`>`>`>`>/}}%
\def\Dtrianglepairppp(#1)|#2|/#3`#4`#5`#6`#7/{%
    \ifnextchar<{\Dtrianglepairpppp(#1)|#2|/#3`#4`#5`#6`#7/}%
    {\Dtrianglepairpppp(#1)|#2|/#3`#4`#5`#6`#7/<\default,\default>}}%

\def\pplace[#1](#2,#3)[#4]{\POS(#2,#3)*+!!<0ex,\axis>!#1{#4}\ignorespaces}%
\def\cplace(#1,#2)[#3]{\POS(#1,#2)*+!!<0ex,\axis>{#3}\ignorespaces}%
\def\place{\ifnextchar[{\pplace}{\cplace}}

\def\pullback#1]#2]{\square#1]\trident#2]\ignorespaces}%

\def\tridentppp|#1#2#3|/#4`#5`#6/<#7,#8>[#9]{%
\def\next[##1;##2`##3`##4]{\def\nodee{##1}\def\labele{##2}%
   \def\labelf{##3}\def\labelg{##4}}%
\next[#9]%
\advance \xpos by -\deltax%
\advance \xpos by -#7\advance \ypos by #8%
\advance\deltax by #7%
\morphism(\xpos,\ypos)|#1|/{#4}/<\deltax,-#8>[\nodee`\nodeb;\labele]%
\advance\deltax by -#7%
\morphism(\xpos,\ypos)|#2|/{#5}/<#7,-#8>[\nodee`\nodea;\labelf]%
\advance\deltay by #8%
\morphism(\xpos,\ypos)|#3|/{#6}/<#7,-\deltay>[\nodee`\nodec;\labelg]%
\ignorespaces}%

\def\trident{\ifnextchar|{\tridentp}{\tridentp|amb|}}%
\def\tridentp|#1|{\ifnextchar/{\tridentpp|#1|}{\tridentpp|#1|/{>}`{>}`{>}/}}%
\def\tridentpp|#1|/#2/{\ifnextchar<{\tridentppp|#1|/#2/}%
  {\tridentppp|#1|/#2/<500,500>}}%

\def\setmorphismwidth#1#2#3#4{%
 \setbox0=\hbox{$#1{\labelstyle#3#3}#2$}#4=\wd0%
 \divide #4 by 2 \divide #4 by \ul%
 \advance #4 by 350 \ratchet{#4}{500}}%

\def\setSquarewidth[#1`#2`#3`#4;#5`#6`#7`#8]{%
 \setmorphismwidth{#1}{#2}{#5}{\topw}%
 \setmorphismwidth{#3}{#4}{#8}{\botw}%
\ratchet{\topw}{\botw}}%

\def\Squarepppp(#1)|#2|/#3/<#4>[#5]{%
 \setSquarewidth[#5]%
 \squarepppp(#1)|#2|/#3/<\topw,#4>[#5]%
\ignorespaces}%

\def\Square{\ifnextchar({\Squarep}{\Squarep(0,0)}}%
\def\Squarep(#1){\ifnextchar|{\Squarepp(#1)}{\Squarepp(#1)|alrb|}}%
\def\Squarepp(#1)|#2|{\ifnextchar/{\Squareppp(#1)|#2|}%
    {\Squareppp(#1)|#2|/>`>`>`>/}}%
\def\Squareppp(#1)|#2|/#3`#4`#5`#6/{%
    \ifnextchar<{\Squarepppp(#1)|#2|/#3`#4`#5`#6/}%
    {\Squarepppp(#1)|#2|/#3`#4`#5`#6/<\default>}}%

\def\hSquarespppp(#1,#2)|#3|/#4/<#5>[#6;#7]{%
\Xpos=#1\Ypos=#2%
\def\next|##1##2##3##4##5##6##7|{%
 \def\Xa{##1}\def\Xb{##2}\def\Xc{##3}\def\Xd{##4}%
 \def\Xe{##5}\def\Xf{##6}\def\Xg{##7}}%
\next|#3|%
\deltaY=#5%
\def\next[##1`##2`##3`##4`##5`##6]{%
 \def\Nodea{##1}\def\Nodeb{##2}\def\Nodec{##3}%
 \def\Noded{##4}\def\Nodee{##5}\def\Nodef{##6}}%
\next[#6]%
\def\next[##1`##2`##3`##4`##5`##6`##7]{%
 \def\Labela{##1}\def\Labelb{##2}\def\Labelc{##3}\def\Labeld{##4}%
 \def\Labele{##5}\def\Labelf{##6}\def\Labelg{##7}}%
\next[#7]%
\dohSquares/#4/}%

\def\dohSquares/#1`#2`#3`#4`#5`#6`#7/{%
\Squarepppp(\Xpos,\Ypos)|\Xa\Xc\Xd\Xf|/#1`#3`#4`#6/<\deltaY>%
 [\Nodea`\Nodeb`\Noded`\Nodee;\Labela`\Labelc`\Labeld`\Labelf]%
 \advance \Xpos by \topw
\Squarepppp(\Xpos,\Ypos)|\Xb\Xd\Xe\Xg|/#2``#5`#7/<\deltaY>%
[\Nodeb`\Nodec`\Nodee`\Nodef;\Labelb``\Labele`\Labelg]%
\ignorespaces}%

\def\hSquares{\ifnextchar({\hSquaresp}{\hSquaresp(0,0)}}%
\def\hSquaresp(#1){\ifnextchar|{\hSquarespp(#1)}{\hSquarespp%
(#1)|aalmrbb|}}%
\def\hSquarespp(#1)|#2|{\ifnextchar/{\hSquaresppp(#1)|#2|}%
    {\hSquaresppp(#1)|#2|/>`>`>`>`>`>`>/}}%
\def\hSquaresppp(#1)|#2|/#3/{%
    \ifnextchar<{\hSquarespppp(#1)|#2|/#3/}%
    {\hSquarespppp(#1)|#2|/#3/<\default>}}%

\def\vSquarespppp(#1,#2)|#3|/#4/<#5,#6>[#7;#8]{%
\Xpos=#1\Ypos=#2%
\def\next|##1##2##3##4##5##6##7|{%
 \def\Xa{##1}\def\Xb{##2}\def\Xc{##3}\def\Xd{##4}%
 \def\Xe{##5}\def\Xf{##6}\def\Xg{##7}}%
\next|#3|%
\deltaX=#5%
\deltaY=#6%
\def\next[##1`##2`##3`##4`##5`##6]{%
 \def\Nodea{##1}\def\Nodeb{##2}\def\Nodec{##3}%
 \def\Noded{##4}\def\Nodee{##5}\def\Nodef{##6}}%
\next[#7]%
\def\next[##1`##2`##3`##4`##5`##6`##7]{%
 \def\Labela{##1}\def\Labelb{##2}\def\Labelc{##3}\def\Labeld{##4}%
 \def\Labele{##5}\def\Labelf{##6}\def\Labelg{##7}}%
\next[#8]%
\dovSquares/#4/\ignorespaces}%

\def\dovSquares/#1`#2`#3`#4`#5`#6`#7/{%
\setmorphismwidth{\Nodea}{\Nodeb}{\Labela}{\topw}%
\setmorphismwidth{\Nodec}{\Noded}{\Labeld}{\botw}%
\ratchet{\topw}{\botw}%
\setmorphismwidth{\Nodee}{\Nodef}{\Labelg}{\botw}%
\ratchet{\topw}{\botw}%
\square(\Xpos,\Ypos)|\Xd\Xe\Xf\Xg|/`#5`#6`#7/<\topw,\deltaX>%
 [\Nodec`\Noded`\Nodee`\Nodef;`\Labele`\Labelf`\Labelg]%
\advance \Ypos by \deltaX%
\square(\Xpos,\Ypos)|\Xa\Xb\Xc\Xd|/#1`#2`#3`#4/<\topw,\deltaY>%
 [\Nodea`\Nodeb`\Nodec`\Noded;\Labela`\Labelb`\Labelc`\Labeld]%
}%

\def\vSquaresp(#1){\ifnextchar|{\vSquarespp(#1)}{\vSquarespp%
(#1)|alrmlrb|}}%
\def\vSquarespp(#1)|#2|{\ifnextchar/{\vSquaresppp(#1)|#2|}%
    {\vSquaresppp(#1)|#2|/>`>`>`>`>`>`>/}}%
\def\vSquaresppp(#1)|#2|/#3/{%
    \ifnextchar<{\vSquarespppp(#1)|#2|/#3/}%
    {\vSquarespppp(#1)|#2|/#3/<\default,\default>}}%

\def\osquarepppp(#1)|#2|/#3`#4`#5`#6/<#7>[#8]{\squarepppp%
 (#1)|#2|/#3`#4`#5`#6/<#7>[#8]%
 \let\Nodea\nodea\let\Nodeb\nodeb%
\let\Nodec\nodec\let\Noded\noded\Xpos=\xpos\Ypos=\ypos%
\deltaX=\deltax \deltaY=\deltay \isquare}

\def\osquarep(#1){\ifnextchar|{\osquarepp(#1)}{\osquarepp(#1)|alrb|}}%
\def\osquarepp(#1)|#2|{\ifnextchar/{\osquareppp(#1)|#2|}%
    {\osquareppp(#1)|#2|/>`>`>`>/}}%
\def\osquareppp(#1)|#2|/#3`#4`#5`#6/{%
    \ifnextchar<{\osquarepppp(#1)|#2|/#3`#4`#5`#6/}%
    {\osquarepppp(#1)|#2|/#3`#4`#5`#6/<1500,1500>}}%

\def\isquarepppp(#1)|#2|/#3`#4`#5`#6/<#7>[#8]{%
 \squarepppp(#1)|#2|/#3`#4`#5`#6/<#7>[#8]%
\ifnextchar|{\cubep}{\cubep|mmmm|}}%
\def\cubep|#1|{\ifnextchar/{\cubepp|#1|}{\cubepp|#1|/>`>`>`>/}}%

\def\isquare{\ifnextchar({\isquarep}{\isquarep(\default,\default)}}%
\def\isquarep(#1){\ifnextchar|{\isquarepp(#1)}{\isquarepp(#1)|alrb|}}
\def\isquarepp(#1)|#2|{\ifnextchar/{\isquareppp(#1)|#2|}%
    {\isquareppp(#1)|#2|/>`>`>`>/}}%
\def\isquareppp(#1)|#2|/#3`#4`#5`#6/{%
    \ifnextchar<{\isquarepppp(#1)|#2|/#3`#4`#5`#6/}%
    {\isquarepppp(#1)|#2|/#3`#4`#5`#6/<500,500>}}%

\def\cubepp|#1#2#3#4|/#5`#6`#7`#8/[#9]{%
\def\next[##1`##2`##3`##4]{\gdef\Labela{##1}%
\gdef\Labelb{##2}\gdef\Labelc{##3}\gdef\Labeld{##4}}\next[#9]%
\xend\xpos \yend\ypos
\Xend\xend\advance\Xend by -\Xpos
\Yend\yend\advance\Yend by -\Ypos
\domorphism(\Xpos,\Ypos)|#2|/#6/<\Xend,\Yend>[\Nodeb`\nodeb;\Labelb]%
\advance\Xpos by-\deltaX
\advance\xend by-\deltax
\Xend\xend\advance\Xend by -\Xpos
\domorphism(\Xpos,\Ypos)|#1|/#5/<\Xend,\Yend>[\Nodea`\nodea;\Labela]%
\advance\Ypos by-\deltaY
\advance\yend by-\deltay
\Yend\yend\advance\Yend by -\Ypos
\domorphism(\Xpos,\Ypos)|#3|/#7/<\Xend,\Yend>[\Nodec`\nodec;\Labelc]%
\advance\Xpos by\deltaX
\advance\xend by\deltax
\Xend\xend\advance\Xend by -\Xpos
\domorphism(\Xpos,\Ypos)|#4|/#8/<\Xend,\Yend>[\Noded`\noded;\Labeld]%
\ignorespaces}

\def\setwdth#1#2{\setbox0\hbox{$\labelstyle#1$}\wdth=\wd0
\setbox0\hbox{$\labelstyle#2$}\ifnum\wdth<\wd0 \wdth=\wd0 \fi}

\def\topppp/#1/<#2>^#3_#4{\allowbreak\mathrel{%
\ifnum#2=0
   \setwdth{#3}{#4}\deltax=\wdth \divide \deltax by \ul
   \advance \deltax by \defaultmargin  \ratchet{\deltax}{200}%
\else \deltax #2
\fi
\xy\ar@{#1}^{#3}_{#4}(\deltax,0) \endxy
\ignorespaces}}

\def\toppp/#1/<#2>^#3{\ifnextchar_{\topppp/#1/<#2>^{#3}}{\topppp/#1/<#2>^{#3}_{}}}
\def\topp/#1/<#2>{\ifnextchar^{\toppp/#1/<#2>}{\toppp/#1/<#2>^{}}}
\def\toop/#1/{\ifnextchar<{\topp/#1/}{\topp/#1/<0>}}
\def\to{\ifnextchar/{\toop}{\toop/>/}}

\def\twopppp/#1`#2/<#3>^#4_#5{\allowbreak\mathrel{%
\ifnum0=#3
  \setwdth{#4}{#5}\deltax=\wdth \divide \deltax by \ul \advance \deltax
  by \defaultmargin \ratchet{\deltax}{200}%
\else \deltax#3 \fi
\xy\ar@{#1}@<2.5pt>^{#4}(\deltax,0)%
\ar@{#2}@<-2.5pt>_{#5}(\deltax,0)\endxy\ignorespaces}}

\def\twoppp/#1`#2/<#3>^#4{\ifnextchar_{\twopppp/#1`#2/<#3>^{#4}}%
  {\twopppp/#1`#2/<#3>^{#4}_{}}}
\def\twopp/#1`#2/<#3>{\ifnextchar^{\twoppp/#1`#2/<#3>}{\twoppp/#1`#2/<#3>^{}}}
\def\twop/#1`#2/{\ifnextchar<{\twopp/#1`#2/}{\twopp/#1`#2/<0>}}

\def\threeppppp/#1`#2`#3/<#4>^#5|#6_#7{\allowbreak\mathrel{%
\ifnum0=#4
\setbox0\hbox{$\labelstyle#5$}\wdth=\wd0
\setbox0\hbox{$\labelstyle#6$}\ifnum\wdth<\wd0 \wdth=\wd0 \fi
\setbox0\hbox{$\labelstyle#7$}\ifnum\wdth<\wd0 \wdth=\wd0 \fi
\deltax=\wdth \divide \deltax by \ul \advance \deltax by
\defaultmargin \ratchet{\deltax}{300}%
\else\deltax#4 \fi
    \xy \ifnum\wd0=0 \ar@{#2}(\deltax,0)
    \else \ar@{#2}|{#6}(\deltax,0)\fi
\ar@{#1}@<4.5pt>^{#5}(\deltax,0)
\ar@{#3}@<-4.5pt>_{#7}(\deltax,0)\endxy\ignorespaces}}

\def\threepppp/#1`#2`#3/<#4>^#5|#6{\ifnextchar_{\threeppppp
  /#1`#2`#3/<#4>^{#5}|{#6}}{\threeppppp/#1`#2`#3/<#4>^{#5}|{#6}_{}}}
\def\threeppp/#1`#2`#3/<#4>^#5{\ifnextchar|{\threepppp
  /#1`#2`#3/<#4>^{#5}}{\threepppp/#1`#2`#3/<#4>^{#5}|{}}}
\def\threepp/#1`#2`#3/<#4>{\ifnextchar^{\threeppp/#1`#2`#3/<#4>}%
  {\threeppp/#1`#2`#3/<#4>^{}}}
\def\threep/#1`#2`#3/{\ifnextchar<{\threepp/#1`#2`#3/}%
  {\threepp/#1`#2`#3/<0>}}

\def\twoar(#1,#2){{%
 \scalefactor{0.1}
 \deltax#1\deltay#2%
 \deltaX=\ifnum\deltax<0-\fi\deltax
 \deltaY=\ifnum\deltay<0-\fi\deltay
 \Xend\deltax \multiply \Xend by \deltax
 \Yend\deltay \multiply \Yend by \deltay
 \advance\Xend by \Yend \multiply \Xend by 3
 \ifnum \deltaX > \deltaY
    \multiply \deltaX by 3 \advance \deltaX by \deltaY
 \else
    \multiply \deltaY by 3 \advance \deltaX by \deltaY
 \fi
 \multiply\deltax by 500
 \multiply\deltay by 500
 \xpos\deltax \multiply \xpos by 3 \divide\xpos by \deltaX
 \Xpos\deltax \multiply \Xpos by \deltaX \divide \Xpos by \Xend
 \advance \xpos by \Xpos
 \ypos\deltay \multiply \ypos by 3 \divide\ypos by \deltaX
 \Ypos\deltay \multiply \Ypos by \deltaX \divide \Ypos by \Xend
 \advance \ypos by \Ypos
 \xy \ar@{=>}(\xpos,\ypos) \endxy
}\ignorespaces}

\def\iiixiiipppppp(#1,#2)|#3|/#4/<#5>#6<#7>[#8;#9]{%
 \xpos#1\ypos#2\relax
 \def\next|##1##2##3##4##5##6##7|{\def\xa{##1}\def\xb{##2}%
 \def\xc{##3}\def\xd{##4}\def\xe{##5}\def\xf{##6}\nextt|##7|}%
 \def\nextt|##1##2##3##4##5##6|{\def\xg{##1}\def\xh{##2}%
 \def\xi{##3}\def\xj{##4}\def\xk{##5}\def\xl{##6}}%
 \next|#3|%
 \def\next<##1,##2>{\deltax##1\deltay##2}%
 \next<#5>%
 \def\next<##1,##2>{\deltaX##1\deltaY##2}%
 \next<#7>%
 \def\next##1{\topw##1\relax
 \ifodd\topw \def\zl{}\else\def\zl{\relax}\fi \divide\topw by 2
 \ifodd\topw \def\zk{}\else\def\zk{\relax}\fi \divide\topw by 2
 \ifodd\topw \def\zj{}\else\def\zj{\relax}\fi \divide\topw by 2
 \ifodd\topw \def\zi{}\else\def\zi{\relax}\fi \divide\topw by 2
 \ifodd\topw \def\zh{}\else\def\zh{\relax}\fi \divide\topw by 2
 \ifodd\topw \def\zg{}\else\def\zg{\relax}\fi \divide\topw by 2
 \ifodd\topw \def\zf{}\else\def\zf{\relax}\fi \divide\topw by 2
 \ifodd\topw \def\ze{}\else\def\ze{\relax}\fi \divide\topw by 2
 \ifodd\topw \def\zd{}\else\def\zd{\relax}\fi \divide\topw by 2
 \ifodd\topw \def\zc{}\else\def\zc{\relax}\fi \divide\topw by 2
 \ifodd\topw \def\zb{}\else\def\zb{\relax}\fi \divide\topw by 2
 \ifodd\topw \def\za{}\else\def\za{\relax}\fi}%
 \next{#6}%
 \def\next[##1`##2`##3`##4`##5`##6`##7`##8`##9]{%
 \def\nodea{##1}\def\nodeb{##2}\def\nodec{##3}%
 \def\noded{##4}\def\nodee{##5}\def\nodef{##6}%
 \def\nodeg{##7}\def\nodeh{##8}\def\nodei{##9}}%
 \next[#8]%
 \def\next[##1`##2`##3`##4`##5`##6`##7]{%
 \def\labela{##1}\def\labelb{##2}\def\labelc{##3}%
 \def\labeld{##4}\def\labele{##5}\def\labelf{##6}\nextt[##7]}%
 \def\nextt[##1`##2`##3`##4`##5`##6]{%
 \def\labelg{##1}\def\labelh{##2}\def\labeli{##3}%
 \def\labelj{##4}\def\labelk{##5}\def\labell{##6}}%
 \next[#9]%
 \def\next/##1`##2`##3`##4`##5`##6`##7/{%
\morphism(\xpos,\ypos)|\xe|/{##5}/<\deltax,0>[\nodeg`\nodeh;\labele]%
 \ifx\zi\empty\relax \morphism(\xpos,\ypos)||/<-/<-\deltaX,0>[\nodeg`0;]\fi
 \ifx\zd\empty\relax \morphism(\xpos,\ypos)||<0,-\deltaY>[\nodeg`0;]\fi
 \advance\xpos by \deltax
 \morphism(\xpos,\ypos)|\xf|/{##6}/<\deltax,0>[\nodeh`\nodei;\labelf]%
 \ifx\ze\empty\relax \morphism(\xpos,\ypos)||<0,-\deltaY>[\nodeh`0;]\fi
 \advance\xpos by \deltax
 \ifx\zf\empty\relax \morphism(\xpos,\ypos)||<0,-\deltaY>[\nodei`0;]\fi
 \ifx\zl\empty\relax \morphism(\xpos,\ypos)||<\deltaX,0>[\nodei`0;]\fi
 \advance\ypos by \deltay
 \ifx\zk\empty\relax \morphism(\xpos,\ypos)||<\deltaX,0>[\nodef`0;]\fi
 \advance\xpos by -\deltax
 \morphism(\xpos,\ypos)|\xd|/{##4}/<\deltax,0>[\nodee`\nodef;\labeld]%
 \advance\xpos by -\deltax
 \morphism(\xpos,\ypos)|\xc|/{##3}/<\deltax,0>[\noded`\nodee;\labelc]%
 \ifx\zh\empty\relax \morphism(\xpos,\ypos)||/<-/<-\deltaX,0>[\noded`0;]\fi
 \advance\ypos by \deltay
 \morphism(\xpos,\ypos)|\xa|/{##1}/<\deltax,0>[\nodea`\nodeb;\labela]%
 \ifx\zg\empty\relax \morphism(\xpos,\ypos)||/<-/<-\deltaX,0>[\nodea`0;]\fi
 \ifx\za\empty\relax \morphism(\xpos,\ypos)||/<-/<0,\deltaY>[\nodea`0;]\fi
 \advance\xpos by \deltax
 \morphism(\xpos,\ypos)|\xb|/{##2}/<\deltax,0>[\nodeb`\nodec;\labelb]%
 \ifx\zb\empty\relax \morphism(\xpos,\ypos)||/<-/<0,\deltaY>[\nodeb`0;]\fi
 \advance\xpos by \deltax
 \ifx\zc\empty\relax \morphism(\xpos,\ypos)||/<-/<0,\deltaY>[\nodec`0;]\fi
 \ifx\zj\empty\relax \morphism(\xpos,\ypos)||<\deltaX,0>[\nodec`0;]\fi
 \nextt/##7/}%
 \def\nextt/##1`##2`##3`##4`##5`##6/{%
 \morphism(\xpos,\ypos)|\xi|/{##3}/<0,-\deltay>[\nodec`\nodef;\labeli]%
 \advance\xpos by -\deltax
 \morphism(\xpos,\ypos)|\xh|/{##2}/<0,-\deltay>[\nodeb`\nodee;\labelh]%
 \advance\xpos by -\deltax
 \morphism(\xpos,\ypos)|\xg|/{##1}/<0,-\deltay>[\nodea`\noded;\labelg]%
 \advance\ypos by -\deltay
 \morphism(\xpos,\ypos)|\xj|/{##4}/<0,-\deltay>[\noded`\nodeg;\labelj]%
 \advance\xpos by \deltax
 \morphism(\xpos,\ypos)|\xk|/{##5}/<0,-\deltay>[\nodee`\nodeh;\labelk]%
 \advance\xpos by \deltax
 \morphism(\xpos,\ypos)|\xl|/{##6}/<0,-\deltay>[\nodef`\nodei;\labell]}%
 \next/#4/\ignorespaces}

\def\iiixiiip(#1){\ifnextchar|{\iiixiiipp(#1)}%
  {\iiixiiipp(#1)|aammbblmrlmr|}}%
\def\iiixiiipp(#1)|#2|{\ifnextchar/{\iiixiiippp(#1)|#2|}%
    {\iiixiiippp(#1)|#2|/>`>`>`>`>`>`>`>`>`>`>`>/}}%
\def\iiixiiippp(#1)|#2|/#3/{%
    \ifnextchar<{\iiixiiipppp(#1)|#2|/#3/}%
    {\iiixiiipppp(#1)|#2|/#3/<\default,\default>}}%
\def\iiixiiipppp(#1)|#2|/#3/<#4>{\ifnextchar[{\iiixiiippppp(#1)|#2|/#3/%
   <#4>0<0,0>}{\iiixiiippppp(#1)|#2|/#3/<#4>}}%
\def\iiixiiippppp(#1)|#2|/#3/<#4>#5{\ifnextchar<%
   {\iiixiiipppppp(#1)|#2|/#3/<#4>{#5}}%
   {\iiixiiipppppp(#1)|#2|/#3/<#4>{#5}<400,400>}}%

\def\iiixiipppppp(#1,#2)|#3|/#4/<#5>#6<#7>[#8;#9]{%
 \xpos#1\ypos#2\relax
 \def\next|##1##2##3##4##5##6##7|{\def\xa{##1}\def\xb{##2}%
 \def\xc{##3}\def\xd{##4}\def\xe{##5}\def\xf{##6}\def\xg{##7}}%
 \next|#3|%
 \def\next<##1,##2>{\deltax##1\deltay##2}%
 \next<#5>%
 \deltaX#7
 \topw#6
 \def\next{%
 \ifodd\topw \def\za{}\else\def\za{\relax}\fi \divide\topw by 2
 \ifodd\topw \def\zb{}\else\def\zb{\relax}\fi \divide\topw by 2
 \ifodd\topw \def\zc{}\else\def\zc{\relax}\fi \divide\topw by 2
 \ifodd\topw \def\zd{}\else\def\zd{\relax}\fi}%
 \next
 \def\next[##1`##2`##3`##4`##5`##6]{%
 \def\nodea{##1}\def\nodeb{##2}\def\nodec{##3}%
 \def\noded{##4}\def\nodee{##5}\def\nodef{##6}}%
 \next[#8]%
 \def\next[##1`##2`##3`##4`##5`##6`##7]{%
 \def\labela{##1}\def\labelb{##2}\def\labelc{##3}%
 \def\labeld{##4}\def\labele{##5}\def\labelf{##6}\def\labelg{##7}}%
 \next[#9]%
 \def\next/##1`##2`##3`##4`##5`##6`##7/{%
 \ifx\zc\empty\relax\morphism(\xpos,\ypos)<\deltaX,0>[0`\noded;]\fi
 \advance\xpos by\deltaX
 \morphism(\xpos,\ypos)|\xc|/##3/<\deltax,0>[\noded`\nodee;\labelc]%
 \advance\xpos by \deltax
 \morphism(\xpos,\ypos)|\xd|/##4/<\deltax,0>[\nodee`\nodef;\labeld]%
 \advance\xpos by \deltax
 \ifx\zd\empty\relax  \morphism(\xpos,\ypos)<\deltaX,0>[\nodef`0;]\fi
 \advance\xpos by -\deltaX  \advance\xpos by -\deltax
 \advance\xpos by -\deltax  \advance\ypos by \deltay
 \ifx\za\empty\relax\morphism(\xpos,\ypos)<\deltaX,0>[0`\nodea;]\fi
 \advance\xpos by\deltaX
 \morphism(\xpos,\ypos)|\xa|/##1/<\deltax,0>[\nodea`\nodeb;\labela]%
 \morphism(\xpos,\ypos)|\xe|/##5/<0,-\deltay>[\nodea`\noded;\labele]%
 \advance\xpos by \deltax
 \morphism(\xpos,\ypos)|\xb|/##2/<\deltax,0>[\nodeb`\nodec;\labelb]%
 \morphism(\xpos,\ypos)|\xf|/##6/<0,-\deltay>[\nodeb`\nodee;\labelf]%
 \advance\xpos by \deltax
 \morphism(\xpos,\ypos)|\xg|/##7/<0,-\deltay>[\nodec`\nodef;\labelg]%
 \ifx\zb\empty\relax \morphism(\xpos,\ypos)<\deltaX,0>[\nodec`0;]\fi}%
 \next/#4/\ignorespaces}

\def\iiixiip(#1){\ifnextchar|{\iiixiipp(#1)}%
  {\iiixiipp(#1)|aabblmr|}}%
\def\iiixiipp(#1)|#2|{\ifnextchar/{\iiixiippp(#1)|#2|}%
    {\iiixiippp(#1)|#2|/>`>`>`>`>`>`>/}}%
\def\iiixiippp(#1)|#2|/#3/{%
    \ifnextchar<{\iiixiipppp(#1)|#2|/#3/}%
    {\iiixiipppp(#1)|#2|/#3/<\default,\default>}}%
\def\iiixiipppp(#1)|#2|/#3/<#4>{\ifnextchar[{\iiixiippppp(#1)|#2|/#3/%
   <#4>{0}<0>}{\iiixiippppp(#1)|#2|/#3/<#4>}}%
\def\iiixiippppp(#1)|#2|/#3/<#4>#5{\ifnextchar<%
   {\iiixiipppppp(#1)|#2|/#3/<#4>{#5}}%
   {\iiixiipppppp(#1)|#2|/#3/<#4>{#5}<0>}}%

\def\node#1(#2,#3)[#4]{%
\expandafter\gdef\csname x@#1\endcsname{#2}%
\expandafter\gdef\csname y@#1\endcsname{#3}%
\expandafter\gdef\csname ob@#1\endcsname{#4}%
}

\newcount\xfinish
\newcount\yfinish

\def\arrowp|#1|{\ifnextchar/{\arrowpp|#1|}{\arrowpp|#1|/>/}}
\def\arrowpp|#1|/#2/[#3`#4;#5]{%
\xfinish=\csname x@#4\endcsname
\yfinish=\csname y@#4\endcsname
\advance\xfinish by -\csname x@#3\endcsname
\advance\yfinish by -\csname y@#3\endcsname
\morphism(\csname x@#3\endcsname,\csname y@#3\endcsname)|#1|/#2/%
<\xfinish,\yfinish>[\csname ob@#3\endcsname`\csname ob@#4\endcsname;#5]%
}

\def\Loop(#1,#2)#3(#4,#5){\POS(#1,#2)*+!!<0ex,\axis>{#3}\ar@(#4,#5)}
\def\iloop#1(#2,#3){\xy\Loop(0,0)#1(#2,#3)\endxy}

\catcode`\@=\atcode%
 
\entrymodifiers={+!!<0pt,\fontdimen22\textfont2>}